\documentclass[10pt]{article}

\usepackage{geometry}

\usepackage[utf8x]{inputenc}
\usepackage[T1]{fontenc}
\usepackage[english]{babel}              
\usepackage{graphicx}
\usepackage{amssymb,amsthm,amsfonts,amsmath}
\usepackage{mathtools}
\mathtoolsset{showonlyrefs=true,showmanualtags=false} 
\usepackage{hyperref}
\hypersetup{
    colorlinks=true,
    linkcolor=black,          
    citecolor=black,        
    filecolor=black,      
    urlcolor=black  
}
\usepackage{url}
\usepackage{subfigure}
\usepackage{enumerate}
\usepackage[normalem]{ulem}

\usepackage{enumitem}  

\usepackage{color}

\usepackage[dvipsnames]{xcolor}

\usepackage{rotating} 
\usepackage[section]{algorithm}
\usepackage{algorithmic}

\usepackage{xspace}

\usepackage{multirow} 
\usepackage{booktabs} 

\theoremstyle{plain}

\newtheorem{thm}{Theorem}[section]
\newtheorem{cor}[thm]{Corollary}
\newtheorem{lem}[thm]{Lemma}
\newtheorem{prop}[thm]{Proposition}

\theoremstyle{definition}
\newtheorem{defn}[thm]{Definition}

\theoremstyle{remark}
\newtheorem{rmrk}[thm]{Remark}

\DeclareMathOperator{\Prob}{\mathbb{P}}
\DeclareMathOperator{\Expe}{\mathbb{E}}

\newcommand{\Amat}{\mathbf{A}}

\newcommand{\Emat}{\mathbf{E}}

\newcommand{\Gmat}{\mathbf{G}}

\newcommand{\Imat}{\mathbf{I}}
\newcommand{\Kmat}{\mathbf{K}}
\newcommand{\Mmat}{\mathbf{M}}

\newcommand{\Pmat}{\mathbf{P}}

\newcommand{\Umat}{\mathbf{U}}

\newcommand{\avec}{\mathbf{a}}

\newcommand{\evec}{\mathbf{e}}
\newcommand{\fvec}{\mathbf{f}}
\newcommand{\gvec}{\mathbf{g}}

\newcommand{\mvec}{\mathbf{m}}

\newcommand{\wvec}{\mathbf{w}}
\newcommand{\xvec}{\mathbf{x}}
\newcommand{\yvec}{\mathbf{y}}
\newcommand{\zvec}{\mathbf{z}}
\newcommand{\Ovec}{\mathbf{0}}

\newcommand{\RR}{\mathbb{R}}
\newcommand{\CC}{\mathbb{C}}
\newcommand{\NN}{\mathbb{N}}

\newcommand{\de}{\,\text{d}}

\DefineNamedColor{named}{JungleGreen} {cmyk}{0.99,0,0.52,0}

\title{Robustness to unknown error in sparse regularization}
\author{Simone Brugiapaglia\footnote{Simon Fraser University, Burnaby, BC, Canada. e-mail: simone\_brugiapaglia@sfu.ca } {} and
   Ben Adcock\footnote{Simon Fraser University, Burnaby, BC, Canada. e-mail: ben\_adcock@sfu.ca}}

\date{\today\\\vspace{0.4cm} \small (To appear in IEEE Transactions on Information Theory)}

\begin{document}

\maketitle


\abstract{

Quadratically-constrained basis pursuit has become a popular device in sparse regularization; in particular, in the context of compressed sensing.  However, the majority of theoretical error estimates for this regularizer  assume an \textit{a priori} bound on the noise level, which is usually lacking in practice.  In this paper, we develop stability and robustness estimates which remove this assumption.  First, we introduce an abstract framework and show that robust instance optimality of any decoder in the noise-aware setting implies stability and robustness in the noise-blind setting.  This is based on certain sup-inf constants referred to as quotients, strictly related to the quotient property of compressed sensing.  We then apply this theory to prove the robustness of quadratically-constrained basis pursuit under unknown error in the cases of random Gaussian matrices and of random matrices with heavy-tailed rows, such as random sampling matrices from bounded orthonormal systems.  We illustrate our results in several cases of practical importance, including subsampled Fourier measurements and recovery of sparse polynomial expansions.

}


\section{Introduction}
\label{sec:intro}

The purpose of this paper is to address the recovery error analysis of the \emph{Quadratically-Constrained Basis Pursuit} (QCBP) optimization program in the presence of unknown sources of error corrupting the measurements.
%
%
The QCBP program is defined as follows:
\begin{equation}
\label{eq:QCBPdef1}
\min_{\zvec \in \CC^N} \|\zvec\|_1, \; \text{s.t. } \|\Amat \zvec - \yvec\|_2 \leq \eta.
\end{equation}
This optimization problem has a long history \cite{Donoho1992,Logan1965}, but here we are particularly interested in its application in \emph{Compressed Sensing} (CS) \cite{Candes2006,Donoho2006,Foucart2013}. We call $\Amat\in\CC^{m \times N}$ the sensing matrix, $\yvec\in\CC^m$ the vector of measurements, and  $\eta > 0$ the threshold parameter. When $\eta =0$, \eqref{eq:QCBPdef1} is also referred to as \emph{Basis Pursuit} (BP). We assume to have measurements
\begin{equation}
\yvec = \Amat\xvec + \evec,
\end{equation}
where $\xvec \in \CC^N$ is the target solution that we aim to recover and $\evec\in \CC^m$ is an unknown error term corrupting the measurements. Moreover, in the CS framework, we are interested in the case where $m \ll N$. 

While there is a rich literature on recovery error estimates for QCBP in the CS literature, almost all results (with a few exceptions, see Section~\ref{sec:literature}) are 
based on an \emph{a priori} estimate of the error term of the form
\begin{equation}
\label{eq:noisebound}
\|\evec\|_2 \leq \eta.
\end{equation}
However, in many, if not all applications of CS, a bound of the form \eqref{eq:noisebound} is unlikely to be known, due to the various ways in which errors arise in practice.  Several common sources of such errors are as follows:
\begin{itemize}
\item \emph{Physical noise}: In signal acquisition (magnetic resonance imaging, tomography, radar, etc.) there is an intrinsic source of error due to physical noise that corrupts the measurements of any sensing device. At best, one may have a reasonable statistical model for this noise, but a bound \eqref{eq:noisebound} is usually unobtainable.

\item \emph{Model error}: In applications such as MRI and tomography there is always a discrepancy between the physics of the sensing device and the formulation of the acquisition process as a linear model $\yvec = \Amat \xvec + \evec$.  The Fourier and Radon transforms -- the standard operators used to describe these sensing devices -- are, in the end, only approximations of the `true' physical model, and in certain situations may lead to significant model errors.  Moreover, it is commonplace to replace the continuous Fourier and Radon transforms by their discrete analogues, leading to an additional source of error.  Since such errors are dependent on the unknown signal $\xvec$, a bound \eqref{eq:noisebound} typically fails to hold.  Similar model errors may arise due to nonlinearity of the acquisition process, or quantization of the measurements.

\item \emph{Truncation error}: In function approximation and interpolation based on CS  \cite{Adcock2017c,Adcock2017b,Rauhut2016}, the choice of the finite-dimensional approximation space introduces a truncation error in addition to any noise in the measurements. This usually depends on the regularity of the function to approximate, which is by definition unknown.

\item \emph{Numerical error}: In recent applications of CS to the numerical treatment of PDEs, there are at least two further sources of errors. In uncertainty quantification, where the measurements correspond to pointwise samples of the solution manifold of a parametric PDE, each sample is based on a black-box PDE solver whose accuracy may not be known in practice \cite{Bouchot2015,Doostan2011,Rauhut2017,Yang2013}. Another example is given by the \textsf{CORSING} approach for deterministic PDEs, where the measurements and entries of the sensing matrix $\Amat$ are computed using numerical quadrature routines \cite{PhDThesis,Brugiapaglia2015,Brugiapaglia2016}. In both instances, bounds of the form \eqref{eq:noisebound} are usually unavailable.
\end{itemize}

Despite this issue, the QCBP program is used frequently in applications of CS.  Indeed, it is not difficult to observe numerically that QCBP is often quite robust even when \eqref{eq:noisebound} does not hold. In Figure~\ref{fig:m_vs_err} we plot the absolute recovery error achieved by BP ($\eta = 0$) and by QCBP with $\eta = 10^{-3}$ for  $\|\evec\|_2 = 10^{-3}$ as a function of $m$ when recovering an exactly sparse solution. We repeat the same experiment for three different types of measurements: random Gaussian measurements, partial discrete Fourier transform, and nonharmonic Fourier measurements. In all cases, the QCBP shows to be quite robust even when \eqref{eq:noisebound} is not satisfied, especially (at least for the first and third types of measurements) in the case where $m \ll N$.  Interestingly, for discrete Fourier measurements there is no reduction in robustness as $m$ approaches $N$.
\begin{figure}
\centering
\includegraphics[width = 9 cm]{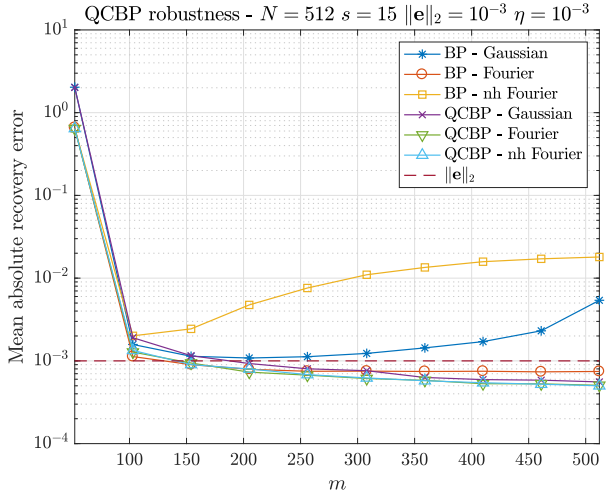}
\caption{\label{fig:m_vs_err}Robustness of QCBP to unknown error for different sensing matrices. We plot the recovery error as a function of $m$ on exactly sparse solutions, averaged over 25 runs.  The error on the measurement is of magnitude $\|\evec\|_2 = 10^{-3}$. We compare the performance of BP ($\eta = 0$) and of QCBP with threshold $\eta = 10^{-3}$ for random Gaussian measurements, partial discrete Fourier transform, and nonharmonic Fourier measurements (see \eqref{eq:defnhFourier}). The recovery is quite robust even when \eqref{eq:noisebound} does not hold. More details about this experiment are given in Section~\ref{sec:Fourier}.}
\end{figure}
 
This figure is illustrative of the recovery properties of QCBP in practice.  However, outside of the Gaussian case there are few theoretical results in the literature that explain this behavior, and correspondingly, the success of QCBP in applications of CS.  The main aim of this paper is to reduce this gap between theory and practice.  We do so by analyzing the robustness of QCBP without assuming any prior knowledge on the error $\evec$ corrupting the measurements. In particular, our main estimates provide a first theoretical explanation of some of the observed empirical behavior of QCBP shown in Figure \ref{fig:m_vs_err}.

\subsection{Problem setting} 

In order to formalize the problem  more precisely, let us fix the notation.

\paragraph{Notation} For every $N \in \NN$, we define $[N]:=\{1,\ldots,N\}$. Given a subset $S \subseteq [N]$, the vector $\zvec_S$ is the restriction of $\zvec$ to the components in $S$. The complement set of $S \subseteq [N]$ is denoted as $\overline{S} = [N] \setminus S$. The bar notation is also used to indicate the conjugate of a complex number, but the difference will be clear from the context. Given a vector $\zvec \in \CC^N$, $\zvec^*$ denotes its conjugate transpose (the same holds for a matrix $\Amat$ and its Hermitian transpose $\Amat^*$) and $\langle \cdot,\cdot\rangle$ is the standard inner product of $\CC^N$. Given a set $X$, we denote its power set as $\mathcal{P}(X)$. We denote the $\ell^0$ norm of $\zvec$ as $\|\zvec\|_0 = |\{j\in[N]: z_j \neq 0\}|$ and define the set of $s$-sparse vectors of $\CC^N$ as $\Sigma^N_s:=\{\zvec \in \CC^N : \|\zvec\|_0 \leq s\}$.   Moreover, we denote the best $s$-term approximation error of $\zvec$ with respect to the $\ell^q$ norm as 
$$
\sigma_s(\zvec)_q := \min_{\wvec \in \Sigma^N_s}\|\zvec-\wvec\|_q .
$$
For any matrix $\Amat \in \CC^{n \times k}$, $\|\Amat\|_2$ denotes the matrix norm induced by the $\ell^2$ norm. The letters $c,d$, $C,D,E$  will be reserved to indicate universal constants. They do not depend on any other parameter, unless otherwise stated. The notation $X \lesssim Y$ hides the presence of a constant $C$ independent of $X$ and $Y$ such that $X \leq C Y$. Moreover, we use the notation $X \sim Y$ when $X \lesssim Y$ and $X \gtrsim Y$ hold simultaneously.\\

In our context, it is convenient to formalize the measurement-recovery process in terms of encoder-decoder pairs. 

\begin{defn}[Encoder-decoder pair]
We call \emph{encoder} a matrix $\Amat \in \CC^{m\times N}$ that maps (or, equivalently, encodes) vectors from $\CC^N$ onto $\CC^m$ as $\xvec \mapsto \Amat \xvec$, i.e., that computes the measurements of the signal $\xvec$. A \emph{decoder} is a function $\Delta:\CC^m\to \mathcal{P}(\CC^N)$, where $\mathcal{P}(\CC^N)$ denotes the power set of $\CC^N$. The function $\Delta$ recovers (or, equivalently, decodes) an encoded vector of $\CC^m$ to a set of vectors in $\CC^N$.  Any pair $(\Amat, \Delta)$ will be referred to as an \emph{encoder-decoder} pair.\footnote{The codomain of the decoder $\Delta$ is not just $\CC^N$ since uniqueness of the decoded solution is not required and it is not a particular concern in our framework.}
\end{defn}

The QCBP decoder studied in this paper corresponds to \eqref{eq:QCBPdef1} and it is defined as follows:

\begin{defn}[Quadratically-constrained basis pursuit decoder] 
Given $\eta \geq 0$, the function $\Delta_\eta : \CC^m \to \mathcal{P}(\CC^N)$ defined as 
\begin{equation}
\label{eq:defQCBPdecoder}
\Delta_\eta : \yvec \longmapsto \widehat{\xvec}(\eta) = \arg\min_{\zvec \in \CC^N} \|\zvec\|_1, \; \text{s.t. } \|\Amat\zvec-\yvec\|\leq \eta,
\end{equation}
is said to be the \emph{Quadratically-Constrained Basis Pursuit} (QCBP) decoder with threshold $\eta$ relative to the norm $\|\cdot\|$ on $\CC^m$. In particular, $\Delta_0$ is  called the \emph{Basis Pursuit} (BP) decoder. When the  set $\{\zvec \in \CC^N : \|\Amat \zvec - \yvec\| \leq \eta\}$ of feasible vectors is empty, we define $\widehat{\xvec}(\eta) = \{\Ovec\}$.
\end{defn}

Usually, assuming $\|\evec\|_2 \leq \eta$, and for $\Amat$ fulfilling suitable hypotheses based on the restricted isometry property or on similar conditions, the recovery error estimates for the QCBP decoder take the form (see, for example, \cite{Foucart2013})
\begin{equation}
\label{eq:usualCSerrorbound}
\|\xvec-\Delta_\eta(\Amat\xvec +\evec)\|_2 \lesssim \frac{\sigma_s(\xvec)}{\sqrt{s}} + \eta.
\end{equation}
In this paper, we  provide recovery error estimates of the form
\begin{equation}
\label{eq:newestimate}
\|\xvec-\Delta_\eta(\Amat\xvec +\evec)\|_2 
\lesssim \frac{\sigma_s(\xvec)}{\sqrt{s}} + \eta + \max\{\|\evec\|_2-\eta,0\},
\end{equation}
where the hypothesis $\|\evec\|_2 \leq \eta$ is removed. For the sake of clarity, let us point out that, for any norm $\|\cdot\|$ on $\CC^N$, we define
$$
\|\xvec -  \Delta_\eta(\yvec)\| = \sup_{\zvec \in \Delta_\eta(\yvec)} \|\xvec - \zvec\|,
$$
since $\Delta_\eta(\yvec)$ is a subset of $\CC^N$. For a more detailed discussion about the nature of the minimizing set $\Delta_\eta(\yvec)$, see \cite{Unser2016}.

In the next section, we summarize the main contributions of the paper, describing what are the assumptions required to prove the validity of inequalities of the form \eqref{eq:newestimate}.

\subsection{Main contributions}
\label{sec:contributions}

We present the main results of the paper. In order to simplify the exposition, we assume the recovery error to be measured with respect to the $\ell^2$ norm, even if the results are proved for the $\ell^q$ norm, with $q \geq 1$. 

\subsubsection{Robustness to unknown error}

The general robustness analysis of QCBP is carried out in Section~\ref{sec:robustness} and it is built upon the concepts of \emph{quotients} and \emph{robust instance optimality}. 

The quotients are sup-inf constants of two types: the $\ell^1$-quotient (Definition~\ref{def:quotient})
\begin{equation}
\mathcal{Q}_\lambda(\Amat)_1:= \sup_{\substack{\evec\in\CC^m \\ \evec\neq \Ovec}} \; \inf_{\substack{\zvec \in \CC^N\\ \Amat \zvec = \evec }} \frac{\|\zvec\|_1}{ \sqrt{\lambda}\;\|\evec\|_2},
\end{equation}
and the simultaneous $(\ell^2,\ell^1)$-quotient (Definition~\ref{def:simulquotient})
\begin{equation}
\mathcal{Q}_\lambda(\Amat)_{2,1}
:= \sup_{\substack{\evec\in\CC^m \\ \evec\neq \Ovec}} \; 
\inf_{\substack{\zvec \in \CC^N\\ \Amat \zvec = \evec}} 
\frac{ \sqrt{\lambda}\|\zvec\|_2 + \|\zvec\|_1}{\sqrt{\lambda}\|\evec\|_2}.
\end{equation}
These two quantities are strictly related to the constants involved in the quotient property and in the simultaneous quotient property, respectively \cite{Foucart2014,Wojtaszczyk2010} (see Section~\ref{sec:quotientproperties} for a detailed discussion). 

The notion of robust instance optimality  generalizes the concept of instance optimality, already known in CS \cite{Cohen2008}. An encoder-decoder pair $(\Amat,\Delta_\eta)$ is said to be \emph{$\eta$-robustly instance optimal}  (Definition~\ref{def:robIO}) if 
\begin{equation}
\|\xvec-\Delta_\eta(\Amat \xvec + \evec)\|_2 \lesssim \frac{\sigma_s(\xvec)_1}{\sqrt{s}} + \eta, \quad \forall \xvec \in \CC^N, \quad \forall \evec\in\CC^m \text{ s.t. } \|\evec\|_2 \leq \eta.
\end{equation}
This definition corresponds to the usual robustness results proved in CS, assuming an \emph{a priori} estimate on the error norm. 

Employing these two notions, in our first result we prove that robust instance optimality implies robustness to unknown error, where the resulting error estimate is of the form \eqref{eq:newestimate} and the unknown error is multiplied by the simultaneous quotient $\mathcal{Q}_s(\Amat)_{2,1}$ (Theorem~\ref{thm:robIO+QP=>RobRec}). In other words, an instance optimal decoder for the noise-aware setting can be applied in the noise-blind case with an explicit error bound.  The version of this result specialized to the $\ell^2$ norm reads as follows:
\begin{thm}[Robust instance optimality $\Rightarrow$ Robustness to unknown error] Assume that the pair $(\Amat,\Delta_\eta)$ is $\eta$-robustly instance optimal. Then for every $\xvec\in\CC^N$, $\evec\in\CC^m$, the following holds
\begin{equation}
\|\xvec-\Delta_\eta(\Amat\xvec+\evec)\|_2 \lesssim \frac{\sigma_s(\xvec)_1}{\sqrt{s}} + \eta + 
\mathcal{Q}_s(\Amat)_{2,1} \; \max\{\|\evec\|_2-\eta,0\}.
\end{equation}
\end{thm}
Now, we recall that the $s^{th}$ restricted isometry constant  $\delta_s(\Amat)$ of $\Amat$  is the smallest constant $\delta$ such that 
$$
(1-\delta)\|\zvec\|_2^2 \leq \|\Amat\zvec\|_2^2 \leq (1+\delta) \|\zvec\|_2^2, \quad \forall \zvec \in \Sigma_s^N.
$$
Moreover, $\Amat$ is said to have the restricted isometry property if $\delta_{s}(\Amat) < 1$ (see Definition~\ref{def:RIP}). It is well-known in CS that the QCBP decoder is robustly instance optimal when $\delta_{2s}(\Amat)<4/\sqrt{41}$ (see \cite[Theorem 6.12]{Foucart2013}). Moreover, we show that under the restricted isometry property -- or, more in general, under the $\ell^2$-robust null space property (Definition~\ref{def:robNSP}) -- we can control the simultaneous quotient by the quotient $\mathcal{Q}_\lambda(\Amat)_{2,1} \lesssim \mathcal{Q}_\lambda(\Amat)_1$ (Proposition~\ref{prop:SQCleqQC}). We prove that QCBP is robust to unknown error under the restricted isometry property in Corollary~\ref{cor:RIP->robrec}, which corresponds to the following result.\footnote{The reason why we do not use the less restrictive (and sharp) condition $\delta_{2s}(\Amat) < 1/\sqrt{2}$ given in \cite{Cai2014} is discussed in Remark~\ref{rmrk:RIPconstant}}

\begin{thm}[Restricted isometry property $\Rightarrow$ Robustness under unknown error]
If the $(2s)^{th}$ restricted isometry constant of $\Amat$ satisfies  $\delta_{2s}(\Amat) < 4/\sqrt{41}$, then
$$
\|\xvec - \Delta_\eta(\Amat\xvec + \evec)\|_2 \lesssim
\frac{\sigma_s(\xvec)_1}{\sqrt{s}} + \eta + \mathcal{Q}_s(\Amat)_1 \max\{\|\evec\|_2-\eta,0\},
$$
where $\Delta_\eta$ is the QCBP decoder defined in \eqref{eq:defQCBPdecoder}.
\end{thm}
Therefore, we establish robustness to unknown error -- up to the magnitude of $\mathcal{Q}_s(\Amat)_1$ -- whenever the sensing matrix has the restricted isometry property. In practice, this translates into a condition on the measurements of the form $
m \gtrsim s \; \mathcal{L}(N,s)$, where $\mathcal{L}(N,s)$ is a polylogarithmic factor depending on the random model considered for the sensing matrix (see also Remark~\ref{rmrk:polylog}). We also note that the term $\max\{\|\evec\|_2-\eta,0\}$ in the right-hand side of the recovery error estimate suggests that there is a benefit to estimating the noise well and to calibrating the threshold parameter $\eta$ accordingly.

\subsubsection{Random Gaussian matrices}
In Section~\ref{sec:Gauss} we prove the robustness of QCBP in the case of Gaussian measurements. The main result is Theorem~\ref{thm:RobustGauss}, which reads as follows in the $\ell^2$ norm case:
\begin{thm}[Robustness to unknown error of QCBP with Gaussian measurements]
Let $s \leq m \leq N/2$ and $\Amat = m^{-1/2}\Gmat$, where $\Gmat\in\CC^{m \times N}$ is a random Gaussian matrix with independent standard normal entries and let
$$
m \gtrsim s \ln(eN/s) + \ln(\varepsilon^{-1}).
$$
Then for every $\xvec\in\CC^N$, $\evec\in\CC^m$, and $s\leq s_* = m/\ln(eN/m)$, we have
$$
\|\xvec - \Delta_\eta(\Amat\xvec + \evec)\|_2
\lesssim
\frac{\sigma_s(\xvec)_1}{\sqrt{s}} + \eta +\max\{\|\evec\|_2-\eta,0\}
$$
with probability at least $1-\varepsilon$, where $\Delta_\eta$ is the QCBP decoder defined in \eqref{eq:defQCBPdecoder}.
\end{thm}
This result extends the previously known robustness result for BP with Gaussian measurements (see Section~\ref{sec:literature}) and based on already known upper bounds to the quotient $\mathcal{Q}_{s_*}(\Amat)_1$ with $s_* = m/\ln(eN/m)$ in probability (see Theorem~\ref{thm:QPGauss}).

\subsubsection{Random matrices with heavy-tailed rows}
While Gaussian random matrices have convenient mathematical properties, they are generally of limited use in applications of CS.  In Section~\ref{sec:heavy-tailed} we apply our robustness theory to a large class of random sampling matrices with heavy-tailed rows.

First, we consider the case of a random sampling matrix $\Amat$ from  a Bounded Orthonormal Systems (BOSs) \cite{Rauhut2010}, namely
$$
A_{ij} = \frac{1}{\sqrt{m}} \phi_j (t_i), \quad \forall j \in [N], \; \forall i \in [m],
$$  
where the functions $\phi_j : \mathcal{D} \to \CC$ form an orthonormal system with respect to the probability measure $\nu$ on $\mathcal{D}$ and are uniformly bounded, namely $\|\phi_j\|_{L^\infty(\mathcal{D})} \leq K$, with $K \geq 1$ (see Definition~\ref{def:BOS}). This is a large class of random matrices, which includes the partial discrete Fourier transform, nonharmonic Fourier measurements, random sampling from orthogonal polynomials, and subsampled isometries (with independent rows) all of which occur commonly in applications of CS.

In the BOS case, the restricted isometry property is known to hold with high probability (Theorem~\ref{thm:RIPforBOS}). Therefore, in order to prove robustness to unknown error of QCBP in this framework, we only need to control the $\ell^1$-quotient. With this aim, we show that the $\ell^1$-quotient $\mathcal{Q}_\lambda(\Amat)_1$ can be bounded as follows (Proposition~\ref{prop:QminSV}):
$$
\mathcal{Q}_\lambda(\Amat)_1
\leq \frac{1}{\sigma_{\min}(\sqrt{\frac{m}{N}}\Amat^*)}\sqrt{\frac{m}{\lambda}},
$$
for every $\Amat \in \CC^{m \times N}$, not necessarily associated with a BOS.  We also show that $\mathcal{Q}_\lambda(\Amat)_1 \geq K^{-1}\sqrt{m/\lambda}$ when  $\max_{ij}|A_{ij}| \leq K$; in other words, the factor $\sqrt{m/\lambda}$ is optimal. Notice that $\sigma_{\min}(\sqrt{\frac{m}{N}}\Amat^*)$  is related to the minimum eigenvalue of the Gram matrix $\frac{m}{N}\Amat\Amat^*$, and not of $\Amat^*\Amat$ usually considered when proving the restricted isometry property in CS. The scaling $\sqrt{m/N}$ factor is required so that $\sigma_{\min}(\sqrt{\frac{m}{N}}\Amat^*) \approx 1$ under suitable conditions on $m$ and $N$.

In order to control $\sigma_{\min}(\sqrt{\frac{m}{N}}\Amat^*)$, we employ tools from the spectral theory of random matrices with heavy-tailed columns \cite{Vershynin2012}. In particular, we consider the \emph{cross coherence} parameter (Definition~\ref{def:cross_coherence}) 
$$
\mu = \left(\frac{m}{N}\right)^2 \Expe \bigg[\max_{k \in [m]} \sum_{j \in [m]\setminus\{k\}}|\langle \avec_j, \avec_k \rangle|^2\bigg],
$$
and the \emph{distrotion} parameter (Definition~\ref{def:distortion})
$$
\xi = \Expe \left[\max_{k \in [m]} \left|\frac{m}{N}\|\avec_k\|_2^2-1\right|\right], 
$$
where $\avec_1,\ldots,\avec_m$ are the rows of $\Amat$.
They control the off-diagonal and the diagonal part of the Gram matrix $\frac{m}{M} \Amat\Amat^*$, respectively. In Theorem~\ref{thm:svheavytailedcols} we prove the following deviation inequality in expectation 
$$
\Expe |\sigma_{\min}(\sqrt{\tfrac{m}{N}}\Amat^*) - 1| \lesssim \xi + \sqrt{(1+\xi) \mu \ln(m)}.
$$
 This allows us to prove robustness results for random matrices with heavy-tailed rows. In the case of BOSs we have Theorem~\ref{thm:BOSrobust},  stated in a simplified version below.
\begin{thm}[Robustness to unknown error of QCBP for BOSs]
Let $\Amat\in\CC^{m \times N}$ be the random sampling matrix \eqref{eq:BOSmatrix} associated with a BOS with constant $K\geq 1$, whose distortion parameter satisfies
\begin{equation}
\xi \lesssim \min\bigg\{\sqrt{\frac{m^2 \ln(m)}{N}},1\bigg\}.
\end{equation}
Then there exist a function $\mathcal{L}(N,s,\varepsilon,K)$ depending at most polylogarithmically on $N$ and $s$ such that the following holds. For every $N \in \NN$ and $\varepsilon \in (0,1)$, assume that the sparsity $s$ satisfies
\begin{equation}
s \lesssim \frac{\varepsilon \; \sqrt{N}}{\mathcal{L}(N,s,\varepsilon,K) \;  \ln^{\frac12}(N)},
\end{equation}
and consider a number of measurements
\begin{equation}
m \sim  s \; \mathcal{L}(N,s,\varepsilon,K).
\end{equation}
Let $\Delta_\eta$ be the QCBP decoder defined in  \eqref{eq:defQCBPdecoder}. Then for every  $\xvec \in \CC^N$, $\evec \in \CC^m$, the following robust error estimate holds                
\begin{equation}   
\|\xvec - \Delta_\eta(\Amat\xvec + \evec)\|_2 \lesssim \frac{\sigma_s(\xvec)_1 }{\sqrt{s}} +  \eta + \mathcal{L}^{\frac12}(N,s,\varepsilon,K)  \max\{\|\evec\|_2 - \eta,0\},
\end{equation}
with probability at least $1 - \varepsilon$. A possible choice for $\mathcal{L}(N,s,\varepsilon,K)$ is given by \eqref{thm:RIPforBOS:eq:L} with $\delta = 1/2$ and $\varepsilon/2$ in place of $\varepsilon$.
\end{thm}
The recovery result does not rely on any assumption on the error $\evec$, but it has three main limitations: (1) The term $\max\{\|e\|_2-\eta,0\}$ in the error estimate is multiplied by a polylogarithmic factor. (2) We require an upper bound to the distortion parameter $\xi$, which has to be computed on a case-by-case basis. (3) The sparsity level is limited to the regime $s \lesssim \varepsilon \sqrt{N}$. In particular, the linear dependence on $\varepsilon$ prevents us from asserting that the theorem holds with `overwhelmingly' high probability. 

In Section~\ref{sec:subisoindep} we discuss the case of subsampled isometries with random independent samples. This allows us to employ the theory showed for BOSs (notice that in this case we have $\xi = 0$), resulting in Theorem~\ref{thm:subisoindepRobust}. The dependence of $s$ on $\varepsilon$ is improved, in that the sparsity is required to satisfy $s \lesssim \ln(\frac{2}{2-\varepsilon})\sqrt{N}$.  

Finally, in Section~\ref{sec:subBernoulli} we discuss robustness to unknown error for subsampled isometries randomly generated via Bernoulli selectors and with the random subset model. In both cases, the sensing matrix $\Amat$ cannot have repeated rows and, consequently, $\sigma_{\min}(\sqrt{\frac{m}{N}}\Amat^*)=1$ with probability 1. Therefore, robustness to unknown error is guaranteed under the restricted isometry property  due to Corollary~\ref{cor:RIP->robrec} (or, more in general, under the robust null space property due to Theorem~\ref{thm:NSP->robrec}). In particular, for subsampled isometries via Bernoulli selectors, provided that the isometry $\Umat \in \CC^{N \times N}$ to be subsampled satisfies 
$$
\max_{i,j\in[N]} |U_{ij}| \leq \frac{K}{\sqrt{N}},
$$
and that 
$$
m \sim K^2 s \ln^2(s) \ln^2(N),
$$
the following robust recovery error estimate holds for every $\xvec\in\CC^N$ and $\evec\in\CC^m$ with high probability:
$$
\|\xvec-\Delta_{\eta}(\Amat \xvec + \evec)\|_2 \lesssim \frac{\sigma_s(\xvec)_1}{\sqrt{s}} + \eta + K\ln(s)\ln(N) \max\{\|\evec\|_2 - \eta,0\},
$$
where $\Delta_\eta$ is the QCBP decoder defined in \eqref{eq:defQCBPdecoder}. Remarkably, there is no need to assume any restriction on the sparsity level as in the previous results. In particular, being $s$ independent of the probability of failure of the recovery error estimate, in this case QCBP is robust to unknown error with `overwhelmingly' high probability. Moreover, there is no restriction on $m$, in contrast to the Gaussian case where $m \leq N/2$.

\subsubsection{Examples}
Section~\ref{sec:examples} is devoted to discussing the application of our robustness analysis to concrete examples of BOSs, usually employed in CS: (1) The partial discrete Fourier transform, which is a particular case of subsampled isometry with $K=1$. (2) Nonharmonic Fourier measurements, where $K = 1$ and $\xi = 0$. (3) Random sampling from orthogonal Chebyshev polynomials, where $K > 1$ and the distortion $\xi$ can be explicitly bounded from above as $\xi \lesssim \sqrt{m/N}$ (Proposition~\ref{prop:distCheby}). 

Finally,  we present some numerical experiments to show the reliability of our analysis and discuss an application to one-dimensional polynomial approximation.

\subsection{Related literature}
\label{sec:literature}
Here we review the main results on the robustness analysis of CS. This is not intended to be an exhaustive literature review, due to the large volume of papers published on CS after 2006. For a wider discussion we refer the reader to \cite[Chapter 11]{Foucart2013}.

\subsubsection{Basis pursuit} 

When considering the BP program, it is possible to show recovery estimates analogous to \eqref{eq:usualCSerrorbound} where $\eta$ is replaced by $\|\evec\|_2$ in the case of random Gaussian matrices \cite{Wojtaszczyk2010}. Namely,
\begin{equation}
\|\xvec- \Delta_0(\Amat \xvec + \evec)\|_2 \lesssim \frac{\sigma_s(\xvec)_1}{\sqrt{s}} + \|\evec\|_2.
\end{equation}
These estimates are based on the so-called \emph{quotient property} (introduced in the context of CS in \cite{Wojtaszczyk2010}), which is known to be fulfilled only by random Gaussian matrices \cite{Wojtaszczyk2010} and by Weibull matrices \cite{Foucart2014}. In \cite{DeVore2009} the quotient property relative to the modified norm $\|y\|^{(N,m)} := \max\{\|\yvec\|_2, \sqrt{\ln(eN/m))}\|\yvec\|_\infty\}$ is shown for Bernoulli matrices, i.e., $\|\evec\|_2$ is replaced by $\|\evec\|^{(N,m)}$ above. See Sections~\ref{sec:QP} and \ref{sec:Gauss} for further comments on these results.

\subsubsection{Quadratically-constrained basis pursuit}
In \cite{Herman2010} the authors study the robustness of QCBP under noise, where both the matrix and the measurements are perturbed, namely $(\Amat + \Emat) \xvec = \yvec + \evec$. The $s^{th}$ restricted isometry constant of $\Amat+\Emat$ is shown to depend on the $s^{th}$ restricted isometry constant of $\Amat$ and on a suitable constant $\varepsilon_{\Amat}^{(s)}$ that controls the sensing matrix's perturbation, namely
\begin{equation}
\frac{\|\Emat\|^{(s)}}{\|\Amat\|^{(s)}}\leq \varepsilon_{\Amat}^{(s)},
\end{equation}
where, given an $m \times N$ matrix $\Mmat$, the quantity  $\|\Mmat\|^{(s)}$ denotes the maximum over all the spectral norms of the $m \times s$ submatrices of $\Mmat$ (see \cite[Theorem 1]{Herman2010}). 

Moreover, the authors give a robust recovery result for QCBP associated with the perturbed sensing matrix and with the perturbed measurement vector. In \cite[Theorem 2]{Herman2010} it is proved that the recovery error of QCBP satisfies \eqref{eq:usualCSerrorbound} when the parameter $\eta$ is chosen in a suitable way. However, the sufficient condition on $\eta$ depends on the perturbation measure $\varepsilon_{\Amat}^{(s)}$, on similar measures controlling the perturbation  $\evec$ over the measurements, on the sparsity level  $s$, and on the exact solution $\xvec$. In particular, this result is still dependent on some \emph{a priori} information on $\|\evec\|_2$.

\subsubsection{Alternatives to basis pursuit}

Similar robust recovery estimates are also available for algorithms such as Iterative Hard Thresholding (IHT), Compressive Sampling Matching Pursuit (CoSaMP), and Orthogonal Matching Pursuit (OMP). In all these cases it is possible to prove that
\begin{equation}
\|\xvec - \Delta_s(\Amat \xvec + \evec)\|_2 \lesssim \frac{\sigma_s(\xvec)_1}{\sqrt{s}} + \|\evec\|_2,
\end{equation}
where $\Delta_s$ is the decoder associated with the aforementioned algorithms, depending on the desired sparsity level $s$. Although not requiring any information about the error $\evec$, all these techniques require an \emph{a priori} knowledge of the sparsity level $s$ that is not necessary for QCBP. The robustness analysis of IHT, CoSaMP, and OMP can be found in \cite[Theorem 6.21]{Foucart2013}, \cite[Theorem 6.25]{Foucart2013}, and \cite[Theorem 6.28]{Foucart2013}, respectively. These results are improvement over arguments presented in \cite{Blumensath2009} for IHT, \cite{Needell2009} for CoSaMP, and \cite{Zhang2011} for OMP. 

It is also worth mentioning the recovery analysis carried out in \cite{Candes2011}. The authors study the LASSO unconstrained minimization program 
\begin{equation}
\min_{\zvec\in\CC^N}\|\zvec\|_1 + \lambda \|\Amat \zvec - \yvec\|_2^2
\end{equation}
The main advantage of their ``RIPless'' theory is that -- as suggested by the name -- it is not based on the RIP, but mainly on the notion of incoherent measurements. However, the choice of the parameter $\lambda$ depends on the  variance of the noise $\evec$ (here assumed to be Gaussian). As a consequence, some \emph{a priori} information about the measurement corruption is still required.

In \cite{Plan2016} the authors provide  recovery error estimates for the so-called $K$-LASSO optimization program
\begin{equation}
\min_{\zvec\in \CC^N} \|\Amat\zvec-\yvec\|_2, \quad \text{s.t. } \zvec \in K,
\end{equation}
where $K$ is a given subset of $\CC^N$, which usually models some structure in the signal. Moreover, $\yvec$ is assumed to depend nonlinearly on $\Amat\xvec$, as in the case of \emph{1-bit} CS \cite{Boufounos2008} and of \emph{binary statistical models} \cite{Davenport2014,Negahban2012}.

\section{Robustness to unknown error}
\label{sec:robustness}

In this section we present the theoretical analysis that establishes the robustness to unknown error for QCBP in a very general framework.

After recalling the concepts of robust null space property and of restricted isometry property of CS in Section~\ref{sec:rNSPRIP}, we introduce the notions of quotient and simultaneous quotient in Section~\ref{sec:QP}. Finally, in Section~\ref{sec:robinstopt} we introduce the notion of robust instance optimality and prove that under the robust null space property (or the restricted isometry property) and under a suitable control over the quotient, QCBP is robust to unknown error.

See also Figure~\ref{fig:theory} for a conceptual map of the notions and results presented in this section.
\begin{figure}
\centering
\includegraphics[width = 8cm]{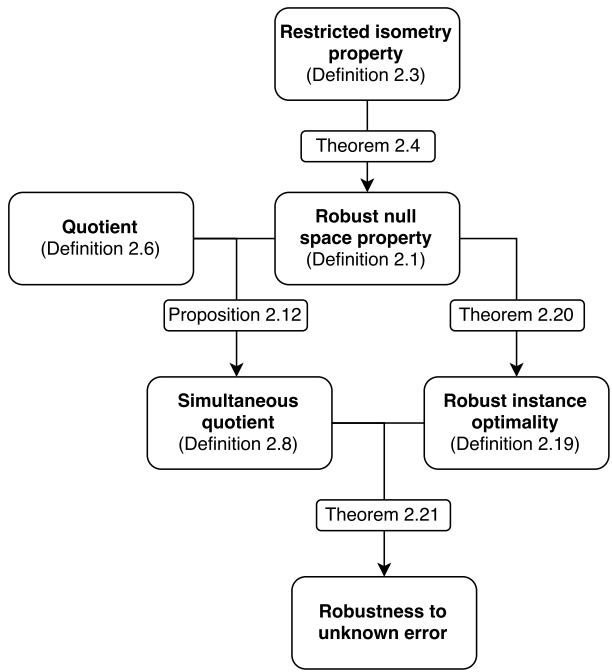}
\caption{\label{fig:theory}Conceptual map of the robustness theory presented in this paper.}
\end{figure}

\subsection{Robust null space and restricted isometry properties}
\label{sec:rNSPRIP}
We introduce the robust null space property and the restricted isometry constants. These are nowadays considered standard tools for the theoretical analysis of CS. We refer the reader to \cite[Chapter 4]{Foucart2013} and \cite[Chapter 6]{Foucart2013} for more details and an extensive literature review.  

For the sake of generality, in this section and in the following one we will present the definitions and the results considering the $\ell^q$ norm with $q \geq 1$ on $\CC^N$ and a generic norm $\|\cdot\|$ on $\CC^m$. However, it is worth keeping in mind that the framework corresponding to $q = 2$ and $\|\cdot\| = \|\cdot\|_2$ is of particular interest.

\begin{defn}[$\ell^q$-robust null space property]
\label{def:robNSP}
Given $q \geq 1$, the matrix $\Amat \in\CC^{m \times N}$ satisfies the \emph{$\ell^q$-robust null space property} of order $s$ relative to a norm $\|\cdot\|$ on $\CC^m$ with constants $0 < \rho < 1$ and $\tau >0$ if, for any set $S\subseteq [N]$ with $|S|\leq s$, the following inequality holds
\begin{equation}
\label{eq:lq-robNSP}
\|\zvec_S\|_q \leq s^{1/q-1}\rho \|\zvec_{\overline S}\|_1 + \tau \|\Amat \zvec\|, \quad \forall \zvec \in \CC^N.
\end{equation}
\end{defn}
Under the robust null space property, it is possible to prove a useful and quite general technical result that will be a fundamental building block to produce robust recovery error estimates for QCBP with respect to the $\ell^p$ norm. This is stated in the following result (see \cite[Theorem 4.25]{Foucart2013}):
\begin{thm}
\label{thm:lqrobNSPest}
Given $1\leq p \leq q$, suppose that the matrix $\Amat\in\CC^{m\times N}$ satisfies the $\ell^q$-robust null space property of order $s$ with constants $0 < \rho < 1$ and $\tau > 0$ relative  to $\|\cdot\|$. Then for any $\zvec,\wvec \in \CC^N$, the following inequality holds
\begin{equation}
\|\zvec-\wvec\|_p \leq \frac{C}{s^{1-1/p}}  (\|\zvec\|_1-\|\wvec\|_1 + 2\sigma_s(\wvec)_1)  + \frac{D}{s^{1/q-1/p}} \|\Amat(\zvec-\wvec)\|, 
\end{equation}
where $C := (1 + \rho)^2/(1-\rho)$ and $D := (3 + \rho)\tau/(1-\rho)$.
\end{thm}

We recall the definition of \emph{restricted isometry constants} and of \emph{restricted isometry property}, ubiquitous in the CS literature. 

\begin{defn}[Restricted isometry constants]
\label{def:RIP}
The \emph{$s^{th}$ restricted isometry constant} $\delta_s(\Amat)$ of a matrix $\Amat \in \CC^{m \times N}$ is the smallest constant $\delta \geq 0$ such that 
\begin{equation}
\label{eq:RIP}
(1 - \delta)\|\zvec\|_2^2 \leq \|\Amat \zvec \|_2^2 \leq (1 + \delta)\|\zvec\|_2^2, \quad \forall\zvec \in \Sigma^N_s.
\end{equation} 
If \eqref{eq:RIP} holds with $0 < \delta < 1$, then $\Amat$ is said to have the \emph{restricted isometry property}.
\end{defn}

The following result gives a sufficient condition on the $(2s)^{th}$ restricted isometry constant for the $\ell^2$-robust null space property to hold (see \cite[Theorem 6.13]{Foucart2013}). 
\begin{thm}
\label{thm:RIP2s=>NSP}
 If the $(2s)^{th}$ restricted isometry constant of $\Amat \in \CC^{m \times N}$ obeys 
\begin{equation}
\label{thm:RIP2s=>NSP:eq:deltacond}
\delta_{2s}(\Amat) < \frac{4}{\sqrt{41}} \approx 0.6247,
\end{equation} 
then the matrix $\Amat$ satisfies the $\ell^2$-robust null space property of order $s$ relative to $\|\cdot\|_2$ on $\CC^m$ with constants $0 < \rho < 1$ and $\tau > 0$ depending only on $\delta_{2s}(\Amat)$. In particular, 
\begin{equation}
\label{thm:RIP2s=>NSP:eq:rhotau}
\rho = \frac{\delta_{2s}(\Amat)}{\sqrt{1-\delta_{2s}(\Amat)^2}-\delta_{2s}(\Amat)/4}, 
\quad 
\tau = \frac{\sqrt{1+\delta_{2s}(\Amat)}}{\sqrt{1-\delta_{2s}(\Amat)^2}-\delta_{2s}(\Amat)/4}.
\end{equation}
\end{thm}

Theorems~\ref{thm:lqrobNSPest} and \ref{thm:RIP2s=>NSP} immediately imply robust recovery error estimates with respect to the $\ell^p$ norm for $1 \leq p \leq 2$ for QCBP under the regime $\|\evec\|\leq\eta$ and when $\delta_{2s}(\Amat)$ is bounded.

\begin{rmrk}
\label{rmrk:RIPconstant}
We note that in \cite{Cai2014} the condition $\delta_{ts}(\Amat)< \sqrt{(t-1)/t}$ with $t \geq 4/3$ (that becomes $\delta_{2s}(\Amat)< 1/\sqrt{2}$ by choosing $t = 2$) has been proved to guarantee robust recovery of QCBP when  $\|\evec\|_2 \leq \eta$. Yet, in our analysis we need a sufficient condition on the restricted isometry constant such that the $\ell^2$-robust null space property holds and this is not explicitly discussed in \cite{Cai2014}. For this reason, we decided to keep the condition $\delta_{2s}(\Amat) < 4 / \sqrt{41}$. 
\end{rmrk}

\subsection{Quotients and quotient properties}
\label{sec:QP}

The \emph{quotient property} is the main tool underlying all the the robust recovery results published in the CS literature so far (see Section~\ref{sec:literature}). This property was introduced  in the CS context by Wojaszczyk in \cite{Wojtaszczyk2010} and generalized later by Foucart in \cite{Foucart2014}. It has been proved to hold  for random Gaussian, subgaussian (in particular, Bernoulli), and Weibull matrices. See \cite[Chapter 11]{Foucart2013} for further details and more pointers to the literature.

In this section, we revisit the quotient property by proposing a novel formalization of it based on some sup-inf constants called \emph{quotients}. These objects will make the role played by the quotient property more transparent and will allow us to generalize some CS robustness results in Section~\ref{sec:robinstopt}. 

The section is divided in two main parts. First, we introduce the notions of quotient and of simultaneous quotient and study how they interrelate (Section~\ref{sec:quotients}). Then we examine the tight link between the quotients and the quotient property in its different declinations and make some historical remarks on the different versions of quotient property published in the literature so far  (Section~\ref{sec:quotientproperties}).

\subsubsection{Quotients}
\label{sec:quotients}
We introduce some sup-inf constants called \emph{quotients}, that will be directly linked to the constants used in the quotient property. In particular, we define two different types of quotients, called the $\ell^q$-quotient and the simultaneous $(\ell^q,\ell^1)$-quotient.

First, we define the concept of $\ell^q$-quotient. For every error $\evec\in\CC^m$ with $\|\evec\|> \eta$, it measures the discrepancy between the amount of error norm exceeding the QCBP threshold $\eta$ and the $\ell^q$ norm of the vectors in the feasible set of QCBP with measurements $\yvec = \evec$.
\begin{defn}[Quotient]
\label{def:quotient}
Given $q \geq 1$, $\lambda >0$, and $\Amat\in\CC^{m \times N}$, we define the \emph{$\ell^q$-quotient of order $\lambda$ and threshold $\eta$ relative to the norm $\|\cdot\|$ on $\CC^m$} as
\begin{equation}
\label{eq:def_lq-quotient}
\mathcal{Q}_\lambda(\Amat,\eta)_q:= \sup_{\substack{\evec\in\CC^m \\ \|\evec\| > \eta}} \; \inf_{\substack{\zvec \in \CC^N\\ \|\Amat \zvec - \evec\| \leq \eta}} \frac{\lambda^{1/2-1/q}\|\zvec\|_q}{ \|\evec\| - \eta }.
\end{equation} 
In order to simplify the notation, we also define  
\begin{equation}
\label{eq:def_lq_quotient_0}
\mathcal{Q}_\lambda(\Amat)_{q} := \mathcal{Q}_\lambda(\Amat,0)_{q}.
\end{equation}
If there exists a vector $\evec\in\CC^m$ with $\|\evec\|> \eta$ such that the feasible set $\{\zvec\in\CC^N:\|\Amat\zvec-\evec\|\leq \eta\}$ is empty, then we define $\mathcal{Q}_\lambda(\Amat,\eta)_q = +\infty$.
\end{defn}
Notice the presence of the threshold parameter $\eta$, which will let us generalize the robust recovery results for BP to novel robust recovery results for QCBP.

\begin{rmrk}
\label{rmrk:supmin_lq_quotient}
In \eqref{eq:def_lq-quotient}, the infimum is actually a minimum because of  the convexity of the map $\zvec \mapsto \|\zvec\|_q$ for $q \geq 1$ and of the set $\{\zvec\in\CC^N : \|\Amat \zvec - \evec\|\leq \eta\}$. 
\end{rmrk}

We also define the \emph{simultaneous quotient}, strictly related to the simultaneous quotient property. For every error $\evec\in\CC^m$ with $\|\evec\|> \eta$, the simultaneous $(\ell^q,\ell^1)$-quotient measures the discrepancy between the amount of error norm exceeding the QCBP threshold $\eta$ and a suitable combination between the $\ell^q$ norm and the $\ell^1$ norm of the vectors in the feasible set of QCBP with measurements $\yvec = \evec$.

\begin{defn}[Simultaneous quotient]
\label{def:simulquotient}
We define the \emph{simultaneous $(\ell^q,\ell^1)$-quotient of order $\lambda$ and threshold $\eta$ relative to the norm $\|\cdot\|$ on $\CC^m$} as
\begin{equation}
\label{eq:def_lql1_quotient}
\mathcal{Q}_\lambda(\Amat,\eta)_{q,1}
:= \sup_{\substack{\evec\in\CC^m \\ \|\evec\| > \eta}} \; 
\inf_{\substack{\zvec \in \CC^N\\ \|\Amat \zvec - \evec\| \leq \eta}} 
\frac{ \lambda^{1/2-1/q} \|\zvec\|_q + \lambda^{-1/2}\|\zvec\|_1}{\|\evec\| - \eta}.
\end{equation} 
In order to simplify the notation, we also define  
\begin{equation}
\label{eq:def_lql1_quotient_0}
\mathcal{Q}_\lambda(\Amat)_{q,1} := \mathcal{Q}_\lambda(\Amat,0)_{q,1}.
\end{equation}
If there exists a vector $\evec\in\CC^m$ with $\|\evec\|> \eta$ such that the feasible set $\{\zvec \in \CC^N:\|\Amat\zvec-\evec\|\leq \eta\}$ is empty, then we define $\mathcal{Q}_\lambda(\Amat,\eta)_{q,1} = +\infty$.
\end{defn}
\begin{rmrk}
\label{rmrk:supmin_sim_lql1_quotient}
Analogously to Remark~\ref{rmrk:supmin_lq_quotient}, we notice that the infimum in \eqref{eq:def_lql1_quotient} is actually a minimum. 
\end{rmrk}

The following lemma shows that the quotients $\mathcal{Q}_\lambda(\Amat)_1$ and $\mathcal{Q}_\lambda(\Amat)_{q,1}$ control all the quotients of the same order  for any value of the threshold parameter $\eta$.
\begin{lem}
\label{lem:quotientQCBP} 
For every $\eta \geq 0$, $q \geq 1$, and $\lambda >0$, the following upper bounds hold
$$
\mathcal{Q}_\lambda(\Amat,\eta)_q \leq \mathcal{Q}_\lambda(\Amat)_q  
\quad 
\text{and}
\quad  
\mathcal{Q}_\lambda(\Amat,\eta)_{q,1} \leq \mathcal{Q}_\lambda(\Amat)_{q,1}.
$$
\end{lem}
\begin{proof}
We prove the first upper bound. For any $\evec\in\CC^m$, consider a vector $\widetilde\zvec \in\CC^N$ such that $\Amat \widetilde\zvec = \evec$ and that $\|\widetilde\zvec\|_q/(\|\evec\| \sqrt{\lambda}) \leq \mathcal{Q}_s(\Amat)_1$ (such a vector exists due to Remark~\ref{rmrk:supmin_lq_quotient}). Then, defining $\zvec := (1 - \eta /\|\evec\|) \widetilde\zvec$, we have $\|\Amat \zvec - \evec\| = \eta$ and 
\begin{equation}
\frac{\lambda^{1/2-1/q}\|\zvec\|_q}{\|\evec\| - \eta} 
= \frac{\lambda^{1/2-1/q}\|\widetilde\zvec\|_q}{\|\evec\|}
\leq \mathcal{Q}_\lambda(\Amat)_1.
\end{equation}
Since $\evec\in\CC^m$ was arbitrary, we get $\mathcal{Q}_\lambda(\Amat,\eta)_1 \leq \mathcal{Q}_\lambda(\Amat)_1$. 

The inequality $\mathcal{Q}_\lambda(\Amat,\eta)_{q,1} \leq \mathcal{Q}_\lambda(\Amat)_{q,1}$ can be proved using an analogous argument and recalling Remark~\ref{rmrk:supmin_sim_lql1_quotient}. 
\end{proof}

Now, we prove a crucial property of the quotients. Under the $\ell^q$-robust null space property, the simultaneous $(\ell^q,\ell^1)$-quotient is dominated by the $\ell^1$-quotient up to a multiplicative constant and an additive constant that are related to the $\ell^q$-robust null space property. Therefore, the $\ell^1$-quotient $\mathcal{Q}_\lambda(\Amat)_1$ is the only object that needs to be studied in order to control all the other (simultaneous) quotients. 

Let us first we recall a technical lemma that relates the best $s$-term approximation with respect to the $\ell^q$ norm and the $\ell^p$ norm of a vector (see \cite[Proposition 2.3]{Foucart2013}).
\begin{lem}
\label{lem:sigmaqleqnormp}
For every $0 < p \leq q$, we have $\sigma_s(\xvec)_q \leq s^{1/q-1/p}\|\xvec\|_p$.
\end{lem}

\begin{prop}
\label{prop:SQCleqQC}
Fix $q \geq 1$ and $\Amat\in\CC^{m \times N}$. Then, if $\Amat$ satisfies the $\ell^q$-robust null space property of order $s$ with constants $0<\rho<1$ and $\tau>0$ relative to the rescaled norm $s^{1/q-1/2}\|\cdot\|$ on $\CC^m$, the following bound holds
\begin{equation}
\label{eq:Q1qleqQ1}
\mathcal{Q}_s(\Amat)_{q,1} \leq (\rho +2) \mathcal{Q}_s(\Amat)_1 + \tau,
\end{equation}
where the quotients are relative to the norm $\|\cdot\|$ on $\CC^m$.
\end{prop}
\begin{proof}
Fix an arbitrary vector $\evec \in\CC^m$ and choose a vector $\zvec \in \CC^N$ such that $\Amat\zvec = \evec$ and that
\begin{equation}
\label{eq:prop:SQCleqQC:norm1z}
\|\zvec\|_1 \leq s^{1/2} \mathcal{Q}_s(\Amat)_1 \|\evec\|.
\end{equation}
Now, consider a  set $S\subseteq[N]$ of cardinality $s$ that contains the $s$ largest entries of $\zvec$. Lemma~\ref{lem:sigmaqleqnormp} with $p=1$ ensures that
\begin{equation}
\label{eq:prop:SQCleqQC:zSbar}
\|\zvec_{\overline{S}}\|_q = \sigma_s(\zvec)_q \leq s^{1/q-1}\|\zvec\|_1.
\end{equation}
Additionally, the $\ell^q$-robust null space property relative to the rescaled norm $s^{1/q-1/2}\|\cdot\|$ implies
\begin{align}
\label{eq:prop:SQCleqQC:zS}
\|\zvec_{S}\|_q \leq \rho s^{1/q-1}\|\zvec_{\overline{S}}\|_1 + \tau s^{1/q-1/2}\|\Amat\zvec\|
 \leq \rho s^{1/q-1}\|\zvec\|_1 + \tau s^{1/q-1/2}\|\evec\|.
\end{align}
Therefore, considering the splitting $\zvec = \zvec_{S} + \zvec_{\overline{S}}$ and combining inequalities \eqref{eq:prop:SQCleqQC:norm1z}, \eqref{eq:prop:SQCleqQC:zSbar}, and \eqref{eq:prop:SQCleqQC:zS} yields
\begin{align}
\|\zvec\|_q  \leq \|\zvec_{S}\|_q + \|\zvec_{\overline{S}}\|_q
\leq (\rho +1)s^{1/q-1}\|\zvec\|_1 + \tau s^{1/q-1/2}\|\evec\|
\leq s^{1/q-1/2}[(\rho +1)\mathcal{Q}_s(\Amat)_1 + \tau]\|\evec\|.
\end{align}
Finally, using the latter inequality and \eqref{eq:prop:SQCleqQC:norm1z} again, we have
\begin{equation}
s^{1/2-1/q}\|\zvec\|_q + s^{-1/2}\|\zvec\|_1
\leq [(\rho +2)\mathcal{Q}_s(\Amat)_1 + \tau]\|\evec\|,
\end{equation}
which implies \eqref{eq:Q1qleqQ1} since  $\evec\in\CC^m$ was arbitrary.
\end{proof}

\subsubsection{Quotient properties}
\label{sec:quotientproperties}

In order to put the concepts of quotient and of simultaneous quotient in context, we discuss the notions of quotient property and of simultaneous quotient property. In particular, after reviewing various versions of quotient property published in the literature, we will establish some rigorous relation between quotients and quotient properties.
 
We start by introducing the original definition of quotient property given by Wojtaszczyk \cite{Wojtaszczyk2010} and its generalization due to Foucart \cite{Foucart2014}.
\begin{defn}[Quotient property of Wojtaszczyk]
We say that $\Amat \in \CC^{m \times N}$ satisfies the $\ell^1$-quotient property with constant $\alpha > 0$ if  $\Amat B_1^N  \supseteq \alpha B_2^m$, where $B_p^n$ denotes the $\ell^p$-unit ball in $\CC^n$.
\end{defn}
\begin{defn}[Quotient property of Foucart]
Given $p \geq 1$, the matrix $\Amat \in \CC^{m\times N}$ is said to satisfy the $\ell^q$-quotient property with
constant $d$ relative to a norm $\|\cdot\|$ on $\CC^m$ if
$$
\forall \evec \in \CC^m,\; \exists \zvec \in \CC^N \text{ such that }\Amat \zvec = \evec \text{ and } \|\zvec\|_q\leq d s_*^{1/q-1/2}\|\evec\|,
$$
where
\begin{equation}
\label{eq:def_s*}
s_* := \frac{m}{\ln(eN/m)}.
\end{equation}
\end{defn}

These two quotient properties are actually equivalent and the constants $\alpha$ and $d$ employed in their definitions are strictly related to the quotient $\mathcal{Q}_\lambda(\Amat)_1$. 

\begin{prop}
\label{prop:QPequiv}
For any $\Amat \in \CC^{m \times N}$, the following are equivalent:
\begin{enumerate}[label=(\roman*)]
\item $\Amat$ satisfies the $\ell^1$-quotient property of Foucart with constant $d$ and relative to $\|\cdot\|_2$; 
\item $\Amat$ satisfies the $\ell^1$-quotient property of Wojtaszczyk with $\alpha = 1/(ds_*)$ and $s_*$ defined as in \eqref{eq:def_s*}.
\end{enumerate}
Moreover, $\Amat$ satisfies the $\ell^q$-quotient property of Foucart with constant $d$ and relative to $\|\cdot\|$ if and only if  $\mathcal{Q}_{s_*}(\Amat)_q \leq d$, where $\mathcal{Q}_{s_*}(\Amat)_q$ is the quotient defined in \eqref{eq:def_lq_quotient_0}.
\end{prop}

\begin{proof}
If (i) holds, then for every $\evec \in \alpha B^m_2$, there exists a $\zvec\in\CC^N$ such that $\Amat \zvec = \evec$ and $\|\zvec\|_1 \leq d s_{*}\alpha$. Since $\alpha = 1/(ds_*)$, we get (ii). 

Conversely, assume (ii). Then, for every $\evec\in\CC^m$, define $\hat\evec = \alpha \evec/\|\evec\|_2 \in \alpha B^m_2$. Thanks to (ii), there exists a $\hat\zvec\in B^N_1$ such that $\Amat \hat\zvec = \hat\evec$, i.e., such that $\Amat (\hat\zvec\|\evec\|_2/\alpha) = \evec$. We conclude by observing that $\zvec = \hat\zvec\|\evec\|_2/\alpha$ satisfies $\|\zvec\|_1 \leq \|\evec\|_2/\alpha$.
 
Let us now deal with the second part of the statement. On the one hand, recalling \eqref{eq:def_lq-quotient} and Remark~\ref{rmrk:supmin_lq_quotient}, we see that
$$
\mathcal{Q}_{s_*}(\Amat)_{q} = 
\sup_{\evec \in \CC^m \setminus \{\Ovec\}}
\min_{\substack{\zvec \in \CC^N\\ \Amat \zvec = \evec}}
\frac{s_*^{1/2-1/q} \|\zvec\|_q}{\|\evec\|},
$$
implies the property
$$
\forall \evec \in \CC^m \setminus\{\Ovec\},\; \exists \zvec \in \CC^N \text{ such that }\Amat \zvec = \evec \text{ and } \|\zvec\|_q\leq \mathcal{Q}_{s_*}(\Amat)_{q} s_*^{1/q-1/2}\|\evec\|.
$$
Note that the case $\evec = \Ovec$ is trivial. On the other hand, if 
$$
\forall \evec \in \CC^m ,\; \exists \zvec \in \CC^N \text{ such that }\Amat \zvec = \evec \text{ and } \|\zvec\|_q\leq d s_*^{1/q-1/2}\|\evec\|,
$$
then we have
$$
\sup_{\evec \in \CC^m \setminus \{\Ovec\}}
\min_{\substack{\zvec \in \CC^N\\ \Amat \zvec = \evec}}
\frac{s_*^{1/2-1/q} \|\zvec\|_q}{\|\evec\|} \leq d,
$$
which reads $\mathcal{Q}_{s_*}(\Amat)_{q} \leq d$.
\end{proof}

A variation on the idea of $\ell^q$-quotient property, is the simultaneous $(\ell^q,\ell^1)$-quotient property. We define this property and the corresponding simultaneous quotient. This idea was originally used in the works of Wojtaszczyk, but its formalization is due to Foucart
\begin{defn} [Simultaneous quotient property]
Given $q \geq 1$, a matrix $\Amat \in \CC^{m \times N}$  has the \emph{simultaneous $(\ell^q,\ell^1)$-quotient property} with constants $d$ and $d'$ relative to a norm $\|\cdot\|$ on $\CC^m$  if,
\begin{equation}
\forall \evec \in \CC^m, \; \exists \zvec \in \CC^N \text{ such that } \Amat\zvec = \evec \text{ and } 
\begin{cases}
\|\zvec\|_q \leq d s_*^{1/q - 1/2} \|\evec\|,\\
\|\zvec\|_1 \leq d' s_*^{1/2} \|\evec\|,
\end{cases}
\end{equation}
where $s_*$ is defined as in \eqref{eq:def_s*}.
\end{defn}

We clarify the relation between the simultaneous $(\ell^q,\ell^1)$-quotient property and  $\mathcal{Q}_{\lambda}(\Amat)_{q,1}$ in the following result, whose proof is similar to that of Proposition~\ref{prop:QPequiv} and is therefore left to the reader. 
\begin{prop}
Fix $\Amat\in\CC^{m \times N}$ and define $s_*$ as in \eqref{eq:def_s*}. Then, the following facts hold:
\begin{enumerate}[label=(\roman*)]
\item If $\Amat$ satisfies the simultaneous $(\ell^q,\ell^1)$-quotient property with constants $d$ and $d'$ relative to $\|\cdot\|$, then $\mathcal{Q}_{s_*}(\Amat)_{q,1} \leq d +d'$.
\item $\Amat$ satisfies the simultaneous $(\ell^q,\ell^1)$-quotient property with constants $d=d'=\mathcal{Q}_{s_*}(\Amat)_{q,1}$ relative to $\|\cdot\|$ (notice that, as a limit case,  we can have $d=d'=+\infty$).
\end{enumerate}
\end{prop}

\subsection{Robust instance optimality}
\label{sec:robinstopt}

Equipped with the notions of quotient and simultaneous quotients introduced in Section~\ref{sec:QP}, in this section we provide robust recovery results under unknown error. The roadmap is the following. First, after recalling the definition of instance optimality  for an encoder-decoder pair (Definition~\ref{def:instopt}), we generalize it by introducing the concept of robust instance optimality (Definition~\ref{def:robIO}). Then, we prove that robust instance optimality implies robust recovery error estimates under unknown error depending on the simultaneous quotient $\mathcal{Q}_{s}(\Amat)_{q,1}$ (Theorem~\ref{thm:robIO+QP=>RobRec}). Finally, taking advantage of Proposition~\ref{prop:SQCleqQC}, we prove robust recovery error estimates under unknown error for the QCBP decoder depending on the quotient $\mathcal{Q}_s(\Amat)_1$, based on the robust null space property (Theorem~\ref{thm:NSP->robrec}) and on the restricted isometry property (Corollary~\ref{cor:RIP->robrec}). 

We start by recalling the notion of instance optimality for an encoder-decoder  pair, introduced in the context of CS in 2008 by Cohen, Dahmen, and DeVore \cite{Cohen2008}. 
\begin{defn}[Instance optimality]
\label{def:instopt}
Given a matrix $\Amat \in \CC^{m \times N}$ and a decoder $\Delta:\CC^m \to \mathcal{P}(\CC^N)$, the pair $(\Amat,\Delta)$ is called \emph{mixed $(\ell^q,\ell^p)$-instance optimal} of order $s$ with constant $C$ if
\begin{equation}
\|\xvec - \Delta(\Amat \xvec)\|_q \leq \frac{C}{s^{1/p-1/q}} \sigma_s(\xvec)_p, \quad \forall \xvec \in \CC^N.
\end{equation}
\end{defn}
Notice that this definition considers only the noise-free setting, i.e.\ with $\evec = \Ovec$. Moreover, the exponent $1/p-1/q$ of $s$ is optimal (see \cite[Remark 11.3]{Foucart2013}).

In order to take care of the noisy setting, where $\evec \neq \Ovec$, we generalize the notion of mixed $(\ell^q,\ell^p)$-instance optimality by introducing the $\eta$-robust mixed $(\ell^q,\ell^p)$-instance optimality, where a noise of limited norm is admitted in the model, i.e., $\|\evec\|\leq \eta$. 
\begin{defn}[Robust instance optimality]
\label{def:robIO}
Given $\eta \geq 0$, a matrix $\Amat \in \CC^{m \times N}$, and a decoder $\Delta_\eta:\CC^m \to \mathcal{P}(\CC^N)$, the pair $(\Amat,\Delta_\eta)$ is said to be \emph{$\eta$-robustly mixed $(\ell^q,\ell^p)$-instance optimal} of order $s$ with constants $C,D > 0$ relative to the norm $\|\cdot\|$ on $\CC^m$ if
\begin{equation}
\|\xvec - \Delta_\eta(\Amat \xvec + \evec)\|_q \leq \frac{C}{s^{1/p-1/q}}\sigma_s(\xvec)_p  + D s^{\frac1q-\frac12} \eta, \quad \forall \xvec \in \CC^N, \quad \forall\evec\in\CC^m \text{ s.t. } \|\evec\| \leq \eta.
\end{equation}
\end{defn}
Observe that for $\eta = 0$ we recover the previous notion of mixed $(\ell^q,\ell^p)$-instance optimality.

The next theorem is the generalization of a well-known result, that states that the robust null space property for a matrix $\Amat$ is a sufficient condition for the pair $(\Amat,\Delta_\eta)$, where $\Delta_\eta$ is the QCBP decoder defined in \eqref{eq:defQCBPdecoder}, to be $\eta$-robustly mixed instance optimal. The very first version of this result dates back to the dawn of CS, where the null space property was called \emph{cone constraint} \cite{Candes2006stable}. The result stated here generalizes \cite[Theorem 4.22]{Foucart2013} with a general norm $\|\cdot\|$ on $\CC^m$ in place of $\|\cdot\|_2$ and the $\ell^q$-robust null space property in place of the $\ell^2$-robust null space property.

\begin{thm}[Robust null space property $\Rightarrow$ Robust instance optimality]
\label{thm:QCBProbIO}
Suppose that the matrix $\Amat \in \CC^{m\times N}$ satisfies the $\ell^q$-robust null space property of order $s$ with constants $0 < \rho < 1$ and $\tau > 0$ relative to the rescaled norm $s^{1/q-1/2}\|\cdot\|$ on $\CC^m$. Then, the pair $(\Amat,\Delta_\eta)$, where $\Delta_\eta: \CC^m \to \mathcal{P}(\CC^N)$ is the QCBP decoder \eqref{eq:defQCBPdecoder}, is $\eta$-robustly mixed $(\ell^p,\ell^1)$-instance optimal of order $s$ for every $1 \leq p \leq q$, with constants $C=2(1+\rho)^2/(1-\rho)$ and $D=2(3+\rho)\tau/(1-\rho)$ relative to $\|\cdot\|$.
\end{thm}
\begin{proof}
This is a direct consequence of Theorem~\ref{thm:lqrobNSPest} with $\zvec = \widehat{\xvec}(\eta)$. 
\end{proof}
 
Moreover, the following theorem shows that robust instance optimality implies robustness to unknown error up to the quotient $\mathcal{Q}_s(\Amat)_{q,1}$:
 
\begin{thm}[Robust instance optimality $\Rightarrow$ Robustness to unknown error]
\label{thm:robIO+QP=>RobRec}
Let $\eta \geq 0$, $\Amat \in \CC^{m \times N}$, and $\Delta_\eta:\CC^{m} \to \mathcal{P}(\CC^N)$ be a decoder. Fix $s \leq m$. If the pair $(\Amat,\Delta_\eta)$ is $\eta$-robustly $(\ell^q,\ell^1)$-instance optimal of order $s$ with constants $C, D$ and with respect to $\|\cdot\|$. Then, for every $\xvec \in \CC^N$ and $\evec\in\CC^m$, it holds
\begin{equation}
\|\xvec - \Delta_\eta(\Amat \xvec + \evec)\|_q \leq \frac{C\, \sigma_s(\xvec)_1}{s^{1-1/q}}  + s^{\frac1q-\frac12}(D\, \eta + E\,\max\{\|\evec\|-\eta,0\}){} ,
\end{equation}
where 
$$
E = \max\{C,1\}\,\mathcal{Q}_{s}(\Amat)_{q,1}.
$$
\end{thm}
\begin{proof}
Fix $\eta \geq0$ and consider an arbitrary vector $\evec\in\CC^m$. If $\|\evec\| \leq \eta$, the result is proved by definition of $\eta$-robust mixed $(\ell^q,\ell^1)$-instance optimality. Let us therefore assume $\|\evec\| > \eta$. 

Recalling Lemma~\ref{lem:quotientQCBP}, we pick a vector $\zvec\in \CC^N$ such that $\|\Amat \zvec - \evec\|\leq \eta$ and that 
\begin{equation}
\label{eq:optimalz}
\frac{s^{1/2-1/q} \|\zvec\|_q + s^{-1/2}\|\zvec\|_1}{\|\evec\|-\eta} \leq \mathcal{Q}_{s}(\Amat)_{q,1}.
\end{equation}
Notice that if such a vector $\zvec$ does not exists, we have $\mathcal{Q}_{s}(\Amat)_{q,1}=+\infty$ (see Remark~\ref{rmrk:supmin_sim_lql1_quotient}) and, as a consequence, the thesis holds trivially. 

Now, taking advantage of the $\eta$-robust mixed $(\ell^q,\ell^1)$-instance optimality property of order $s$, for any $\xvec\in\CC^N$ we have the chain of inequalities
\begin{align}
\|\xvec - \Delta_\eta(\Amat \xvec + \evec)\|_q 
& =  \|\xvec + \zvec - \Delta_\eta(\Amat \xvec +\Amat\zvec + \evec-\Amat\zvec) - \zvec\|_q \\
& \leq \|(\xvec +\zvec) - \Delta_\eta(\Amat (\xvec + \zvec) + (\evec - \Amat\zvec))\|_q + \|\zvec\|_q\\
& \leq \frac{C }{s^{1-1/q}} \sigma_s(\xvec + \zvec)_1 + D s^{\frac1q-\frac12} \eta + \|\zvec\|_q \\
& \leq \frac{C }{s^{1-1/q}}\sigma_s(\xvec)_1  + D s^{\frac1q-\frac12} \eta + C\frac{\|\zvec\|_1}{s^{1-1/q}}  + \|\zvec\|_q.
\end{align}
Note that we used the inequality $\sigma_s(\xvec + \zvec)_1 \leq \sigma_s(\xvec)_1 + \|\zvec\|_1$ in the last step. To conclude, we observe that
\begin{equation}
C\frac{\|\zvec\|_1}{s^{1-1/q}}  + \|\zvec\|_q
\leq \max\{C,1\} \frac{s^{1/2-1/q} \|\zvec\|_q + s^{-1/2}\|\zvec\|_1 }{s^{1/2-1/q}(\|\evec\|-\eta)} (\|\evec\|-\eta),
\end{equation}
and exploit inequality \eqref{eq:optimalz}. 
\end{proof}

As a result, the robust null space property for $\Amat$ is a sufficient condition to have robustness to unknown error up to $\mathcal{Q}_s(\Amat)_{1}$. 
\begin{thm}[Robust null space property $\Rightarrow$ Robustness to unknown error of QCBP]
\label{thm:NSP->robrec}
Consider a matrix $\Amat \in \CC^{m \times N}$ that satisfies the $\ell^q$-robust null space property of order $s$ with constants $0<\rho<1$ and $\tau>0$ relative to the rescaled norm $s^{1/q-1/2}\|\cdot\|$. Then, for every  $\xvec \in \CC^N$, $\evec\in\CC^m$, $\eta \geq 0$, and $1\leq p \leq q$ the following estimate holds 
\begin{equation}
\label{thm:NSP->robrec:eq:estimate}
\|\xvec - \Delta_\eta(\Amat \xvec + \evec)\|_p \leq  \frac{C\, \sigma_s(\xvec)_1}{s^{1-1/p}}  + s^{\frac1p-\frac12}(D \eta + E \max\{\|\evec\| -\eta, 0\}),
\end{equation}
where $\Delta_\eta$ is the QCBP decoder defined in \eqref{eq:defQCBPdecoder} and 
\begin{align}
\label{thm:NSP->robrec:eq:defCDE}
C  =\frac{2(1+\rho)^2}{1-\rho}, \quad
D  =\frac{2(3+\rho)\tau}{1-\rho}, \quad
E  = C\,[(\rho+2)\mathcal{Q}_s(\Amat)_1 + \tau].
\end{align}
\end{thm}
\begin{proof}
Thanks to Theorem~\ref{thm:QCBProbIO}, the pair  $(\Amat,\Delta_\eta)$ is robustly $(\ell^p,\ell^1)$-instance optimal with constants $C=2(1+\rho)^2/(1-\rho)$ and $D=2(3+\rho)\tau/(1-\rho)$ for every $1\leq p \leq q$. The conclusion follows by applying Lemma~\ref{thm:robIO+QP=>RobRec} and by noticing that $\mathcal{Q}_s(\Amat)_{q,1} \leq (\rho+2)\mathcal{Q}_s(\Amat)_1 + \tau$ thanks to Lemma~\ref{prop:SQCleqQC}.
\end{proof}
The estimate \eqref{thm:NSP->robrec:eq:estimate} depends linearly on $\max\{\|\evec\|- \eta,0\}$ and the effect of this term amplified essentially only by the $\ell^1$-quotient $\mathcal{Q}_s(\Amat)_1$, where $s$ is the order the robust null space property that  $\Amat$ satisfies.

As an immediate consequence, we have that the restricted isometry property implies a robust recovery result up to the quotient $\mathcal{Q}_s(\Amat)_1$. Remarkably, in \eqref{thm:NSP->robrec:eq:estimate} the constants $C,D$, which multiply the best $s$-term approximation error and the threshold parameter $\eta$, depend only on $\delta_{2s}(\Amat)$ and $E$, which multiplies the norm of the unknown error,  depends on $\delta_{2s}(\Amat)$ and on $\mathcal{Q}_{s}(\Amat)_1$. This makes the role of the quotient totally transparent.
\begin{cor}[Restricted isometry property $\Rightarrow$ Robustness to unknown error of QCBP]
\label{cor:RIP->robrec}
Assume that the $(2s)^{th}$ restricted isometry constant  of $\Amat \in \CC^{m \times N}$ satisfies 
$$
\delta_{2s}(\Amat)<\frac{4}{\sqrt{41}}.
$$ 
Then, for every  $\xvec \in \CC^N$, $\evec\in\CC^m$,  $\eta \geq 0$, and  $1\leq p \leq 2$ the following estimate holds
\begin{equation}
\|\xvec - \Delta_\eta(\Amat \xvec + \evec)\|_p \leq  \frac{C\, \sigma_s(\xvec)_1}{s^{1-1/p}}  + s^{\frac1p-\frac12}(D \eta + E \max\{\|\evec\| -\eta, 0\}),
\end{equation}
where $C,D,E$ are defined as in \eqref{thm:NSP->robrec:eq:defCDE} and where, in turn, the dependency of $0<\rho<1$ and $\tau>0$ on $\delta_{2s}$ is expressed as in \eqref{thm:RIP2s=>NSP:eq:rhotau}. In particular, the constants $C,D$ depend only on $\delta_{2s}(\Amat)$ and $E$ depends on $\delta_{2s}(\Amat)$ and on $\mathcal{Q}_s(\Amat)_1$.
\end{cor}
\begin{proof}
Thanks to Theorem~\ref{thm:RIP2s=>NSP} the matrix $\Amat$ satisfies the $\ell^2$-robust null space property of order $s$ relative to $\|\cdot\|$ with constants $0<\rho<1$ and $\tau >0$ defined as in \eqref{thm:RIP2s=>NSP:eq:rhotau}. We conclude by applying Theorem~\ref{thm:NSP->robrec} with $q = 2$.
\end{proof}

Using Corollary~\ref{cor:RIP->robrec}, we can produce robust recovery error estimates for QCBP whenever we are able to control $\delta_{2s}(\Amat)$ and $\mathcal{Q}_s(\Amat)_1$ in probability. This will be the strategy followed in the forthcoming sections.

\section{Random Gaussian matrices}
\label{sec:Gauss}

We apply the general theory of Section~\ref{sec:robustness} to the case of random Gaussian matrices. In particular, in order to obtain bounds of the form $\mathcal{Q}_s(\Amat)_1 \lesssim 1$ in probability, we exploit already published results, mainly corresponding to the works \cite{DeVore2009,Wojtaszczyk2010} for Gaussian and subgaussian matrices. Notice that these results can be extended to random Weibull matrices by resorting to the results shown in \cite{Foucart2014}. Our main reference for this section is \cite[Chapter 11]{Foucart2013}.

In this section, we consider $m \times N$ sensing matrices $\Amat$ with real entries defined as
\begin{equation}
\label{eq:defNormGauss}
\Amat = m^{-1/2} \Gmat, 
\end{equation}
where $\Gmat$ is a random Gaussian matrix. Namely, the entries of $\Gmat$ are  independent standard Gaussian random variables $G_{ij}\sim\mathcal{N}(0,1)$. 

We start by recalling an already known result. Namely,  the $\ell^1$-quotient $\mathcal{Q}_{s_*}(\Amat)_1$, with $s_*$ defined as in \eqref{eq:def_s*}, of a random Gaussian matrix $\Amat$ is uniformly bounded from above with high probability.
\begin{lem}[Quotient bound in probability for random Gaussian matrices]
\label{thm:QPGauss}
Let $\Amat$ be an $m \times N$ random Gaussian matrix  defined as in \eqref{eq:defNormGauss}. Then, provided
\begin{equation}
m \leq N/2, 
\end{equation} 
we have
\begin{equation}
\Prob\{\mathcal{Q}_{s_*}(\Amat)_1 \leq 34\} \geq 1 - \exp(-m/100).
\end{equation}
where $s_* = m/\ln(eN/m)$ and the $\ell^1$-quotient is relative to $\|\cdot\|_2$ on $\CC^m$.
\end{lem}
\begin{proof}
The lemma is a direct consequence of \cite[Theorem 11.19]{Foucart2013}, and Proposition~\ref{prop:QPequiv}.
\end{proof}

The following restricted isometry property result for random Gaussian matrices is a direct consequence of
\cite[Theorem 9.27]{Foucart2013} and of \cite[Remark 9.28]{Foucart2013}.
\begin{lem}[Restricted isometry property for Gaussian matrices]
\label{lem:RIPGauss}
Let $\Amat$ be an $m \times N$ matrix defined as in \eqref{eq:defNormGauss} with $m < N$. There exists a universal constant $c < 81$ such that the following holds. For $\delta, \varepsilon \in (0,1)$, assume that
\begin{equation}
 m \geq c\, \delta^{-2} [s \ln(eN /s) + \ln(2\varepsilon^{-1})]. 
\end{equation}
Then, 
\begin{equation}
\Prob\{\delta_s(\Amat) \leq \delta\} \geq 1-\varepsilon.
\end{equation}
\end{lem}

Combining these two results, we give a robust recovery result for QCBP with random Gaussian measurements.

\begin{thm}
\label{thm:RobustGauss}
Let $s \leq m \leq N/2$ and $\Amat$ be a random Gaussian matrix defined as in \eqref{eq:defNormGauss}. Then, there exist universal constants $c,C,D,E>0$ such that the following holds. Consider a number of measurements 
\begin{equation}
\label{thm:RobustGauss:eq:regime}
m \geq c \, (s  \ln(eN/s)  + \ln(3\varepsilon^{-1})).
\end{equation}
Then, for every $\xvec\in\CC^N$,  $\evec\in\CC^m$, $1 \leq p \leq 2$, and $s\leq s_* = m/\ln(eN/m)$, the  following estimate holds  
\begin{equation}
\label{thm:RobustGauss:eq:error}
\|\xvec - \Delta_{\eta}(\Amat\xvec + \evec)\|_p \leq \frac{C \sigma_s(\xvec)_1}{s^{1-1/p}} + \left(\frac{s_*}{c}\right)^{\frac1p-\frac12}(D\eta + E \max\{\|\evec\|_2-\eta,0\}),
\end{equation}
 with probability at least $1- \varepsilon$.  
\end{thm}
\begin{proof}
We define the events $\Omega_{RIP}:=\{\delta_{2s}(\Amat) \leq \delta\}$ and $\Omega_{Q}:= \{\mathcal{Q}_{s_*}(\Amat)_1 \leq 34\}$. 
Observe that condition \eqref{thm:RobustGauss:eq:regime} with $c = 1417$ implies $m \geq c (s \ln(eN/(2s))$, which, in turn, implies the condition
$$
m \geq \left(\frac{2c'\delta^{-2}}{1+c'\delta^{-2}/300}\right) s \ln(eN/(2s)),
$$
with $c' = 81$, $\delta = 0.62 < 4 /\sqrt{41}$. This is, in turn, equivalent to
\begin{equation}
m \geq c' \delta^{-2} [2s \ln(eN/(2s)) + \ln(2/\widetilde\varepsilon)],
\end{equation}
where $\widetilde\varepsilon = 2\exp(-m/300)$. Therefore, thanks to Lemma~\ref{lem:RIPGauss}, we have $\Prob(\Omega_{RIP}) \geq 1-2\exp(-m/300)$. 

Moreover, condition \eqref{thm:RobustGauss:eq:regime} also implies
\begin{equation}
\label{thm:RobustGauss:eq:s<=s*/c}
s \leq \frac{m}{c \ln(eN/s)} \leq \frac{s_*}{c}.
\end{equation}
Therefore, after observing that 
\begin{equation}
\Prob(\Omega_{RIP} \cap \Omega_{Q}) \geq 1-\exp(-m/100)-2\exp(-m/300) \geq 1-3\exp(-m/300),
\end{equation}
we apply Corollary~\ref{cor:RIP->robrec} inside the event $\Omega_{RIP} \cap \Omega_{Q}$ with $s = \lfloor s_*/c \rfloor$, which is a valid choice due to \eqref{thm:RobustGauss:eq:s<=s*/c}. Noticing that the function $s \mapsto\sigma_s(\xvec)_1/s^{1-1/p}$ is nonincreasing and that, inside $\Omega_Q$, we have
$$
\mathcal{Q}_{\lfloor s_*/c \rfloor}(\Amat) = \sqrt{\frac{s_*/c}{\lfloor s_*/c \rfloor}} \sqrt{c}\mathcal{Q}_{s_*}(\Amat) \leq \sqrt{2c}\mathcal{Q}_{s_*}(\Amat) \leq 1811,
$$
we conclude that \eqref{thm:RobustGauss:eq:error} holds with probability at least $1-3\exp(-m/300)$. Finally,  notice that we have $3\exp(-m/300) \leq \varepsilon$ for $m \geq 300 \ln(3 \varepsilon^{-1})$, which is, in turn, ensured by \eqref{thm:RobustGauss:eq:regime} since $c=1147 \geq 300 $.
\end{proof}
\begin{rmrk}
Tracking the constants in the proof of Theorem~\ref{thm:RobustGauss}, and using the fact that $c' < 81$ (see Lemma~\ref{lem:RIPGauss}), we can show some upper bounds for the universal constants, employing relations \eqref{thm:NSP->robrec:eq:defCDE} and \eqref{thm:RIP2s=>NSP:eq:rhotau}. Namely, $c \leq 1417$, $C \leq 517$, $D \leq 1057$, and 
$E \leq 5707478$. These bounds can be further optimized by replacing Lemma~\ref{lem:RIPGauss} with more accurate restricted isometry property results for random Gaussian matrices (see \cite[Theorem 9.27]{Foucart2013}). Moreover, notice that choosing $\delta < 0.62$ in the proof leads to larger values for $c$ and smaller values for $C$, $D$, and $E$. 
\end{rmrk}

It is possible to prove a result analogous to Theorem~\ref{thm:RobustGauss} for subgaussian matrices by replacing $\|\cdot\|_2$ with the norm
\begin{equation}
\|\cdot\|^{(m,N)} = \max\{\|\cdot\|_2,\sqrt{\ln(eN/m)}\|\cdot\|_\infty\}.
\end{equation}
A particular application in this case are random Bernoulli matrices. For further details, we refer the reader to \cite[Chapter 11]{Foucart2013}.

\section{Random matrices with heavy-tailed rows}
\label{sec:heavy-tailed}

While Gaussian random matrices lead to convenient estimates in the noise-blind case, they are largely impractical in applications of CS.
The goal of this section is to produce robust recovery error estimates for QCBP when the sensing matrix $\Amat$ has heavy-tailed rows, taking advantage of the general robustness analysis  based on  quotients carried out in Section~\ref{sec:robustness}. First, in Section~\ref{sec:quotientestimates} we  provide upper and lower bounds for the $\ell^1$-quotient based on the minimum singular value of the matrix $\sqrt{\frac{m}{N}}\Amat^*$. Then, in Section~\ref{sec:devineqSV} we deal with deviation inequalities in probability and in expectation for the singular values of random matrices with heavy-tailed columns. Finally, combining results from Sections~\ref{sec:quotientestimates} and \ref{sec:devineqSV}, we  prove the robustness to unknown error of QCBP for random sampling matrices associated with bounded orthonormal systems (Section~\ref{sec:BOS}), which includes the case of subsampled isometries with random independent samples (Section~\ref{sec:subisoindep}) and in Section~\ref{sec:subBernoulli} we  discuss robustness to unknown error for subsampled isometries randomly generated via Bernoulli selectors and with the random independent subset model.

\subsection{Quotient estimates based on singular value analysis}

\label{sec:quotientestimates}

We first prove a  lemma that links the $\ell^1$-quotient to the minimum singular value of the matrix $\sqrt{\frac{m}{N}}\Amat^*$. In the following, we will arrange the singular values of the $N\times m$ matrix $\Amat^*$ as
\begin{equation}
\sigma_{\max}(\Amat^*) = \sigma_1(\Amat^*) \geq \sigma_2(\Amat^*) \geq \cdots \geq \sigma_m(\Amat^*) = \sigma_{\min}(\Amat^*), 
\end{equation} 
recalling that they are the eigenvalues of the $m \times m$ matrix $(\Amat\Amat^*)^{1/2}$.

\begin{prop}[Quotient bounds]
\label{prop:QminSV}
For any matrix $\Amat\in\CC^{m \times N}$, the following upper bound holds
\begin{equation}
\label{prop:QminSV:eq:bound}
\mathcal{Q}_{\lambda}(\Amat)_1 \leq \frac{1}{\sigma_{\min}(\sqrt{\frac{m}{N}}\Amat^*)} \sqrt{\frac{m}{\lambda}},
\end{equation}
where the quotient is relative to $\|\cdot\|_2$ on $\CC^m$. Moreover, if there exists a constant $K >0$ such that
\begin{equation}
\label{prop:QminSV:eq:maxAij}
\max_{i \in[m], \; j \in [N]}|A_{ij}| \leq  \frac{K}{\sqrt{m}},
\end{equation}
the following lower bound holds
\begin{equation}
\label{prop:QminSV:eq:lowerbound}
\mathcal{Q}_\lambda(\Amat)_1 \geq \frac{1}{K} \sqrt{\frac{m}{\lambda}}.
\end{equation}
\end{prop}
\begin{proof}
We first deal with the upper bound. Notice that if $\Amat$ does not have full rank, then the claim is trivial since the right-hand side of \eqref{prop:QminSV:eq:bound} is equal to $+\infty$. Therefore, we assume $\Amat$ of full rank. In this case, the matrix $\Amat$ admits a right Moore-Penrose pseudoinverse $\Amat^\dag = \Amat^*(\Amat\Amat^*)^{-1}$. Consequently, we produce an upper bound to the $\ell_1$-quotient by considering the ansatz $\zvec = \Amat^\dag\evec$, for every $\evec\in\CC^m$. In particular, after observing that
\begin{equation}
\|\Amat^\dag\evec\|_1 
\leq \sqrt{N}\|\Amat^\dag\|_2\|\evec\|_2 
= \frac{\sqrt{N}\|\evec\|_2}{\sigma_{\min}(\Amat^*)}
= \frac{\sqrt{m}\|\evec\|_2}{\sigma_{\min}(\sqrt{\frac{m}{N}}\Amat^*)},
\end{equation}
and that $\Amat(\Amat^{\dag}\evec) = \evec$, we have the estimate
\begin{equation}
\mathcal{Q}_\lambda(\Amat)_1
= \sup_{\evec\in\CC^m} \inf_{\substack{\zvec\in\CC^N\\\Amat\zvec=\evec}} \frac{\|\zvec\|_1}{\sqrt{\lambda}\|\evec\|_2}
\leq \sup_{\evec\in\CC^m} \frac{\|\Amat^{\dag}\evec\|_1}{\sqrt{\lambda}\|\evec\|_2}.
\end{equation}
Combining the two latter inequalities yields \eqref{prop:QminSV:eq:bound}.

Conversely, assume that the matrix $\Amat\in\CC^{m \times N}$ satisfies \eqref{prop:QminSV:eq:maxAij}. Then, considering a vector $\fvec\in\CC^m$ with $f_1 = 1$, for any $\zvec \in\CC^N$ such that $\Amat\zvec = \fvec$, we have
$$
1 = f_1 = (\Amat \zvec)_1 = \sum_{j = 1}^N A_{ij}z_j \leq \frac{K \|\zvec\|_1}{\sqrt{m}}.
$$
(Notice that if such a vector $\zvec$ does not exist, then, by definition, $\mathcal{Q}_\lambda(\Amat)_1 = +\infty$  and \eqref{prop:QminSV:eq:lowerbound} holds trivially).  As a consequence, the $\ell^1$-quotient can be bounded from below as
$$
\mathcal{Q}_{\lambda}(\Amat)_1 
= \sup_{\substack{\evec \in\CC^m\\ \evec \neq \Ovec}} \inf_{\substack{\zvec \in\CC^N\\\Amat\zvec = \evec}} \frac{\|\zvec\|_1}{\sqrt{\lambda}\|\evec\|_2} 
\geq \inf_{\substack{\zvec \in\CC^N\\\Amat\zvec = \evec_i}} \frac{\|\zvec\|_1}{\sqrt{\lambda}\|\evec_i\|_2} 
\geq \frac1K\sqrt{\frac{m}{\lambda}}.
$$
This completes the proof.
\end{proof}

This result shows that we can produce robust recovery estimates for arbitrary sensing matrices $\Amat$ whenever we are able to estimate $\sigma_{\min}(\sqrt{\frac{m}{N}}\Amat^*)$. In the following section, we will provide a recipe to estimate the probability that $\sigma_{\min}(\sqrt{\frac{m}{N}}\Amat^*) \gtrsim 1$. As a consequence, this will imply that $\mathcal{Q}_s(\Amat)_1 \lesssim \sqrt{m/s}$ in probability.

In the case of random sampling from bounded orthonormal systems, discussed in Section~\ref{sec:BOS}, condition \eqref{prop:QminSV:eq:maxAij} holds and Proposition~\ref{prop:QminSV} shows that  $\mathcal{Q}_s(\Amat)_1$  scales  like $\sqrt{m/s}$ up to a constant factor larger than $1/K$ and smaller than $1/\sigma_{\min}(\sqrt{\frac{m}{N}}\Amat^*)$. In particular, the upper bound $\sqrt{m/s}$ is sharp.

\subsection{Deviation inequalities for singular values}

\label{sec:devineqSV}

In this section, we prove  deviation inequalities in expectation and in probability for the minimum singular value of random matrices having heavy-tailed independent isotropic columns, employing tools and ideas from \cite{Vershynin2012}. First, we recall the definition of an isotropic random vector.
\begin{defn}[Isotropic random vector]
A random vector $\zvec$ is \emph{isotropic} if $\Expe[\zvec\zvec^*]=\Imat$. 
\end{defn}
It is immediate to verify the following property (see \cite[Lemma~5.20]{Vershynin2012}).
\begin{lem}
\label{lem:isotropy}
Let $\zvec,\wvec\in\CC^N$ be independent isotropic random vectors. Then, $\Expe[\|\zvec\|_2^2]=\Expe[|\langle\zvec,\wvec\rangle|^2]=N$.
\end{lem}

In order to produce asymptotic estimates in expectation for the singular values of a random matrix with random independent isotropic rows, we define two parameters called \emph{cross coherence} and  \emph{distortion}. The cross coherence controls the angles between the rows of $\Amat$. It can also be interpreted as an a priori control of the off-diagonal part of the Gram matrix $\frac{m}{N}\Amat\Amat^*$ of $\sqrt{\frac{m}{N}}\Amat^*$.\footnote{The cross coherence is referred to as ``incoherence'' in \cite{Vershynin2012}, but we preferred to change its name in order to avoid possible confusion with other definitions of coherence given in CS.}
\begin{defn}[cross coherence] 
\label{def:cross_coherence}
Given a random matrix $\Amat\in\CC^{m\times N}$ with rows $\avec_1,\ldots,\avec_m$, we define its \emph{cross coherence} as
\begin{equation}
\label{eq:defcrosscoherence}
\mu = \left(\frac{m}{N}\right)^2 \Expe \bigg[\max_{k \in [m]} \sum_{j \in [m]\setminus\{k\}}|\langle \avec_j, \avec_k \rangle|^2\bigg].
\end{equation}
\end{defn}
The distortion parameter quantifies how far the rows of $\Amat$ are from being ``correctly'' normalized, or, equivalently, how far the Gram matrix $\frac{m}{N}\Amat\Amat^*$ of $\sqrt{\frac{m}{N}}\Amat^*$ is from having a unit diagonal.
\begin{defn}[Distortion]
\label{def:distortion}
Given a random matrix $\Amat\in\CC^{m\times N}$ with rows $\avec_1,\ldots,\avec_m$, we define its \emph{distortion} as
\begin{equation}
\label{eq:defdistortion}
\xi = \Expe \left[\max_{k \in [m]} \left|\frac{m}{N}\|\avec_k\|_2^2-1\right|\right].
\end{equation}
\end{defn}
Observe that when $\Amat$ is a randomly subsampled isometry (see Section~\ref{sec:subisoindep}), the  distortion is $\xi=0$.

We provide an upper bound in expectation for the singular values of a random matrix with heavy-tailed columns, based on the cross coherence and on the distortion. This is a generalization of \cite[Theorem 5.62]{Vershynin2012}, where it is assumed that $\xi=0$. For the sake of readability, we postpone the proof of this result to Appendix~\ref{sec:SVheavytailed}.
\begin{thm}[Deviation inequality in expectation]
\label{thm:svheavytailedcols}
Let $\Amat$ be an $m \times N$ matrix ($m \leq N$) whose rescaled rows $\sqrt{m}\avec_k$ are independent isotropic random vectors of $\CC^N$. Then, there exists a universal constant $C>0$ such that
\begin{equation}
\label{eq:svestimate}
\Expe \max_{k \in [m]} \left|\sigma_k (\sqrt{\tfrac{m}{N}}\Amat^*) - 1\right| \leq \xi + C\sqrt{(1 + \xi)\mu \ln m}.
\end{equation}
\end{thm}

In the following, we prove a deviation inequality  in probability for the minimum singular value of a random matrix with heavy-tailed columns, assuming a suitable decay of the distortion parameter with respect to $N$. This particular decay property will be verified for random sampling matrices associated with Chebyshev polynomials in Section~\ref{sec:Chebyshev}. 

\begin{prop}[Deviation inequality in probability]
\label{prop:minsingvalprob}
Let $\Amat\in\CC^{m \times N}$ be a random matrix whose rescaled rows $\sqrt{m}\avec_j$ are independent isotropic random vectors.
Moreover, assume that there exist two constants $D_1,D_2>0$ independent of $N$ and $m$ such that  
\begin{equation}
\label{prop:minsingvalprob:hypodist}
\xi \leq \min\bigg\{D_1 \sqrt{\frac{m^2 \ln(m)}{N}},D_2\bigg\}.
\end{equation}
Then, there exists a  constant $C>0$ depending only on $D_1$ and $D_2$ such that
\begin{equation}
\Prob\bigg\{\sigma_{\min}(\sqrt{\tfrac{m}{N}}\Amat^*) \geq \frac{1}{2}\bigg\} \geq 1- \sqrt{\frac{C m^2 \ln(m)}{N}}.
\end{equation}
In particular, $C = 4 (D_1 + C'\sqrt{1 + D_2})^2$, where $C'$ is the universal constant of Theorem~\ref{thm:svheavytailedcols}. 
\end{prop}
\begin{proof}
Thanks to the isotropy of the rescaled rows $\sqrt{m}\avec_j$, Lemma~\ref{lem:isotropy} yields the following equality  
\begin{equation}
m^2 \Expe[|\langle \avec_j, \avec_k\rangle|^2]
= \Expe[|\langle \sqrt{m}\avec_j, \sqrt{m}\avec_k\rangle|^2] 
= \Expe[\|\sqrt{m}\avec_j\|_2^2] 
= N.
\end{equation}
As a consequence, the cross coherence $\mu$ of $\Amat^*$ satisfies
\begin{align}
\label{prop:minsingvalprob:eq:cross-cohe_UB}
\mu  &
= \left(\frac{m}{N}\right)^2\Expe\bigg[\max_{j\in[m]} \sum_{k \in[m]\setminus\{j\}} |\langle \avec_j, \avec_k\rangle|^2\bigg] 
\leq \left(\frac{m}{N}\right)^2 \Expe\bigg[\sum_{j\in[m]} \sum_{k \in[m]\setminus\{j\}} |\langle \avec_j, \avec_k\rangle|^2\bigg] \\
& = \frac{1}{N^2} \sum_{j\in[m]} \sum_{k \in[m]\setminus\{j\}} m^2\Expe[|\langle \avec_j, \avec_k\rangle|^2]
\leq \frac{m^2}{N}.
\end{align}
Now, using this estimate and applying Theorem~\ref{thm:svheavytailedcols}, there exists a universal constant $C'>0$ such that
\begin{equation}
\Expe[|\sigma_{\min}(\tfrac{1}{\sqrt{N}}\Amat^*) - \frac{1}{\sqrt{m}}|] 
= \frac{1}{\sqrt{m}}
\Expe[|\sigma_{\min}(\sqrt{\tfrac{m}{N}}\Amat^*) - 1|]
\leq \frac{\xi}{\sqrt{m}} + C' \sqrt{\frac{(1+\xi)m\ln(m)}{N}}.
\end{equation}
Using the upper bound for $\xi$ assumed in \eqref{prop:minsingvalprob:hypodist}, we obtain
\begin{equation}
\Expe[|\sigma_{\min}(\tfrac{1}{\sqrt{N}}\Amat^*) - \frac{1}{\sqrt{m}}|] 
\leq C'' \sqrt{\frac{m\ln(m)}{N}},
\end{equation}
with $C'' = D_1 + C'\sqrt{1 + D_2}$.

Then, applying the Markov inequality for $t >0$ yields
\begin{equation}
\Prob\bigg\{|\sigma_{\min}(\tfrac{1}{\sqrt{N}}\Amat^*) - \frac{1}{\sqrt{m}}| \geq t \bigg\}
\leq \frac{C''}{t} \sqrt{\frac{m\ln(m)}{N}},
\end{equation}
and the choice $t = 1/(2\sqrt{m})$ gives the estimate
\begin{align}
\Prob\bigg\{\sigma_{\min}(\sqrt{\tfrac{m}{N}}\Amat^*) \geq \frac{1}{2}\bigg\} 
 & = 1- \Prob\bigg\{\sigma_{\min}(\tfrac{1}{\sqrt{N}}\Amat^*) \leq \frac{1}{2\sqrt{m}}\bigg\} \\
 & \geq 1- \Prob\bigg\{|\sigma_{\min}(\tfrac{1}{\sqrt{N}}\Amat^*) - \frac{1}{\sqrt{m}}| \geq \frac{1}{2\sqrt{m}}\bigg\}\\
 & \geq 1- 2C'' \sqrt{\frac{m^2 \ln(m)}{N}}.
\end{align}
The proof is completed by setting $C:=(2C'')^2$.
\end{proof}

\begin{rmrk}[Distortion decay property]
If we replace hypothesis \eqref{prop:minsingvalprob:hypodist} with the more general assumption
$$
\xi \leq \min\left\{D_1 \sqrt{\mu \ln(m)}, D_2\right\},
$$
following the same steps in the proof of Proposition~\ref{prop:minsingvalprob}, we obtain  
$$
\Prob\left\{\sigma_{\min}(\sqrt{\tfrac{m}{N}}\Amat^*) \geq \frac{1}{2}\right\} \geq 1- C \sqrt{\mu\ln(m)}.
$$
This would lead to an improved deviation inequality when a cross coherence upper bound better than $\mu \leq  m^2/N$ is available. Numerical experiments suggest that improving this estimate is possible only for large values of $m$ (see Figure~\ref{fig:cross_coherence_sharp}).
\end{rmrk}

\subsection{Random sampling from bounded orthonormal systems}

\label{sec:BOS}

In this section, we prove robustness to unknown error of QCBP for random sampling matrices associated with bounded orthonormal systems (BOSs).
We start by recalling the definition of BOS and of random sampling matrix associated with a BOS. This is a wide class of structured random matrices containing, for example, the random partial Fourier transform, subsampled isometries, random sampling from orthogonal polynomials. For further details we refer the reader to \cite{Rauhut2010} and \cite[Chapter 12]{Foucart2013}.

\begin{defn}[Random sampling from a bounded orthonormal system]
\label{def:BOS}
Let $\mathcal{D} \subseteq \RR^d$ be endowed with probability measure $\nu$. Then, a set $\Phi=\{\phi_1,\ldots,\phi_N\}$ of complex-valued functions on $\mathcal{D}$  is called a \emph{Bounded Orthonormal System} (BOS) with constant $K$ if it satisfies 
\begin{equation}
\int_{\mathcal{D}} \phi_j(\tau) \overline{\phi_k(\tau)} d \nu(\tau) = \delta_{jk}, \quad \forall j,k \in [N],
\end{equation}
and if
\begin{equation}
\|\phi_j\|_\infty := \sup|\phi_j(\tau)| \leq K, \quad \forall j \in [N].
\end{equation}
Moreover, given $m$ independent random variables $\tau_1,\ldots,\tau_m$, distributed according to $\nu$, we define the random sampling matrix $\Amat\in\CC^{m \times N}$ associated with the BOS $\Phi$ as
\begin{equation}
\label{eq:BOSmatrix}
A_{ij} := m^{-1/2}\phi_j(\tau_i), \quad \forall i\in[m],\; j\in[N].
\end{equation}
\end{defn}

Assuming a linear dependence between $m$ and $s$ and an inverse proportionality between $m$ and $\delta^2$  (up to logarithmic factors), the $s^{th}$ restricted isometry constant is bounded from above by $\delta$ with overwhelmingly high probability. This results corresponds to  \cite[Theorem 2.2]{Chkifa2016}, but it admits analogous versions with different polylogarithmic factors $\mathcal{L}(N,s,\varepsilon,\delta,K)$. In the following, by polylogarithm we mean a polynomial of logarithms and of logarithm of logarithms of the variables.\footnote{We decided to use the notation  $\mathcal{L}(N,s,\varepsilon,\delta,K)$ in order to make our results independent on the particular polylogarithmic factors associated with the minimum number of measurements required in the various restricted isometry property results published in the literature (see also Remark~\ref{rmrk:polylog}).}

\begin{thm}[Restricted isometry property for a BOS]
\label{thm:RIPforBOS}
Let $\Amat \in \CC^{m\times N}$ be the random sampling matrix \eqref{eq:BOSmatrix} associated  with a BOS with constant $K \geq 1$ and let $\varepsilon, \delta \in (0,1)$. Then, there exists a universal constant $c > 0$ and a function function $\mathcal{L}(N,s,\varepsilon,\delta,K)$ depending at most polylogarithmically on $s$ and $N$ such that, provided
\begin{equation}
m \geq c \; s \; \mathcal{L}(s,N,\varepsilon,\delta,K),
\end{equation}
the $s^{th}$ restricted isometry constant of $\Amat$ satisfies 
\begin{equation}
\Prob\{\delta_s(\Amat) \leq \delta\} \geq 1 - \varepsilon.
\end{equation}
In particular, we can choose 
\begin{equation}
\label{thm:RIPforBOS:eq:L}
\mathcal{L}(N,s,\varepsilon, \delta,K)
=
\frac{K^2}{\delta^2} \ln\left(\frac{K^2 s}{\delta^2}\right) 
\max\left\{\frac{1}{\delta^{4}}\ln\left(\frac{K^2 s}{\delta^2} \ln\left(\frac{K^2 s}{\delta^2}\right)\right) \ln(N), \frac{1}{\delta} \ln\left(\frac{\ln(\frac{K^2 s}{\delta^2})}{\varepsilon\delta} \right)\right\}.
\end{equation}
\end{thm}

\begin{rmrk}[Choice of the polylogarithmic factor]
\label{rmrk:polylog}
Notice that the expression \eqref{thm:RIPforBOS:eq:L} for the polylogarithmic factor $\mathcal{L}(N,s,\varepsilon,\delta,K)$ can be replaced by similar formulas associated with other restricted isometry property results available in the literature. For example,  considering \cite[Theorem 12.32]{Foucart2013} and \cite[Remark 12.33-(b)]{Foucart2013} leads to the choice
\begin{equation}
\label{eq:logfactorHolger}
\mathcal{L}(N,s,\varepsilon,\delta,K) 
=  \frac{K^2}{\delta^{2}}  \max\left\{\ln^2(s)\ln\left(\frac{K^2 s}{\delta^{2}} \ln(N)\right)\ln(N),\ln(\varepsilon^{-1})\right\}.
\end{equation}
Similar options can be considered, based on the recent restricted isometry results in \cite{Bourgain2014,Haviv2017,Rauhut2016}.
\end{rmrk}

We are now in a position to prove robust recovery for BOSs. Assuming a sparsity level $s \lesssim \varepsilon \, \sqrt{N}$ (up to logarithmic factors), a suitable decay of the distortion parameter with respect to  $N$, and a linear scaling between $m$ and $s$ (up to logarithmic factors), Theorem~\ref{thm:BOSrobust} provides a recovery error estimate for QCBP that does not assume any prior knowledge of the error term $\evec$, with probability at least $1-\varepsilon$. It is worth pointing out that the term $\max\{\|\evec\|_2-\eta,0\}$ in the error estimate is only amplified by logarithmic factors.
\begin{thm}[Robustness to unknown error of QCBP for BOSs]
\label{thm:BOSrobust}
Let $\Amat\in\CC^{m \times N}$ be the random sampling matrix \eqref{eq:BOSmatrix} associated with a BOS with constant $K\geq 1$, whose distortion parameter satisfies
\begin{equation}
\label{thm:BOSrobust:eq:xi}
\xi \leq \min\bigg\{D_1 \sqrt{\frac{m^2 \ln(m)}{N}},D_2\bigg\}, 
\end{equation}
for suitable constants $D_1,D_2$, independent of $m$ and $N$. Then, there exist constants $c,d,C,D,E>0$ and a function $\mathcal{L}(N,s,\varepsilon,K)$ depending at most polylogarithmically on $N$ and $s$ such that the following holds. For every $N \in \NN$ and $\varepsilon \in (0,1)$, assume that the sparsity $s$ satisfies
\begin{equation}
\label{thm:BOSrobust:eq:sparsity}
s \leq \frac{\varepsilon \; \sqrt{N}}{c \; \mathcal{L}(N,s,\varepsilon,K) \;  \ln^{\frac12}(N)},
\end{equation}
and consider a number of measurements
\begin{equation}
\label{thm:BOSrobust:eq:m}
m = \lceil d \; s \; \mathcal{L}(N,s,\varepsilon,K) \rceil,
\end{equation}
and let $\Delta_\eta$ be the QCBP decoder defined in  \eqref{eq:defQCBPdecoder}. Then, for every  $\xvec \in \CC^N$, $\evec \in \CC^m$, and $1 \leq p \leq 2$, the following robust error estimate holds                
\begin{equation}
\|\xvec - \Delta_\eta(\Amat\xvec + \evec)\|_p \leq \frac{C\sigma_s(\xvec)_1 }{s^{1-1/p}} + {s^{\frac1p-\frac12}}(D\;  \eta + E \; \mathcal{L}^{\frac12}(N,s,\varepsilon,K)  \max\{\|\evec\|_2 - \eta,0\}),
\end{equation}
with probability at least $1 - \varepsilon$. In particular, the constant $c$ depends on $D_1$ and $D_2$, whereas the constants $d, C, D, E$ are universal. A possible choice for $\mathcal{L}(N,s,\varepsilon,K)$ is given by \eqref{thm:RIPforBOS:eq:L} with $\delta = 1/2$ and $\varepsilon/2$ in place of $\varepsilon$.
\end{thm}
\begin{proof}

First, define the event
$$
\Omega_{RIP}:=\{\delta_{2s}(\Amat) \leq 1/2 \}. 
$$
Then, consider a universal constant $d$ and a function $\mathcal{L}(N,s,\varepsilon,K)$ such that Theorem~\ref{thm:RIPforBOS} applied with $\delta = 1/2$ and with failure probability $\varepsilon/2$ ensures 
$$
\Prob(\Omega_{RIP}) \geq 1- \varepsilon/2.
$$ 
for every $m \geq d \;s \; \mathcal{L}(N,s,\varepsilon,K)$.

In order to apply Proposition~\ref{prop:minsingvalprob}, we need to verify that the rows of $\sqrt{m}\Amat$ are isotropic. Indeed, recalling Definition~\ref{def:BOS}, for every $i\in[m]$, we have
$$
\Expe[(\sqrt{m} \avec_i) (\sqrt{m} \avec_i)^*)_{jk}] 
=  \Expe[\phi_j(\tau_i)\overline{\phi_k(\tau_i)}] 
= \int_{\mathcal{D}} \phi_j(\tau)\overline{\phi_k(\tau)} d \nu(\tau) 
= \delta_{jk}, \quad \forall j,k \in[N].
$$
Hence, defining the event 
$$
\Omega_{SV}:= \bigg\{\sigma_{\min}(\sqrt{\tfrac{m}{N}}\Amat^{*}) \geq \tfrac12\bigg\},
$$
Proposition~\ref{prop:minsingvalprob} ensures that $\Prob(\Omega_{SV})\geq 1 - \sqrt{c'm^2\ln(m) / N}$, where $c'$ depends on $D_1$ and $D_2$.
 Now, employing \eqref{thm:BOSrobust:eq:sparsity} and \eqref{thm:BOSrobust:eq:m}, we obtain
\begin{align}
\sqrt{\frac{c' m^2 \ln(m)}{N}} 
&\leq \sqrt{\frac{c' (d+1)^2 \; s^2\; \mathcal{L}^2(N,s,\varepsilon,K)\ln(N)}{N}}
\leq \frac{\varepsilon}{2},
\end{align}
where $c = 2\sqrt{c'}(d+1)$. Note that $c$ depends on $c'$ and $d$, and, consequently only on $D_1$ and $D_2$.

Finally, inside the event $\Omega_{RIP} \cap \Omega_{SV}$ Corollary~\ref{cor:RIP->robrec} holds with constants $C$, $D$, and $E'$ depending on $\rho$ and $\tau$, which are fixed since $\delta=1/2$. Then, we apply  Lemma~\ref{prop:QminSV} and obtain the estimate
$$
\mathcal{Q}_s(\Amat)_1 \leq \frac{\sqrt{m/s}}{\sigma_{\min}(\sqrt{\frac{m}{N}}\Amat^*)} \leq 2 \sqrt{(d+1) \mathcal{L}(N,s,\varepsilon,K)},
$$
and conclude by letting $E = C[2(\rho +2) \sqrt{d+1} +  \tau\}$.
\end{proof}

In Section~\ref{sec:Chebyshev}, we will discuss the application of Theorem~\ref{thm:BOSrobust} to the case of Chebyshev polynomials, providing an explicit estimate  for the distortion $\xi$ that ensures the validity of \eqref{thm:BOSrobust:eq:xi}.

\subsection{Subsampled isometries with random independent samples}
\label{sec:subisoindep}

In this section, we examine the case of subsampled isometries where the rows of the isometry matrix are chosen by drawing $m$ random independent samples from the set of rows of an $N \times N$ isometry. This allows us to take advantage of the theory for BOSs.

The following robustness result for the QCBP decoder in the case of subsampled isometries improves Theorem~\ref{thm:BOSrobust} in two ways: a factor $\ln^{1/2}(N)$ is removed at the denominator for the constraint on the sparsity level and  the dependence of the sparsity level on the probability of failure is more favorable.

\begin{thm}[Robustness to unknown error of QCBP for subsampled isometries]
\label{thm:subisoindepRobust}
Let $N \in \NN$, $K \geq 1$, and $\Umat\in\CC^{N \times N}$ be a unitary matrix such that
\begin{equation}
\label{thm:subisoindepRobust:eq:Uijbound}
\max_{i,j \in [N]}|U_{ij}| \leq \frac{K}{\sqrt{N}}.
\end{equation} 
Draw $m$ indices $\tau_1, \ldots,\tau_m$ uniformly and independently at random from $[N]$  and consider the resulting subsampled isometry  matrix $\Amat\in\CC^{m \times N}$ defined as
\begin{equation}
\label{thm:subisoindepRobust:eq:Aij}
A_{ij} = \sqrt{\tfrac{N}{m}} \; U_{\tau_i \,j}, \quad \forall i \in [m], \;\forall j \in [N].
\end{equation} Then, fixed $\varepsilon \in (0,1)$, there exist universal constants $c,d, C,D,E$ and a function $\mathcal{L}(N,s,\varepsilon,K)$ depending at most polylogarithmically on $N$ and $s$ such that the following holds. For every $N \in \NN$ and $\varepsilon \in (0,1)$, assume that the sparsity $s$ satisfies
\begin{equation}
\label{thm:subisoindepRobust:sregime}
s \leq  \frac{\ln^{\frac12}(\frac{2}{2-\varepsilon})}{c \; \mathcal{L}(N,s,\varepsilon,K)}  \sqrt{N},
\end{equation}
and consider a number of measurements
\begin{equation}
\label{thm:subisoindepRobust:mregime}
m = \lceil d \, s \, \mathcal{L}(N,s,\varepsilon,K) \rceil.
\end{equation}
Then, for every $\xvec \in \CC^N$, $\evec \in \CC^m$,  $1 \leq p \leq 2$ the following robust error estimate holds
\begin{equation}
\label{thm:subisoindepRobust:error}
\|\xvec-\Delta_\eta(\Amat \xvec + \evec)\|_p 
\leq \frac{C \sigma_s(\xvec)_1}{s^{1-1/p}} + s^{\frac1p-\frac12} (D \, \eta + E \, \mathcal{L}^{\frac12}(N,s,\varepsilon,K) \max\{\|\evec\|_2-\eta,0\}),
\end{equation}
with probability at least $1-\varepsilon$, where $\Delta_\eta$ is the QCBP decoder defined in \eqref{eq:defQCBPdecoder}. A possible choice for $\mathcal{L}(N,s,\varepsilon,K)$ is given by \eqref{thm:RIPforBOS:eq:L} with $\delta = 1/2$ and $\varepsilon/2$ in place of $\varepsilon$.
\end{thm}
\begin{proof}

Similarly to the proof of Theorem~\ref{thm:BOSrobust}, we consider the events $\Omega_{RIP} = \{\delta_{2s}(\Amat) \leq 1/2\}$ and $\Omega_{SV} = \{\sigma_{\min}\left(\sqrt{\tfrac{m}{N}}\Amat^*\right) \geq 1/2\}$. 

Thanks to \eqref{thm:subisoindepRobust:eq:Uijbound}, the columns of the matrix $\sqrt{N} \Umat$ form a BOS with respect to the uniform measure on $\mathcal{D}=[N]$. Thus, according to Definition~\ref{def:BOS}, $\Amat$ is a random sampling matrix associated with a BOS. Thus, Theorem~\ref{thm:RIPforBOS} applied with $\delta = 1/2$ and with failure probability $\varepsilon/2$ ensures the existence of a universal constant $d >0 $ and of a suitable function $\mathcal{L}(N,s,\varepsilon,K)$ such that $\Prob(\Omega_{RIP})\geq 1- \varepsilon/2$ for every $m \geq d \, s \, \mathcal{L}(N,s,\varepsilon,K)$.

Besides, we observe that
\begin{equation}
\Omega_{SV}
= \left\{\sigma_{\min}\left(\sqrt{\tfrac{m}{N}}\Amat^*\right) = 1\right\}
= \left\{\tau_1,\ldots,\tau_m \text{ are distinct}\right\}.
\end{equation}
Therefore, using Taylor expansion, there exist $\zeta_1,\ldots,\zeta_m\in\RR$ such that
\begin{equation}
\Prob(\Omega_{SV}) 
= \prod_{k=1}^{m-1} \left(1-\frac{k}{N}\right)
= \prod_{k=1}^{m-1} \left(\text{e}^{-k/N} + \text{e}^{-\zeta_k} \frac{k^2}{2N^2}\right)
\geq \prod_{k=1}^{m-1} \text{e}^{-k/N} \geq \text{e}^{-m^2/(4N)}.
\end{equation}
In particular, $\Prob(\Omega_{SV})  \geq 1-\varepsilon/2$ whenever $m \leq 2 \sqrt{\ln(\frac{2}{2-\varepsilon})N}$, which is implied by \eqref{thm:subisoindepRobust:sregime} and \eqref{thm:subisoindepRobust:mregime} where $c = (d+1)/2$.

Inside the event $\Omega_{SV}$, notice that the $\ell^1$-quotient can be controlled as
\begin{equation}
\mathcal{Q}_s(\Amat)_1 \leq 2\sqrt{\frac{m}{s}} \leq 2(c+1)^{\frac12} \mathcal{L}^{\frac12}(N,s,\varepsilon,K),
\end{equation}
thanks to Lemma~\ref{prop:QminSV}.
Finally, employing Corollary~\ref{cor:RIP->robrec} inside the event $\Omega_{RIP} \cap \Omega_{SV}$ implies the thesis.

\end{proof}


The proof of Theorem~\ref{thm:subisoindepRobust} highlights that the assumption $s \lesssim \sqrt{N}$ (up to logarithmic factors and discarding the dependence on the failure probability $\varepsilon$) on the sparsity level is due to the need of having repeated rows in $\Amat$ with low probability.

\subsection{Subsampled isometries with Bernoulli selectors and random subset model}
\label{sec:subBernoulli}

We discuss subsampled isometries randomly generated according to  two  random matrix models: the Bernoulli selectors and the uniform random subset model. These two study cases are of particular interest since the resulting sensing matrix $\Amat$ cannot have repeated rows. This implies that 
$$
\Prob\left\{\sigma_{\min}\left(\sqrt{\tfrac{m}{N}}\Amat^*\right) = 1\right\} = 1.
$$
As a consequence, Proposition~\ref{prop:QminSV} allows us to control the $\ell^1$-quotient as follows:
\begin{equation}
\label{eq:subiso:controlQ}
\Prob\left\{\mathcal{Q}_s(\Amat)_1 \leq \sqrt{\frac{m}{s}}\right\} = 1.
\end{equation}
Therefore, we only need to prove that the restricted isometry property (or the robust null space property) holds with high probability and apply Corollary~\ref{cor:RIP->robrec} (or Theorem~\ref{thm:NSP->robrec}, respectively) to obtain robustness to unknown error.

\paragraph{Subsampled isometries via Bernoulli selectors.}
Consider $N$ independent random variables $\delta_1,\ldots,\delta_N$ identically distributed such that  
$$
\Prob\{\delta_i = 1\}=\frac{m}{N} \quad \text{and} \quad \Prob\{\delta_i = 0\}=1-\frac{m}{N},
$$ 
called \emph{Bernoulli selectors}. Then, we define the corresponding  subsampled isometry  as
\begin{equation}
\Amat = \sqrt{\frac{N}{m}} \sum_{i=1}^N \delta_i \evec_i \evec_i^* \Umat,
\end{equation}
where $\evec_1,\ldots,\evec_N$ are the standard unit vectors of $\CC^N$. Notice that $\Amat$ is $N \times N$, but it has $m$ nonzero rows in expectation. In this case, following the same steps of  \cite[Theorem 12.32]{Foucart2013} it is possible to show that $\Amat$ has the restricted isometry property of order $s$ and constant $\delta$ with high probability, provided that
\begin{equation}
\label{eq:subisoBernoulli_m}
m \gtrsim \delta^{-2} N \max_{i,j\in[N]}|U_{ij}|^2 s \ln^2(s) \ln^2(N) .
\end{equation}
Comparing this condition with \eqref{eq:logfactorHolger}, we notice the presence of a suboptimal factor $\ln(N)$ in place of $\ln(s)$. Therefore, assuming that
$$
\max_{i,j\in[N]}|U_{ij}| \leq \frac{K}{\sqrt{N}},
$$
and provided that
$$
m \sim K^2 s \ln^2(s) \ln^2(N),
$$
Corollary~\ref{cor:RIP->robrec} combined with \eqref{eq:subiso:controlQ}, implies that for every $\xvec \in \CC^N$, $\evec\in\CC^m$, and $1 \leq p \leq 2$, the following robust recovery error estimate holds with high probability for the QCBP decoder: 
\begin{equation}
\label{eq:subisoBernoulliRobust}
\|\xvec-\Delta_\eta(\Amat \xvec + \evec)\|_p 
\lesssim \frac{\sigma_s(\xvec)_1}{s^{1-1/p}} + s^{\frac1p-\frac12} (\eta + K\ln(s)\ln(N) \max\{\|\evec\|_2-\eta,0\}).
\end{equation}
It is worth noting that there is no need to add any additional hypothesis on the sparsity level of the form \eqref{thm:subisoindepRobust:sregime}. In particular, the fact that $s$ is independent of the probability of failure of the restricted isometry property implies that the robust recovery error estimate holds with `overwhelmingly' high probability. Moreover, there is no need to require any upper bound on $m$, in contrast to the Gaussian case where $m \leq N/2$.

\paragraph{Subsampled isometries with random subset model.}
Consider the subsampled isometry model 
\begin{equation}
\label{eq:rndsubsetmodel}
\Amat = \sqrt{\tfrac{N}{m}} \Pmat_{T_m} \Umat,
\end{equation}
where $T_m$ is a random subset of $[N]$ having cardinality $m$. In this case, carrying out an analysis analogous to Theorem~\ref{thm:subisoindepRobust}  combined with a conversion argument from the random matrix model \eqref{eq:rndsubsetmodel} to the model \eqref{thm:subisoindepRobust:eq:Aij} with random independent samples (see \cite[Section 12.6]{Foucart2013}) it is possible to prove an error estimate analogous to \eqref{thm:subisoindepRobust:error}  with a suboptimal factor $\sqrt{m}$ multiplying $\eta$ and $\max\{\|\evec\|_2-\eta,0\}$. How to get rid of this suboptimal factor is currently an open issue.

\section{Examples}
\label{sec:examples}

In this section, we apply the robustness theory presented in the paper to some case studies of practical importance in CS. In particular, we consider random Gaussian measurements, the partial discrete Fourier transform, and nonharmonic Fourier measurements in Section~\ref{sec:Fourier} and presents some results for random sampling from Chebyshev polynomials in Section~\ref{sec:Chebyshev}. For further numerical experiments we refer the reader to \cite{Brugiapaglia2017recovery}.

All the numerical experiments presented in this section have been performed using MATLAB 2016b version 9.1 on a MacBook Pro equipped with a 3 GHz Intel Core i7 processor and with 8 GB DDR3 RAM. In all the experiments involving BP and QCBP the convex optimization is performed using the  MATLAB package SPGL1 \cite{spgl1}.

\subsection{Gaussian and Fourier measurements}
\label{sec:Fourier}

We define two notable examples of BOS. Namely, the partial discrete Fourier transform and the nonharmonic Fourier measurements. 

\paragraph{Partial discrete Fourier transform.} The BOS is defined as
\begin{equation}
\phi_j(t) = \exp\left(\frac{2\pi\text{i} (j-1) (t-1)}{N}\right), \quad j \in[N],
\end{equation}
and the measure $\nu(t)$ is the uniform measure over $\mathcal{D} = [N]$.  This is a BOS with constant $K = 1$. The resulting random matrix $\Amat$ is called \emph{partial discrete Fourier transform}. This is also an example of subsampled isometry, as defined in \eqref{thm:subisoindepRobust:eq:Aij}, where $\Umat$ is the classical Fourier matrix, namely
\begin{equation}
U_{kj} = \frac{1}{\sqrt{N}} \exp\left(\frac{2\pi\text{i}(j-1)(k-1)}{N}\right), \quad k,j \in [N]. 
\end{equation}
The distortion parameter defined in \eqref{eq:defdistortion} is $\xi = 0$.

\paragraph{Nonharmonic Fourier measurements.} A further example of BOS is given by
\begin{equation}
\label{eq:defnhFourier}
\phi_j(t) = \exp\left(2\pi\text{i} \left(j-\left\lceil\frac{N}{2}\right\rceil\right) t\right), \quad j \in[N],
\end{equation}
where $\mathcal{D} = [0,1]$ and $\nu(t)$ is the uniform distribution over $\mathcal{D}$. This set forms a BOS with constant $K=1$, but it is not a subsampled isometry. The distortion parameter defined in \eqref{eq:defdistortion} is $\xi = 0$.\\

We now discuss some aspects of the robustness analysis employing these two BOSs and the random Gaussian measurements.

\paragraph{Numerical robustness of BP and QCBP.} The first numerical experiment corresponds to Figure~\ref{fig:m_vs_err}, which we have already discussed briefly in Section~\ref{sec:intro}, where the robustness of BP and QCBP is assessed for different types of sensing matrices. We consider random Gaussian measurements as defined in \eqref{eq:defNormGauss}, the partial discrete Fourier transform, and nonharmonic Fourier measurements. We plot the recovery error $\|\widehat{\xvec}(\eta) - \xvec\|_2$ as a function of $m$ averaged over 25 runs, where $\xvec$ is a randomly generated sparse solution with $s=15$ nonzero entries (the support of $\xvec$ is a random subset of $[N]$ and the entries are generated independently at random according to the normal distribution). In particular, we fix $N = 512$ and consider  $m = \lceil k N/10 \rceil$, with $k \in[10]$. For each run, we generate a random noise of magnitude $\|\evec\|_2 = 10^{-3}$ as $\evec = 10^{-3} \cdot \fvec/\|\fvec\|_2$, where $\fvec$ is a standard Gaussian vector. In the case of BP, we are of course in the regime $\|\evec\|_2 > \eta$ since $\eta = 0$. For QCBP, we set $\eta = 10^{-3}$ in order to have $\|\evec\|_2 \leq \eta$. We notice that for every type of measurement, QCBP is very robust, and achieves a recovery error below the level $\|\evec\|_2$. Also BP is quite robust, but the accuracy slightly deteriorates in the case of Gaussian and nonharmonic Fourier measurements, where the error starts to increase (although in a moderate way) when $m$ gets too large. Interestingly, BP turns out to be extremely robust with the partial discrete Fourier transform, achieving an accuracy below the error magnitude for every value of $m$. 

\paragraph{Sharpness of the cross coherence estimate for BOSs.}
In the proof of Proposition~\ref{prop:minsingvalprob}, we have employed the upper bound \eqref{prop:minsingvalprob:eq:cross-cohe_UB} to the cross coherence parameter $\mu$ defined in \eqref{eq:defcrosscoherence}. Namely, 
$$
\mu \leq \frac{m^2}{N}.
$$ 
In Figure~\ref{fig:cross_coherence_sharp}, we check that that this inequality is sharp in the case of the partial discrete Fourier transform and for small values of $m$, which is the regime considered in this paper and of  interest in CS.  This is showed by plotting the quantity $N \mu$ for different values of $N$ and $m$ and by verifying that its growth is proportional to $m^2$.
\begin{figure}
\centering
\includegraphics[width  = 9cm]{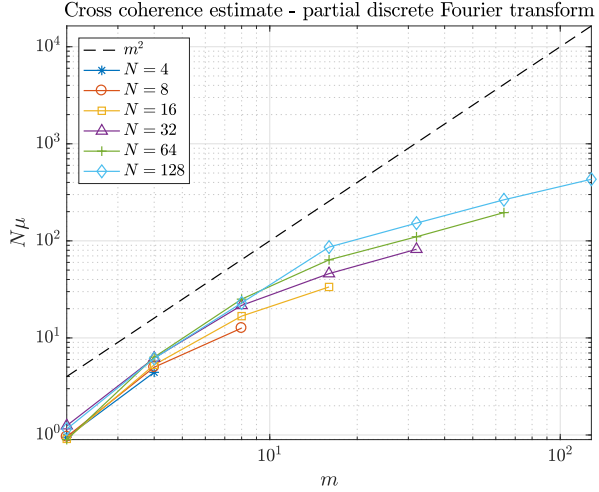}
\caption{\label{fig:cross_coherence_sharp} Logarithmic plot of the quantity $N\mu$, where $\mu$ is the cross coherence parameter defined in \eqref{eq:defcrosscoherence}, as a function of $m$, for various values of $N$ in the case of the subsampled Fourier transform. For each value of $m$ and $N$, $\mu$ is computed by averaging over 500 random trials. The values considered are $N = 4,8,16,32,64,128$ and $m = 2^k$, with $k = 1,\ldots,\log_2(N)$. The quantity $N\mu$  is compared with the upper bound $m^2$, employed in the proof of Proposition~\ref{prop:minsingvalprob} (see \eqref{prop:minsingvalprob:eq:cross-cohe_UB}).}
\end{figure}

\paragraph{Intrinsic limitations of the robustness analysis for BOSs.} This experiment aims at showing numerically the intrinsic limitations of the robustness analysis based on the minimum singular value $\sigma_{\min}(\sqrt{\frac{m}{N}}\Amat^*)$ in the case of bounded orthonormal systems that do not correspond to subsampled isometries. In particular, we consider nonharmonic Fourier measurements and the BP solver ($\eta=0$). In Figure~\ref{fig:mean_err_s_min}, we show the boxplot of the absolute error $\|\Delta_0(\Amat\xvec+\evec) - \xvec\|_2$ and of the quantity $\sigma_{\min}(\sqrt{\frac{m}{N}}\Amat^*)$  as a function of $s$ computed over $100$ random runs, where $\xvec$ is a random $s$-sparse vector (the support is a random subset of $[N]$ of cardinality $s$ and the nonzero entries are normally distributed).
\begin{figure}
\centering
\includegraphics[width = 7.5cm]{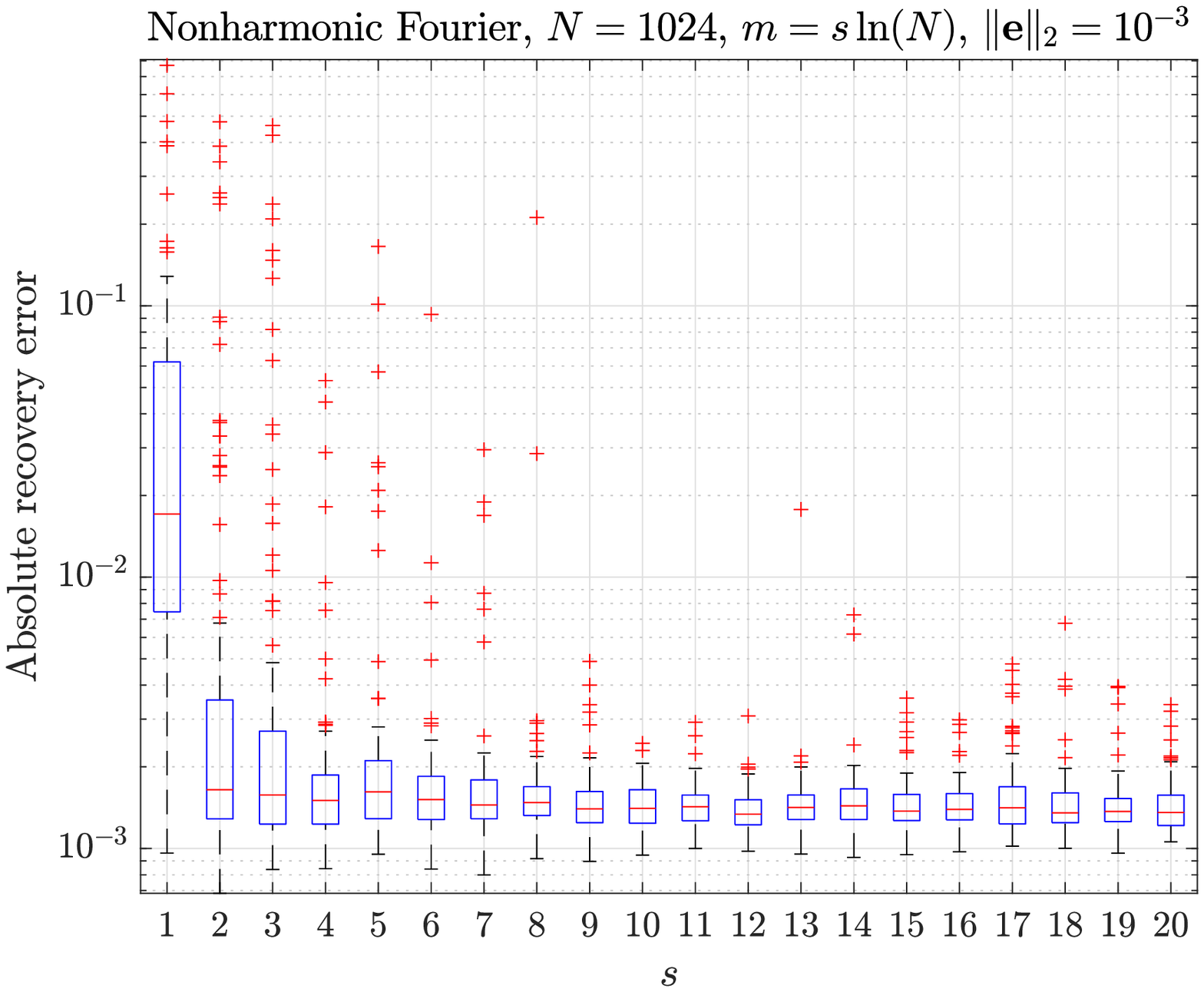}
\includegraphics[width = 7.5cm]{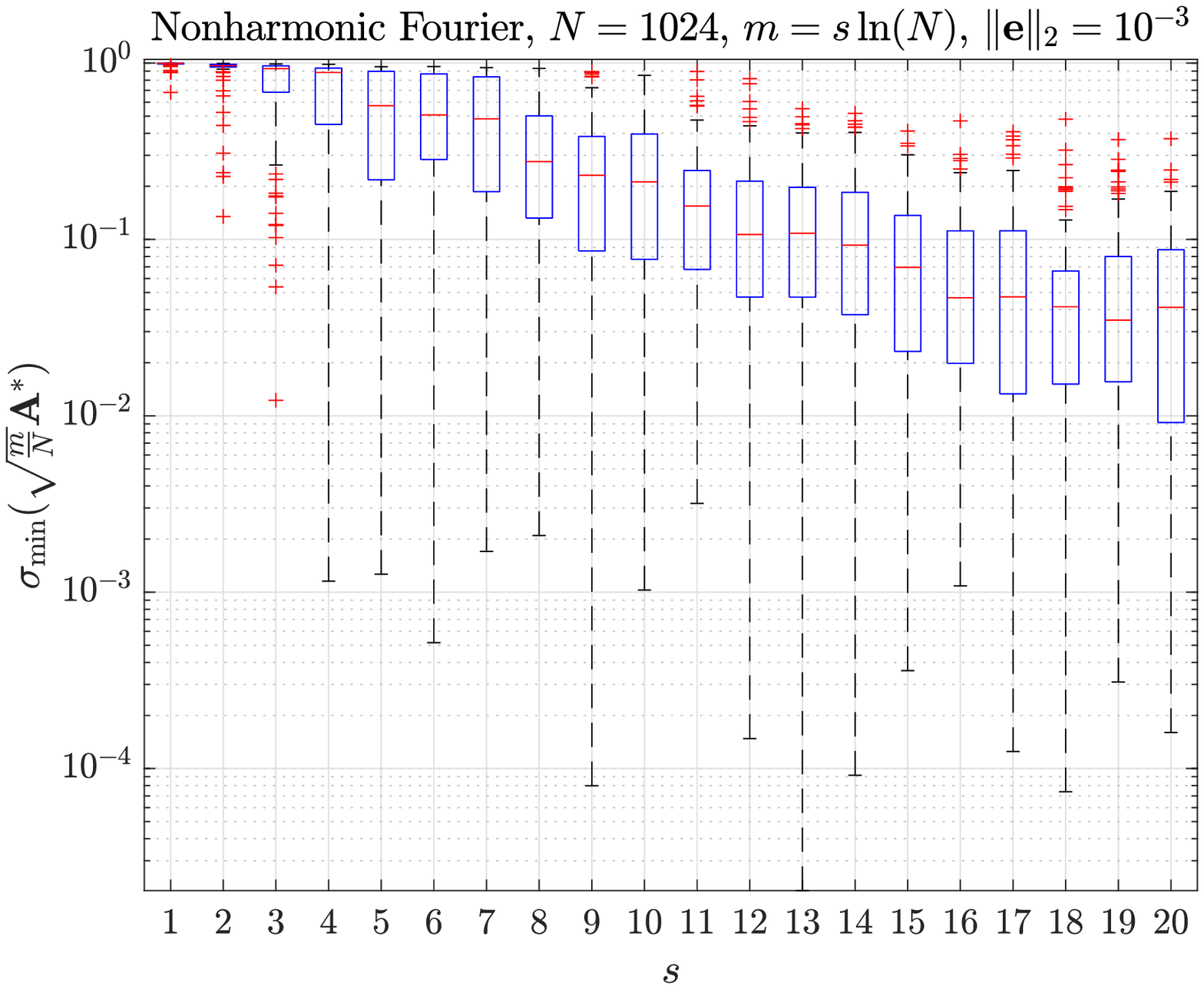}
\caption{\label{fig:mean_err_s_min} Limitations of the error analysis based on the quantity $\sigma_{\min}(\sqrt{\frac{m}{N}}\Amat^*)$ in the case of nonharmonic Fourier measurements. The boxplots show the absolute error $\|\Delta_0(\Amat\xvec +\evec)\|_2$ (left) and the quantity $\sigma_{\min}(\sqrt{\frac{m}{N}}\Amat^*)$ (right) as a function of $s$ computed over 100 random runs, where $\xvec$ is a randomly generated $s$-sparse vector. We fix $N = 1024$ and  $m = \lceil s \ln(N)\rceil$ for every value of $s$. The error is randomly generated with $\|\evec\|_2 = 10^{-3}$. } 
\end{figure}
We fix $N = 1024$ and for each value of $s$ we set $m = \lceil s \ln(N) \rceil$. The additive noise is $\evec = 10^{-3} \cdot \fvec/\|\fvec\|_2$, where $\fvec$ is a standard Gaussian vector, so that $\|\evec\|_2 = 10^{-3}$. We see that the absolute recovery error is quite stable around the level $10^{-3}$ (apart from the $1$-sparse case, where $m = \lceil 1 \ln(1024)\rceil = 7$ measurements are probably too few). On the contrary, the minimum singular value gets smaller and smaller as $s$ gets larger. In the $s$-sparse case, we have $\sigma_s(\xvec)_1 = 0$. Consequently, recalling Corollay~\ref{cor:RIP->robrec} and Proposition~\ref{prop:QminSV}, we have a robust recovery error estimate of the form
$$
\|\Delta_0(\Amat\xvec + \evec) - \yvec\|_2 
\lesssim \mathcal{Q}_s(\Amat)_1 \|\evec\|_2
\lesssim \sqrt{\frac{m}{n}} \frac{1}{\sigma_{\min}(\sqrt{\frac{m}{N}}\Amat^*)} \|\evec\|_2.
$$
Therefore, the fact that the minimum singular value gets smaller and smaller means that the second inequality becomes more and more crude as $s$ gets larger. This shows an intrinsic fundamental limitation of our analysis in the general BOS case. In particular, it suggests that the upper bound to the quotient given in Proposition~\ref{prop:QminSV} can be improved.

\subsection{Chebyshev polynomials}
\label{sec:Chebyshev}

In this section we specialize the robustness error analysis for general BOSs provided by Theorem~\ref{thm:BOSrobust} to the case of random sampling from Chebyshev orthogonal polynomials. The main motivation is the interest that $\ell^1$ minimization has recently received in the context of polynomial approximation \cite{Adcock2017b,Rauhut2016}. A major application is the uncertainty quantification of PDEs with random inputs, where the function to approximate is usually high-dimensional (see \cite{Adcock2017} and the references therein). Here, we limit our discussion to one-dimensional polynomial approximation and refer to reader to \cite{Adcock2017} for numerical experiments regarding the multi-dimensional case. 

We consider the Chebyshev orthogonal polynomials on $\mathcal{D}=[-1,1]$, defined as 
\begin{align}
\label{eq:defCheby1}\phi_1(t) &\equiv 1, \quad \forall t \in \mathcal{D}\\
\label{eq:defCheby2}\phi_{j+1}(t) &:=\sqrt{2}\cos(j \arccos(t)), \quad \forall t \in\mathcal{D}, \quad \forall j \in[N-1].
\end{align}
For every $N\in \NN$, the set $\{\phi_j\}_{j \in[N]}$ is a BOS with respect to the Chebyshev measure 
\begin{equation}
\label{eq:defChebymeas}
\nu(t) = \pi^{-1}(1-t^2)^{-1/2}, \quad\forall t\in\mathcal{D}. 
\end{equation}
Note that in this case, we have $K = \sqrt{2} > 1$.

The analysis is based on a suitable estimate of the distortion $\xi$ defined in \eqref{eq:defdistortion}. Indeed, the following result, whose proof is postponed to Appendix~\ref{sec:Chebyproof} for the sake of compactness, shows that $\xi$ can be bounded from above by $\sqrt{m/N}$ up to a universal constant. The sharpness of this estimate is numerically confirmed in Figure~\ref{fig:dist_sharp}.
\begin{prop}[Distortion bound for Chebyshev polynomials]
\label{prop:distCheby}
For every $1 < m \leq N$, the distortion parameter associated with the BOS of Chebyshev polynomials with respect to the Chebyshev measure satisfies 
\begin{equation}
\label{eq:dist_ub_cheby}
\xi\leq \frac{9 \sqrt{6}}{2\pi^{1/4}} \sqrt{\frac{m}{N}}.
\end{equation}
\end{prop}

An immediate consequence of Theorem~\ref{thm:BOSrobust} and Proposition~\ref{prop:distCheby} is the robustness of the QCBP decoder for random sampling from Chebyshev polynomials.

\begin{thm}[Robustness to unknown error of QCBP for Chebyshev polynomials]
\label{thm:ChebyRobust}
Let $\Amat$ be the random sampling matrix \eqref{eq:BOSmatrix} associated with the BOS of Chebyshev polynomials \eqref{eq:defCheby1}-\eqref{eq:defCheby2} with respect to the Chebyshev measure \eqref{eq:defChebymeas}. Then, there exist universal constants $c,d,C,D,E>0$ such that the following holds. For every $N \in \NN$ and $\varepsilon \in (0,1)$, assume that the sparsity $s$ satisfies
\begin{equation}
\label{thm:ChebyRobust:eq:sparsity}
s \leq \frac{\varepsilon \; \sqrt{N}}{c \; \mathcal{L}(N,s,\varepsilon,K) \;  \ln^{\frac12}(N)},
\end{equation}
and consider a number of measurements
\begin{equation}
\label{thm:ChebyRobust:eq:m}
m = \lceil d \; s \; \mathcal{L}(N,s,\varepsilon,K) \rceil.
\end{equation}
Then, for every  $\xvec \in \CC^N$, $\evec \in \CC^m$, and $1 \leq p \leq 2$, the following robust error estimate holds                
\begin{equation}
\|\xvec - \Delta_\eta(\Amat\xvec + \evec)\|_p \leq \frac{C\sigma_s(\xvec)_1 }{s^{1-1/p}} + {s^{\frac1p-\frac12}}(D\;  \eta + E \; \mathcal{L}^{\frac12}(N,s,\varepsilon,K)  \max\{\|\evec\|_2 - \eta,0\}),
\end{equation}
with probability at least $1 - \varepsilon$, where $\Delta_\eta$ is the QCBP decoder defined in  \eqref{eq:defQCBPdecoder}. A possible choice for $\mathcal{L}(N,s,\varepsilon,K)$ is given by \eqref{thm:RIPforBOS:eq:L} with $\delta = 1/2$ and $\varepsilon/2$ in place of $\varepsilon$.
\end{thm}

We conclude this section with two numerical illustrations.

\paragraph{Sharpness of the distortion estimate.} In Figure~\ref{fig:dist_sharp} we show the sharpness of the upper bound \eqref{eq:dist_ub_cheby} to the distortion parameter given in Proposition~\ref{prop:distCheby}.
\begin{figure}
\centering
\includegraphics[height = 5.5cm]{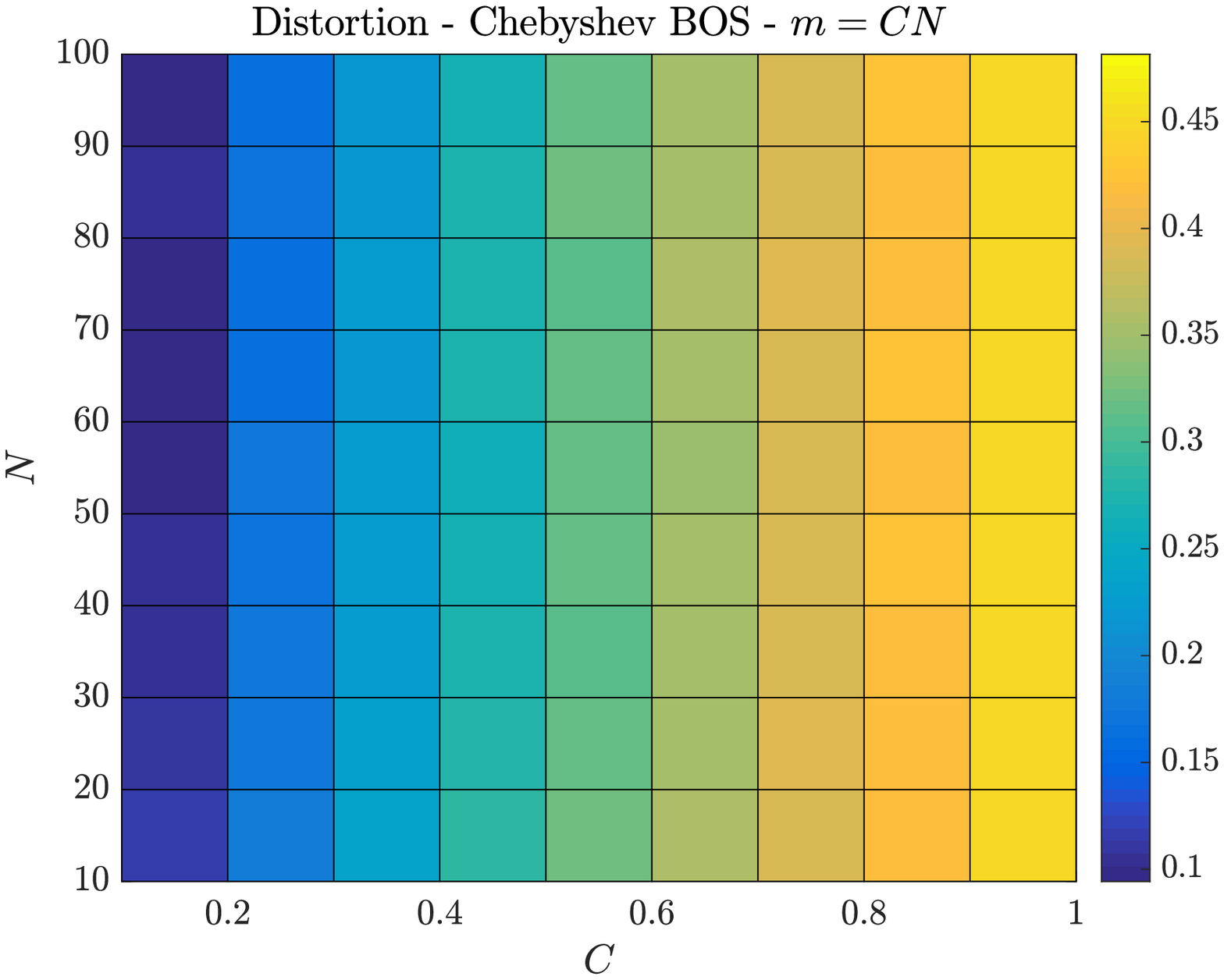}
\includegraphics[height = 5.5cm]{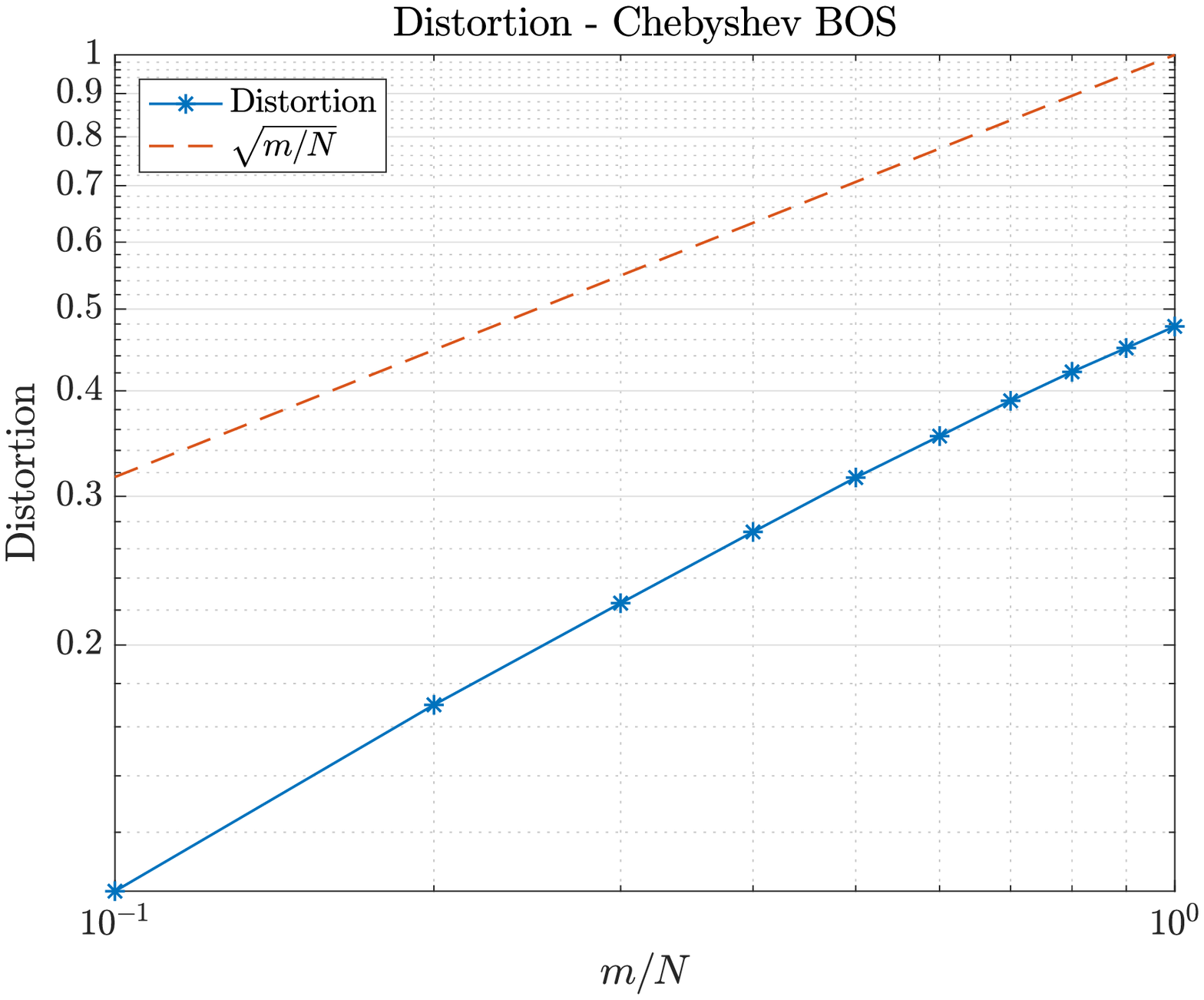}
\caption{\label{fig:dist_sharp}Sharpness of the upper bound \eqref{eq:dist_ub_cheby} to the distortion $\xi$ provided in Proposition~\ref{prop:distCheby}. On the left, we show that the distortion only depends on the ratio $m/N$. On the right, we check that the growth rate is $\sqrt{m/N}$, as predicted by the theory.}
\end{figure}
In particular, we consider $N = 10,20,\ldots,100$ and set $m = C N$ for $C = 0.1,0.2,\ldots,1$. Then, for each value of $N$ and $m$, we approximate the distortion by averaging the quantity $\max_{i \in[m]}|(m/N) \|\avec_i\|_2^2-1|$ over $10000$ trials. In Figure~\ref{fig:dist_sharp} (left), we plot the resulting approximate distortion as a function of $N$ and $C$. We can see that the resulting only depends on $C= m/N$. In particular, it is increasing in $m/N$. In Figure~\ref{fig:dist_sharp} (right), we check that the growth rate is $\sqrt{m/N}$ by considering the average of the approximate distortion over $N$, for each value of $C=m/N$.

\paragraph{Polynomial approximation via $\ell^1$ minimization.} 
Here, we assess the robustness of QCBP when employed in one-dimensional polynomial approximation. Consider the function $f : [-1,1]\to \RR$ defined as
\begin{equation}
\label{eq:ftoapprox}
f(t) = \cos(\pi t) \text{e}^{-t}, \quad \forall t \in [-1,1]. 
\end{equation}
We want to compute a sparse approximation of $f$ with respect to the Chebyshev polynomials from a few noisy pointwise samples. Namely, we want to find 
\begin{equation}
\label{eq:feta}
\widehat{f}_\eta(t) = \sum_{j = 1}^N \widehat{x}_j(\eta) \phi_j(t), \quad \forall t\in [-1,1], 
\end{equation}
where $\widehat\xvec(\eta)=\Delta_{\eta}(\yvec + \evec)$ is the QCBP solution relative to the random sampling matrix $\Amat$ from the Chebyshev system, with measurements $\yvec$ defined as $
y_i = m^{-1/2}f(t_i)$, and where the error is defined as $\evec = E \cdot  \fvec/\|\fvec\|_2$, where $\fvec$ is a normal Gaussian vector, so that $\|\evec\|_2 = E$. 

In Figure~\ref{fig:eta_vs_err} we assess the performance of QCBP by plotting the absolute approximation error $\|f-\widehat{f}\|_{L^2}$ as a function of $\eta$ for different values of the error level $E$ on the measurements.
\begin{figure}
\includegraphics[width = 7cm]{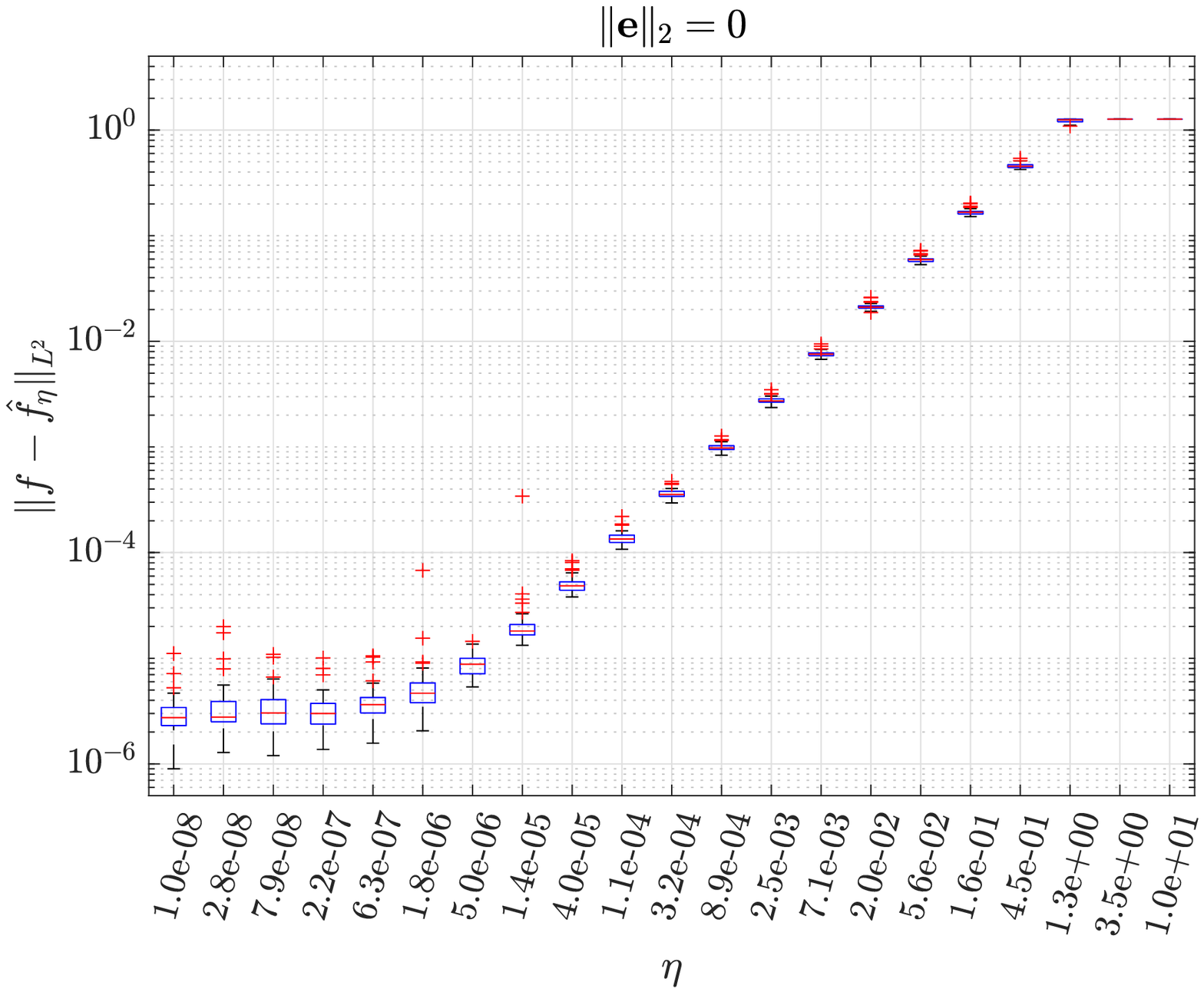}
\includegraphics[width = 7cm]{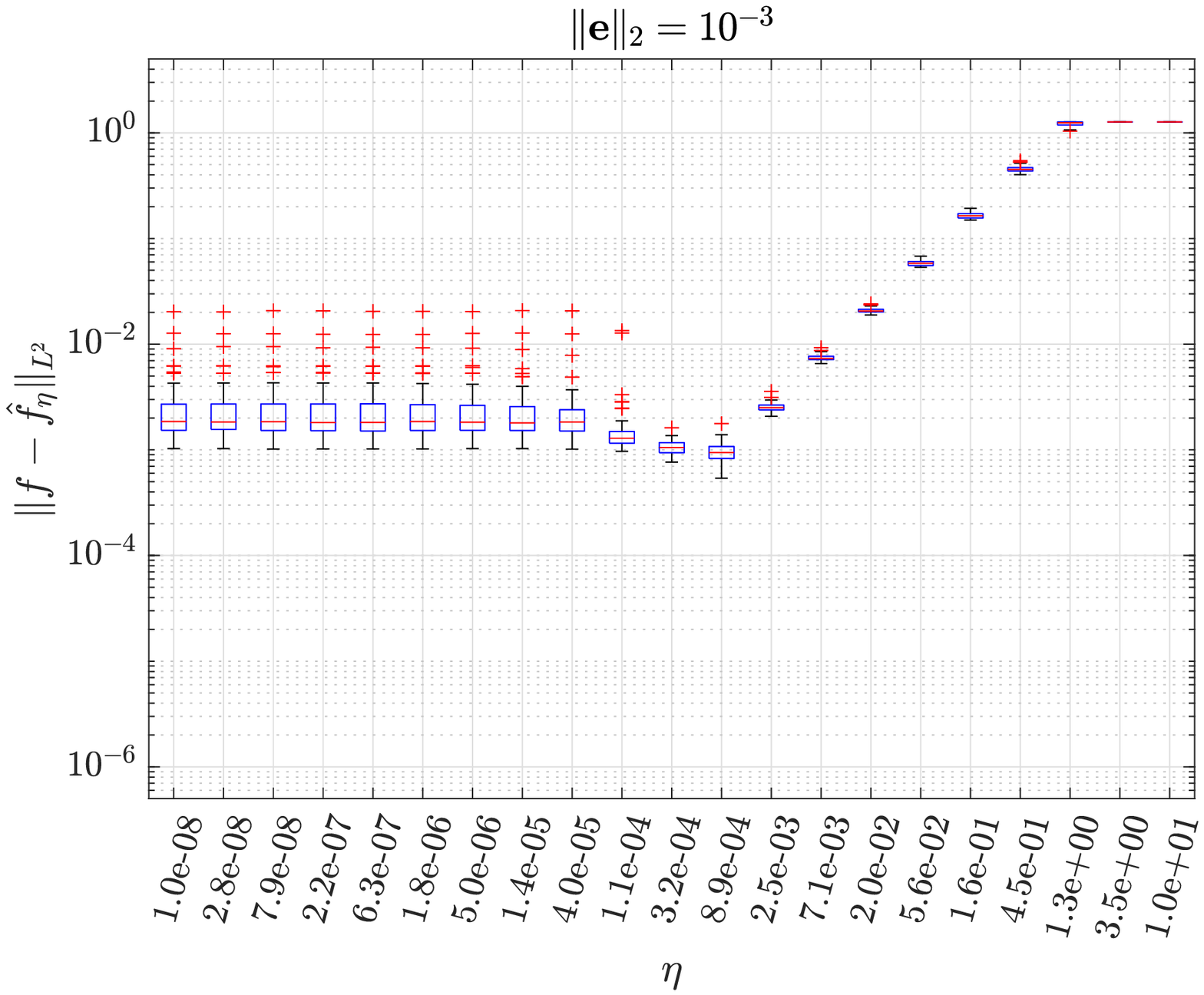}
\caption{\label{fig:eta_vs_err}Numerical assessment for polynomial approximation with Chebyshev polynomials through QCBP from pointwise samples. The figure shows the absolute approximation error $\|f-\widehat{f}_\eta\|_{L^2}$ as a function of $\eta$ for $f$ and $f_\eta$ defined as in \eqref{eq:ftoapprox} and \eqref{eq:feta}, respectively. The boxplots correspond to 50 random runs. In the left figure, we consider accurate pointwise samples. In the right figure, the pointwise samples are subject to an additive noise $\evec$ with $\|\evec\|_2 = 10^{-3}$. We can see that in both cases the choice of $\eta$ affects the quality of the computed approximation.}
\end{figure}
In particular, we consider $E = 0$ and $E=10^{-3}$. It is worth pointing out that when $E=0$, the pointwise samples $f(t_i)$ are exact, but there is still a source of unknown error due to the truncation of the Chebyshev expansion at level $N$ \cite{Adcock2017c,Adcock2017b}. We fix $N = 500$ and consider $m = 100$ pointwise samples. We let $\eta$ range from $10^{-8}$ to $10$ by choosing $\eta = 10^{k}$, with $k$ belonging to a uniform grid of 12 points on $[-8,1]$. The box plots show the results over 50 random trials.\footnote{Notice that the outliers are sometimes aligned in Figure~\ref{fig:eta_vs_err}. This is due to the structure of the proposed numerical experiment: for each randomized choice of $\Amat$,  $\yvec$, and $\evec$ (over 50 trials), all the values of $\eta$ are tested using the same sampling matrix, vector of measurements, and noise vector.}

In both cases, underestimating $\eta$ is better than overestimating it. When $E = 0$, there seems to be an optimal choice of $\eta$ around $10^{-4}$ where the mean and the standard deviation of the approximation error are minimized. However, in both cases, QCBP shows to be quite robust with respect to the choice of $\eta$, that does not seem to worsen the recovery performance significantly when underestimated. Furthermore, this numerical experiment confirms that estimating the noise level accurately and tuning the threshold parameter accordingly improves the recovery performance of QCBP, as empirically observed when $\eta$ is estimated via cross validation \cite{Doostan2011}. This behavior can be justified from a theoretical perspective by the presence of the term $\max\{\|\evec\|_2-\eta,0\}$ in the recovery error estimate of Theorem~\ref{thm:ChebyRobust}.

\section{Conclusions}

We have presented a theoretical analysis of the robustness of the quadratically-constrained basis pursuit decoder employed in compressed sensing when there is no \emph{a priori} information available on the the error corrupting the measurements.  Our  analysis relies on the concepts of robust null space property, restricted isometry property, quotient, simultaneous quotient, and robust instance optimality.  This analysis has been applied to the case of random Gaussian matrices and random matrices with heavy-tailed rows, including random sampling from Bounded Orthonormal Systems (BOSs). The study cases assessed in detail are subsampled isometries, partial discrete Fourier transform, nonharmonic Fourier measurements, and random sampling from Chebyshev polynomials. The main open problems brought to light by our analysis are the following: 
\begin{itemize}
\item The assumption \eqref{thm:BOSrobust:eq:xi} on the distortion in Theorem~\ref{thm:BOSrobust} requires an \emph{ad hoc} estimation of $\xi$ for each particular BOS under exam, as the one carried out in Proposition~\ref{prop:distCheby} for the Chebyshev polynomials. An open problem is to find an upper bound for $\xi$ in the general BOS case.

\item Comparing the robust recovery result for BOSs (Theorem~\ref{thm:BOSrobust}) and for subsampled isometries with independent samples (Theorem~\ref{thm:subisoindepRobust}) it emerges that in the assumption on the sparsity regime there is a suboptimal $\ln^{1/2}(N)$ factor at the denominator (due to the suboptimal factor $\ln(m)$ in Theorem~\ref{thm:svheavytailedcols}) and that linear dependence between $s$ and $\varepsilon$ could be improved to $\ln(\frac{2}{2-\varepsilon})$.

\item Based on the numerical results in Figure~\ref{fig:mean_err_s_min}, the upper bound for the $\ell^1$-quotient based on the minimum singular value of $\sqrt{\frac{m}{N}}\Amat^*$ given in Proposition~\ref{prop:QminSV} seems to not be sharp in the general BOS case, where the minimum singular value could decay too fast as $m \rightarrow N$. Therefore, an improvement of the upper bound for the $\ell^1$-quotient given in Proposition~\ref{prop:QminSV} may be necessary.
\end{itemize}
To conclude, we would like to draw the reader's attention to Figure~\ref{fig:m_vs_err} -- presented in the introduction -- and examine to what extent our theoretical analysis is able to explain the phenomena observed therein. Firstly, for every type of measurement (i.e., random Gaussian measurements, partial discrete Fourier transform, and nonharmonic Fourier measurements) the figure shows how the accuracy of quadratically-constrained basis pursuit improves when the threshold parameter $\eta$ is tuned according to a reliable estimate of the noise level. This behavior is theoretically explained by the presence of the term $\max\{\|\evec\|_2-\eta,0\}$ in all the recovery error estimates provided in our analysis (Section~\ref{sec:robustness}). Secondly, we are able to understand the Gaussian case more deeply, since robustness under unknown error is now proved not only for basis pursuit, but also for quadratically-constrained basis pursuit for any value of the threshold parameter $\eta$ (Section~\ref{sec:Gauss}). Thirdly, our analysis is able to explain -- to a certain extent -- why the partial discrete Fourier transform outperforms random Gaussian measurements. Indeed, considering the subsampled isometry model based on Bernoulli selectors, the discussion carried out in  Section~\ref{sec:subBernoulli} shows that no upper bound on $s$ (nor $m$) is needed to obtain robust recovery, as opposed to the Gaussian case where $m \leq N/2$ and $s \leq m /\ln(eN/m)$. However, the reason why the performance of basis pursuit is similar to that of quadratically-constrained basis pursuit for the partial discrete Fourier transform with noisy measurements still remains an  interesting open question.

\section{Acknowledgements}
The authors gratefully acknowledge Simon Foucart, Michael Friedlander, Yaniv Plan, Paul Tupper, Ozgur Yilmaz  for the insightful discussions about the material presented in this paper. The authors acknowledge the support of the Alfred P.\ Sloan Foundation and the Natural Sciences and Engineering Research Council of Canada through grant 611675. The first author also acknowledges the Postdoctoral Training Center in Stochastics of the Pacific Institute for the Mathematical Sciences for the support.

\appendix
\section{Proof of Theorem~\ref{thm:svheavytailedcols}}
\label{sec:SVheavytailed}

The aim of this section is to prove Theorem~\ref{thm:svheavytailedcols} about the expected singular values of a tall random matrix with heavy-tailed random isotropic columns. The proof of Theorem~\ref{thm:svheavytailedcols} relies on the results from \cite{Vershynin2012}, that is the main reference to keep handy throughout the appendix. In particular, Theorem~\ref{thm:svheavytailedcols} is a generalization to \cite[Theorem 5.62]{Vershynin2012}, which provides bounds for the expected singular values of a tall random matrix with heavy-tailed  random isotropic columns based on the \emph{cross coherence} (referred to as \emph{coherence} in \cite{Vershynin2012}) and defined as in \eqref{eq:defcrosscoherence}, assuming a suitable \emph{normalization hypothesis} on the columns of the random matrix considered. This normalization hypothesis turns out to be too  restrictive for the application to BOSs. Therefore, we relax it by means of the the \emph{distortion} parameter $\xi$ associated with the random matrix, defined in  \eqref{eq:defdistortion}, following the remarks in \cite[Section 5.7, paragraph ``For Section 5.5'']{Vershynin2012}.

The appendix is divided in two subsections. In the first subsection, we provide some technical results about random matrices that will be useful to prove Theorem~\ref{thm:svheavytailedcols} in the second subsection.

\subsection{Preliminary results}

We prove a generalization of \cite[Lemma 5.63]{Vershynin2012} regarding \emph{matrix decoupling}. The proof provided here is analogous to that in \cite{Vershynin2012} and relies on a simpler scalar decoupling result (see \cite[Lemma 5.60]{Vershynin2012}), based on random selectors.
\begin{lem}[Matrix decoupling]
\label{lem:decoupling} 
Let $\Mmat\in \CC^{m \times N}$ be a random matrix. Denote the rows of $\Mmat$ as $\mvec_1,\ldots,\mvec_m$. Then, 
$$
\Expe \|\Mmat\Mmat^*-\Imat\|_2 
\leq \Expe\max_{i \in [m]}|\|\mvec_i\|_2^2-1| + 4 \max_{T \subseteq [m]} \Expe\|\Mmat_T \Mmat_{\overline{T}}^*\|_2,
$$
where $\Mmat_T$ is the restriction of $\Mmat$ to the rows in $T$ and $\overline{T} = [m]\setminus T$.
\end{lem}

\begin{proof}
First we notice that, thanks to the diagonalizability of $\Mmat\Mmat^*$, we have
\begin{align}
\|\Mmat\Mmat^*-\Imat\|_2
& = \sup_{\substack{\xvec\in\CC^m\\ \|\xvec\|_2 = 1}}|\xvec^*(\Mmat\Mmat^*-\Imat)\xvec|
 = \sup_{\substack{\xvec\in\CC^m\\ \|\xvec\|_2 = 1}} |\|\Mmat^* \xvec\|_2^2 -1|.
\label{lem:decoupling:eq:normMM*-I}
\end{align}
Furthermore, for every $\xvec \in \CC^{m}$ with $\|\xvec\|_2 = 1$, we can split
\begin{align}
\|\Mmat^* \xvec\|_2^2 
& = \sum_{i \in[m]} |x_i|^2 \|\mvec_i\|_2^2 
+ \sum_{i \in[m]}\sum_{k \in [m]\setminus\{i\} } x_i \overline{x}_k \langle \mvec_i,\mvec_k\rangle\\
& = 1 + \sum_{i \in[m]} |x_i|^2 (\|\mvec_i\|_2^2-1) 
+ \sum_{i \in[m]}\sum_{k \in [m]\setminus\{i\} } x_i \overline{x}_k \langle \mvec_i,\mvec_k\rangle.
\end{align}
Therefore, we obtain the estimate
\begin{align}
|\|\Mmat^* \xvec\|_2^2 -1| & \leq \max_{i \in [m]}|\|\mvec_i\|_2^2-1| 
+ \bigg|\sum_{i \in[m]}\sum_{k \in [m]\setminus\{i\} } x_i \overline{x}_k \langle \mvec_i,\mvec_k\rangle\bigg|. \label{lem:decoupling:eq:normMx}
\end{align}
In order to deal with the second addendum in the right-hand side of \eqref{lem:decoupling:eq:normMx}, we introduce a random set $T = \{i \in [m]: \theta_i = 1\}$, where $\theta_1,\ldots,\theta_m$ are independent random selectors such that $\Expe[\theta_i]= 1/2$. Moreover, let us denote as $\Expe_\Mmat$ and $\Expe_T$ the expectations with respect to the random matrix $\Mmat$ and to the random set $T$, respectively. Then, observing that $\Expe_T[\theta_i(1-\theta_k)]=1/4$ for every $i\neq k$, we have
\begin{align}
\Expe_T\sum_{i \in T}\sum_{k \in \overline{T}} x_i \overline{x}_k \langle \mvec_i,\mvec_k\rangle
& = \Expe_T\sum_{i \in [m]}\sum_{k \in [m]\setminus\{i\}} \theta_i (1-\theta_k) x_i \overline{x}_k \langle \mvec_i,\mvec_k\rangle\\
& = \frac{1}{4}\sum_{i \in [m]}\sum_{k \in [m]\setminus\{i\}} x_i \overline{x}_k \langle \mvec_i,\mvec_k\rangle. \label{lem:decoupling:eq:ExpeTprod}
\end{align}
Combining \eqref{lem:decoupling:eq:normMM*-I},\eqref{lem:decoupling:eq:ExpeTprod}, and  \eqref{lem:decoupling:eq:normMx},  and employing Jensen's inequality, we obtain 
\begin{align}
\Expe_{\Mmat}\|\Mmat\Mmat^*-\Imat\|_2 & \leq \Expe_{\Mmat}\max_{i \in [m]}|\|\mvec_i\|_2^2-1|  + 4 \Expe_{\Mmat}\Expe_T\sup_{\substack{\xvec\in\CC^m\\ \|\xvec\|_2 = 1}}\bigg|\sum_{i \in T}\sum_{k \in \overline{T}} x_i \overline{x}_k \langle \mvec_i,\mvec_k\rangle\bigg|.
\end{align}
Bounding from above the expectation with respect to $T$ with the maximum over $T$ concludes the argument.

\end{proof}

We need two further technical results  before proceeding with the proof of Theorem~\ref{thm:svheavytailedcols}. Both results are proved in \cite{Vershynin2012} for random matrices with real entries, but the reader can check that they hold as well in the complex case. The first result is about the expected singular values of random matrices with heavy-tailed rows. This result corresponds to \cite[Theorem 5.45]{Vershynin2012}.
\begin{lem}[Expected singular values of a random matrix with isotropic independent rows]
\label{lem:ExpSVHTrows}
Let $\Amat \in \CC^{m \times N}$, with $m \leq N$, be a random matrix whose rescaled rows $\sqrt{m}\avec_1,\ldots,\sqrt{m}\avec_m$ are independent isotropic random vectors in $\CC^N$.  Then, there exists a universal constant $C>0$ such that
\begin{equation}
\Expe\max_{k\in [N]}\left|\sigma_k(\Amat)-1\right| \leq C \bigg(\Expe\Big[\max_{i \in [m]}\|\avec_i\|_2^2  \Big] \ln(m)\bigg)^{1/2}.
\end{equation}
\end{lem}

The second technical result provides an upper bound in expectation for the singular values of a random matrix with independent (but not necessarily isotropic) heavy-tailed rows. 
\begin{lem}[Expected norm of a random matrix with independent rows.]
\label{lem:svnonisorows}
Let $\Mmat \in \CC^{n \times k}$ be a random matrix whose rows $\mvec_1,\ldots,\mvec_n$ are independent random vectors in $\CC^k$ with second moment matrices 
$$
\Kmat_i = \Expe[\mvec_i \mvec_i^*], \quad \forall i \in[n].
$$ Then, there exists a universal constant $C>0$ such that
\begin{equation}
\label{lem:svnonisorows:eq:claim}
\Expe\|\Mmat\|_2
\leq \bigg\|\sum_{i = 1}^n\Kmat_i\bigg\|^{1/2}_2
+ C \bigg(\Expe\Big[\max_{i \in[n]}\|\mvec_i\|_2^2\Big] \ln(\min\{n,k\})\bigg)^{1/2}.
\end{equation}
\end{lem}
\begin{proof}
We can bound $(\Expe\|\Mmat\|^2_2)^{1/2}$ by the right-hand side of \eqref{lem:svnonisorows:eq:claim} using \cite[Theorem 5.48]{Vershynin2012} and \cite[Remark  5.49]{Vershynin2012}. Then, we observe that $\Expe\|\Mmat\|_2 \leq(\Expe\|\Mmat\|^2_2)^{1/2}$ thanks to Jensen's inequality.
\end{proof}

We are now in a position to prove Theorem~\ref{thm:svheavytailedcols}. 

\subsection{Proof of Theorem~\ref{thm:svheavytailedcols}}

We start by applying the matrix decoupling Lemma~\ref{lem:decoupling} with $\Mmat = \sqrt{\frac{m}{N}}\Amat$. Recalling the definition \eqref{eq:defdistortion} of the distortion parameter $\xi$, we obtain
\begin{equation}
\label{thm:svheavytailedcols:eq:Ebound}
\Expe\|\tfrac{m}{N} \Amat\Amat^* - \Imat\|_2
\leq \xi + \frac{4}{N} \max_{T \subseteq[m]}\Expe\|m \Amat_T \Amat_{\overline{T}}^* \|_2
= \xi + \frac{4}{N}\max_{T \subseteq[m]} \Expe\|\Gmat_T\|_2,
\end{equation}
where $\Gmat_T \in \CC^{|T| \times |\overline{T}|}$ denotes the decoupled Gram matrix 
$$
\Gmat_T = m \Amat_T \Amat_{\overline{T}}^*,
$$
whose entries contain the inner products between the rows of $\sqrt{m}\Amat$ relative to $T$ against the rows of $\sqrt{m}\Amat$ relative to its complement set $\overline{T} = [m]\setminus T$. Namely,\footnote{We should write $j \in [|T|]$ and $k \in [|\overline{T}|]$ instead of $j \in T$ and $k \in \overline{T}$, respectively. We are sure that the reader will forgive this small abuse of notation, committed for the sake of readability.} 
\begin{equation}
(G_T)_{ik} = \langle\sqrt{m} \avec_i, \sqrt{m} \avec_k\rangle, \quad \forall i \in T, \; \forall k \in \overline{T}.
\end{equation}
The next goal is to provide an estimate for $\Expe\|\Gmat_T\|_2$, which is uniform in $T$.

Given a set $S \subseteq [m]$, we denote by $\Expe_{\Amat_S}$ the  expectation with respect to $\Amat_S$. In particular, notice that $\Expe = \Expe_{\Amat_{\overline{T}}}\Expe_{\Amat_T}$. Now, we first estimate $\Expe_{\Amat_{T}}\|\Gmat_T\|_2$. Indeed, fixing $\Amat_{\overline{T}}$ we see that the matrix $\Gmat_T$ has independent (but not necessarily isotropic) rows. This allows for an application of Lemma~\ref{lem:svnonisorows}, which yields
\begin{align}
\Expe_{\Amat_T}\|\Gmat_T\|_2 
 & \leq \bigg\|\sum_{i\in T}\Kmat_{T,i}\bigg\|_2^{1/2}
+ C \bigg(\Expe_{\Amat_{T}}\Big[\max_{i \in T}\|\gvec_{T,i}\|_2^2\Big] \ln(\min\{|T|,|\overline{T}|\})\bigg)^{1/2},\\
& \leq \sqrt{m} \cdot \max_{i \in T}\|\Kmat_{T,i}\|^{1/2}_2
+ C \bigg(\Expe_{\Amat_{T}}\Big[\max_{i \in T}\|\gvec_{T,i}\|_2^2\Big] \ln(m)\bigg)^{1/2},
\label{thm:svheavytailedcols:eq:EATNormGT}
\end{align}
where $\gvec_{T,i}$ is the $i^{th}$ row of $\Gmat_T$,  $C>0$ is a universal constant, and
\begin{equation}
\Kmat_{T,i}= \Expe_{\Amat_T}[(\gvec_{T,i})(\gvec_{T,i})^*], \quad \forall i \in T,
\end{equation}
is the second moment matrix of the random vector $\gvec_{T,i}$. 

In order to estimate the right-hand side of \eqref{thm:svheavytailedcols:eq:EATNormGT}, we first provide an upper bound to $\|\Kmat_{T,i}\|$, uniform in $i$. Notice that, due to the diagonalizability of $\Kmat_{T,i}$,
\begin{equation}
\|\Kmat_{T,i}\|_2
= \sup_{\substack{\xvec \in \CC^{|\overline{T}|}\\\|\xvec\|_2 = 1}}|\xvec^*\Expe_{\Amat_T}[(\gvec_{T,i})(\gvec_{T,i})^*]\xvec|
= \sup_{\substack{\xvec \in \CC^{|\overline{T}|}\\\|\xvec\|_2 = 1}} \Expe_{\Amat_T} |\langle\gvec_{T,i},\xvec\rangle|^2.
\label{thm:svheavytailedcols:eq:boundNormKTi}
\end{equation}
For every $\xvec \in \CC^{|T|}$ with $\|\xvec\|_2 = 1$, using the isotropy of the vectors $\sqrt{m} \avec_i$ and Lemma~\ref{lem:isotropy}, we see that
\begin{align}
\Expe_{\Amat_T}|\langle\gvec_{T,i},\xvec\rangle|^2
& = \Expe_{\Amat_T}\bigg|\sum_{k \in \overline{T}}\langle\sqrt{m}\avec_i,\sqrt{m}\avec_k\rangle \overline{x}_k\bigg|^2 
=  \Expe_{\Amat_T}\Big|\Big\langle\sqrt{m}\avec_i,\sum_{k \in \overline{T}}\sqrt{m}\avec_k x_k\Big\rangle\Big|^2 \\
& = \bigg\|\sum_{k \in \overline{T}}\sqrt{m}\avec_k x_k\bigg\|_2^2
= m \|\xvec^\intercal\Amat_{\overline{T}}\|_2^2
\leq m \|\Amat_{\overline{T}}\|^2_2, 
\label{thm:svheavytailedcols:eq:boundEgx}
\end{align}
Combining  \eqref{thm:svheavytailedcols:eq:boundNormKTi} and \eqref{thm:svheavytailedcols:eq:boundEgx}, it follows that
\begin{equation}
\max_{i \in T}\|\Kmat_{T,i}\|_2\leq m \|\Amat_{\overline{T}}\|^2_2.
\label{thm:svheavytailedcols:eq:maxNormKTi}
\end{equation}

In order to deal with the second addendum in the right-hand side of \eqref{thm:svheavytailedcols:eq:EATNormGT}, we  provide an upper estimate to $\Expe_{\Amat_T}\max_{i \in T}\|\gvec_{T,i}\|_2^2$. We observe that
\begin{equation}
\max_{i \in T}\|\gvec_{T,i}\|_2^2 
= \max_{i \in T}\sum_{k \in \overline{T}} |\langle \sqrt{m} \avec_i, \sqrt{m} \avec_k \rangle|^2
\leq m^2 \max_{i \in [m]}\sum_{k \in [m]\setminus\{i\}} |\langle \avec_i, \avec_k \rangle|^2,
\end{equation}
which, in turn, implies 
\begin{equation}
\Expe_{\Amat_T}\Big[\max_{i \in T}\|\gvec_{T,i}\|_2^2 \Big]
\leq  N^2 \Expe_{\Amat_T} [Z],
\label{thm:svheavytailedcols:eq:EATmaxGTi}
\end{equation}
where 
\begin{equation}
Z := \bigg(\frac{m}{N}\bigg)^2 \max_{i \in [m]}\sum_{k \in [m]\setminus\{i\}} |\langle \avec_i, \avec_k \rangle|^2.
\end{equation}
Notice that the random variable $Z$ is defined such that its expected value is the cross coherence parameter defined in \eqref{eq:defcrosscoherence}, namely $\Expe[Z] = \mu$.

Plugging the inequalities \eqref{thm:svheavytailedcols:eq:maxNormKTi} and \eqref{thm:svheavytailedcols:eq:EATmaxGTi} into \eqref{thm:svheavytailedcols:eq:EATNormGT}, we see that
\begin{equation}
\Expe_{\Amat_T}\|\Gmat_T\|_2 
\leq m  \|\Amat_{\overline{T}}\|_2
+ C  N  (\Expe_{\Amat_{T}}[Z]\ln(m))^{1/2},
\label{thm:svheavytailedcols:eq:EATNormGTNew}
\end{equation}

Finally, in order to estimate the quantity $\Expe\|\Gmat_T\|_2$, we take the expectation of both sides of \eqref{thm:svheavytailedcols:eq:EATNormGTNew} with respect to $\Amat_{\overline{T}}$. Considering  the expectation of first addendum in the right-hand side of  \eqref{thm:svheavytailedcols:eq:EATNormGTNew}, we observe that
\begin{equation}
\Expe_{\Amat_{\overline{T}}}\|\Amat_{\overline{T}}\|_2
= \Expe \|\Amat_{\overline{T}}\|_2
\leq \Expe \|\Amat\|_2.
\label{thm:svheavytailedcols:eq:EATleqEA}
\end{equation}
Moreover, noting that the rows $\avec_i$ of $\Amat$ satisfy
\begin{equation}
\label{thm:svheavytailedcols:eq:Emaxnormrows}
\Expe \max_{i \in [m]} \left\|\avec_i\right\|_2^2 
= \frac{N}{m} \left(1 + \Expe \max_{i \in [m]} \left\|\sqrt{\tfrac{m}{N}}\avec_i\right\|_2^2-1\right)
\leq \frac{N}{m}(1 + \xi),
\end{equation}
Lemma~\ref{lem:ExpSVHTrows} yields 
\begin{align}
\label{thm:svheavytailedcols:eq:EAbound}
\Expe\|\Amat\|_2 
 =\Expe[\sigma_{\max}(\Amat)]
 \leq 1 + C' \sqrt{\frac{(1+\xi) N \ln(m)}{m}}
\leq 2C' \sqrt{\frac{(1+\xi) N\ln(m)}{m}},
\end{align}
for a suitable universal constant $C'>0$.

About the expectation of the second addendum in the right-hand side of \eqref{thm:svheavytailedcols:eq:EATNormGTNew}, Jensen's inequality implies
\begin{equation}
\Expe_{\Amat_{\overline{T}}}(\Expe_{\Amat_{T}}[Z]\ln(m))^{1/2}
\leq (\Expe_{\Amat_{\overline{T}}}\Expe_{\Amat_{T}}[Z]\ln(m))^{1/2}
=(\Expe[Z]\ln(m))^{1/2} 
= \sqrt{\mu \ln(m)}.
\label{thm:svheavytailedcols:eq:SecondTermBound}
\end{equation}
Therefore, considering the expectation of \eqref{thm:svheavytailedcols:eq:EATNormGTNew} with respect to $\Amat_{\overline{T}}$ and using \eqref{thm:svheavytailedcols:eq:EATleqEA}, \eqref{thm:svheavytailedcols:eq:EAbound}, and \eqref{thm:svheavytailedcols:eq:SecondTermBound}, we obtain
\begin{align}
\Expe \|\Gmat_T\|_2 
& \leq 2C' \sqrt{(1+\xi) mN \ln(m)} + CN \sqrt{\mu \ln(m)}
 \leq C'' N \sqrt{(1+\xi) \mu \ln(m)},
\end{align}
where we have used the fact that $\mu \geq (m-1)/N \geq m/(2N)$ (thanks to Lemma~\ref{lem:isotropy} and to the isotropy of the vectors $\sqrt{m}\avec_i$) and $C'' = 2\sqrt{2}C' + C$. 

Recalling \eqref{thm:svheavytailedcols:eq:Ebound}, this shows that
\begin{equation}
\Expe\|\tfrac{m}{N} \Amat\Amat^* - \Imat\|_2
\leq \xi + C'' \sqrt{(1+\xi)\mu \ln(m)},
\end{equation}
where the constant $C''>0$ is universal. Finally, we observe that
\begin{align}
\|\tfrac{m}{N}\Amat\Amat^* - \Imat\|_2 
= \max_{k \in [m]} |\sigma_k(\sqrt{\tfrac mN}\Amat^*)^2-1 |
\geq \max_{k \in [m]} |\sigma_k(\sqrt{\tfrac mN}\Amat^*)-1 |.
\end{align}
This completes the proof.

\section{Proof of Proposition~\ref{prop:distCheby}}
\label{sec:Chebyproof}

This appendix is devoted to the proof of Proposition~\ref{prop:distCheby}, which provides an upper bound to the distortion parameter $\xi$ defined in \eqref{eq:defdistortion} in the case of the Chebyshev polynomials defined in \eqref{eq:defCheby1}-\eqref{eq:defCheby2}. The crucial tool employed here is the \emph{normalized Christoffel function}.
\begin{defn}[Normalized Christoffel function]
We define the \emph{normalized Christoffel function} associated with a  BOS $\Phi=\{\phi_j\}_{j=1}^{N}$ as
\begin{equation}
\mathcal{C}_N(t) := \frac{1}{N}\sum_{j = 1}^N |\phi_j(t)|^2.
\end{equation}
\end{defn}
In the case of Chebyshev polynomials, defining $x = \arccos(t)$, we have 
\begin{align}
\mathcal{C}_N(t) 
& = \frac{1}{N}\sum_{j = 1}^{N} (\phi_j(t))^2
  = \frac{1}{N} + \frac{2}{N} \sum_{j = 1}^{N-1} \cos^2(j x)
  = 1 + \frac{1}{N}\sum_{j = 1}^{N-1} \cos(2 j x)\\
& = \frac{N-1}{N} + \frac{\sin((2N-1)x)}{N\sin(x)}.
\label{eq:defChristoCheby}
\end{align}

We prove a technical result that  reveals a tight link between the asymptotic behavior of the normalized Christoffel function and the decay of the distortion parameter $\xi$ for any BOS associated with the one-dimensional Chebyshev measure. 

\begin{lem}[Distortion  bound based on the Christoffel function]
\label{lem:ChristoffelBOSCheby}
Consider a BOS $\Phi=\{\phi_j\}_{j=1}^{N}$ associated with the Chebyshev measure \eqref{eq:defChebymeas}. Fix  $0<\varepsilon<1$ and $\tau >0$ and assume that there exists a positive integer $\overline{N}_{\varepsilon,\tau}$ such that, for every $N \geq \overline{N}_{\varepsilon,\tau}$, the following holds
\begin{equation}
\label{eq:lem:ChristoffelBOSCheby:hypo}
\|\mathcal{C}_N - 1\|_{L^\infty(-1+\varepsilon, 1-\varepsilon)} \leq \tau.
\end{equation}
Then, for every $N \geq \overline{N}_{\varepsilon,\tau}$, the distortion parameter $\xi$ defined in \eqref{eq:defdistortion}
satisfies
\begin{equation}
\label{eq:lem:ChristoffelBOSCheby:thesis}
\xi \leq \tau + 2(K^2+1)m \sqrt{\varepsilon/\pi}.
\end{equation}
\end{lem}
\begin{proof}
For any random matrix $\Amat\in\CC^{m \times N}$ with rows $\avec_i$ built as in \eqref{eq:BOSmatrix}, considering the random variable 
\begin{equation}
X(\Amat) = \max_{i\in[m]}\left|\frac{m}{N}\|\avec_i\|_2^2-1\right|,
\end{equation}
we have $\xi = \Expe[X(\Amat)]$. Since $m \|\avec_i\|_2^2 \leq K^2 N$, we have  $|X(\Amat)| \leq K^2+1$. This, in turn, implies
\begin{equation}
\label{eq:lem:ChristoffelBOSCheby:ineq1}
\xi
\leq \tau +  (K^2+1) \Prob\{X > \tau\}, \quad \forall\tau >0. 
\end{equation}

Now, in order to estimate $\Prob\{X(\Amat) > \tau\}$, we notice that $(m/N)\|\avec_i\|_2^2 = \mathcal{C}_N(t_i)$ and, exploiting the independence of the sampling points $t_i$, we resort to a union bound
\begin{align}
\Prob\{X(\Amat) > \tau\}
& = \Prob\bigg\{\max_{i \in [m]} |\mathcal{C}_N(t_i) - 1| > \tau\bigg\}
 = \Prob\bigg(\bigcup_{i \in [m]} \bigg\{|\mathcal{C}_N(t_i) - 1| > \tau\bigg\}\bigg)\\
& \leq m \Prob\bigg\{|\mathcal{C}_N(t_1) - 1| > \tau\bigg\}
 \leq m \underbrace{\Prob \{t_1 \in[-1,1]\setminus(-1+\varepsilon,1-\varepsilon)\}}_{=:\lambda_{\varepsilon}},
 \label{eq:lem:ChristoffelBOSCheby:ineq2}
\end{align}
where in the last inequality hinges on \eqref{eq:lem:ChristoffelBOSCheby:hypo}.

The next step is to estimate the probability $\lambda_\varepsilon$ that the distance of a sample $t_1$ from the boundary is less than $\varepsilon$. By direct computation
\begin{align}
\lambda_{\varepsilon} & 
= 1 - \int_{-1+\varepsilon}^{1-\varepsilon} \de \nu(t) 
= 1 - \pi^{-1}\int_{-1+\varepsilon}^{1-\varepsilon} (1-t^2)^{-1/2}\de t \\
& = 1 - \pi^{-1}(\arcsin(1-\varepsilon)- \arcsin(-1+\varepsilon))\\
& \leq 1-\pi^{-1}(\pi/2 - \sqrt{\pi \varepsilon} + \pi/2 - \sqrt{\pi\varepsilon}) = 2\sqrt{\varepsilon/\pi}. \label{eq:lem:ChristoffelBOSCheby:ineq3}
\end{align}
The last inequality hinges on the estimates 
\begin{align}
\arcsin(x) & \geq \phantom{-}\pi/2 - \sqrt{\pi(1-x)}, \quad \forall x \in [0,1],\\
\arcsin(x) & \leq -\pi/2 + \sqrt{\pi(1+x)}, \quad \forall x \in[-1,0],
\end{align} 
which, in turn, can be deduced from 
\begin{align}
\sin(x) & \geq -1 + (x+\pi/2)^2 / \pi, \quad \forall x\in[-\pi/2,0]\\
\sin(x) & \leq \phantom{-}1 - (x-\pi/2)^2 / \pi, \quad \forall x\in[0,\pi/2].
\end{align}
Combining \eqref{eq:lem:ChristoffelBOSCheby:ineq1},  \eqref{eq:lem:ChristoffelBOSCheby:ineq2}, and \eqref{eq:lem:ChristoffelBOSCheby:ineq3} completes the proof.
\end{proof}

An immediate consequence of Lemma~\ref{lem:ChristoffelBOSCheby} is that when the normalized Christoffel function $\mathcal{C}_N \to 1$ uniformly on every interval of the from $[-1+\varepsilon,1-\varepsilon]$ for $N\to \infty$, then the distortion $\xi \to 0$  for $N \to \infty$ and $m$ fixed.
Recalling \eqref{eq:defChristoCheby}, the reader can verify that this is the case for the Chebyshev polynomials. Moreover, we are able to track the dependency on $m$ and $N$ in the decay of $\xi$ to $0$ as $N \to 0$. We are now in a position to prove Proposition~\ref{prop:distCheby}.

\begin{proof}[Proof of Proposition~\ref{prop:distCheby}]
The idea is to employ Lemma~\ref{lem:ChristoffelBOSCheby}, giving an explicit value for $\overline{N}_{\varepsilon,\tau}$ that guarantees the validity of relation \eqref{eq:lem:ChristoffelBOSCheby:hypo}.

First, recalling  \eqref{eq:defChristoCheby}, notice that
\begin{equation}
|1-\mathcal{C}_N(t)| 
= \bigg|1-\frac{N-1}{N} - \frac{\sin((2N-1)x)}{N\sin(x)}\bigg|
\leq  \frac{1}{N} + \frac{1}{N} \bigg|\frac{\sin((2N-1)x)}{\sin(x)}\bigg|.
\end{equation}
Now, fixed $0<\varepsilon <1$, for every $t \in [-1+\varepsilon,1-\varepsilon]$,  we have
\begin{equation}
|\sin(x)|\geq \sin(\arccos(1-\varepsilon)) =\sqrt{1-(1-\varepsilon)^2}=  \sqrt{\varepsilon(2-\varepsilon)} \geq \sqrt{\varepsilon}.
\end{equation}
Therefore, we obtain 
\begin{equation}
\|1-\mathcal{C}_N(t)\|_{L^{\infty}(-1+\varepsilon, 1-\varepsilon)}
\leq \frac{1}{N} + \frac{1}{\sqrt{\varepsilon}N} \leq \frac{2}{\sqrt{\varepsilon} N},
\end{equation}
and, consequently, \eqref{eq:lem:ChristoffelBOSCheby:hypo} holds with $\overline{N}_{\varepsilon,\tau} = \lceil2/(\sqrt{\varepsilon}\tau)\rceil$. Now, if we choose 
$$
\tau =\tau(\varepsilon) = (K^2+1) m\sqrt{\varepsilon/\pi} = 3 m\sqrt{\varepsilon/\pi},
$$ 
we have  that, for every $0 < \varepsilon < 1$, the following implication holds
\begin{equation}
N \geq \frac{2\sqrt{\pi}}{3m\varepsilon}
\Longrightarrow
\xi \leq 9 m \sqrt{\varepsilon/\pi}. 
\end{equation}
 In particular, for every $1< m \leq N$, we can choose  $0 < \varepsilon < 1$ such that $N = \lceil2\sqrt{\pi}/(3m\varepsilon)\rceil$. Then, we have $\frac{2\sqrt{\pi}}{3m\varepsilon} \leq N \leq \frac{2\sqrt{\pi}}{3m\varepsilon} +1$. Thus, $\varepsilon \leq \frac{2\sqrt{\pi}}{3 m (N-1)}$, and, finally,
\begin{equation}
\xi \leq \frac{3 \sqrt{6}}{\pi^{1/4}} \sqrt{\frac{m}{N-1}} \leq \frac{9 \sqrt{6}}{2\pi^{1/4}} \sqrt{\frac{m}{N}},
\end{equation}
which is the desired conclusion.
\end{proof}

\bibliographystyle{plain}
\bibliography{library}

\begin{thebibliography}{10}

\bibitem{Adcock2017c}
B.~Adcock.
\newblock {Infinite-dimensional compressed sensing and function interpolation}.
\newblock {\em Found. Comput. Math.}, to appear, 2017.

\bibitem{Adcock2017b}
B.~Adcock.
\newblock Infinite-dimensional $\ell^1$ minimization and function approximation
  from pointwise data.
\newblock {\em Constr. Approx.}, 45(3):345--390, 2017.

\bibitem{Adcock2017}
B.~Adcock, S.~Brugiapaglia, and C.~G. Webster.
\newblock {Compressed sensing approaches for polynomial approximation of
  high-dimensional functions}.
\newblock {\em arXiv:1703.06987}, 2017.

\bibitem{Blumensath2009}
T.~Blumensath and M.~E. Davies.
\newblock {Iterative hard thresholding for compressed sensing}.
\newblock {\em Appl. Comput. Harmon. Anal.}, 27(3):265--274, 2009.

\bibitem{Bouchot2015}
J.-L. Bouchot, B.~Bykowski, H.~Rauhut, and C.~Schwab.
\newblock {Compressed sensing Petrov-Galerkin approximations for parametric
  PDEs}.
\newblock In {\em Sampling Theory and Applications (SampTA), 2015 International
  Conference on}, pages 528--532. IEEE, 2015.

\bibitem{Boufounos2008}
P.~T. Boufounos and R.~G. Baraniuk.
\newblock 1-bit compressive sensing.
\newblock In {\em Information Sciences and Systems, 2008. CISS 2008. 42nd
  Annual Conference on}, pages 16--21. IEEE, 2008.

\bibitem{Bourgain2014}
J.~Bourgain.
\newblock {An Improved Estimate in the Restricted Isometry Problem}.
\newblock In {\em Geom. Asp. Funct. Anal.}, pages 65--70. Springer
  International Publishing, 2014.

\bibitem{PhDThesis}
S.~Brugiapaglia.
\newblock {\em {COmpRessed SolvING: sparse approximation of PDEs based on
  compressed sensing}}.
\newblock PhD thesis, MOX - Politecnico di Milano, 2016.

\bibitem{Brugiapaglia2017recovery}
S.~Brugiapaglia, B.~Adcock, and R.~K. Archibald.
\newblock Recovery guarantees for compressed sensing with unknown errors.
\newblock In {\em Sampling Theory and Applications (SampTA), 2017 International
  Conference on}. IEEE, 2017.

\bibitem{Brugiapaglia2015}
S.~Brugiapaglia, S.~Micheletti, and S.~Perotto.
\newblock {Compressed solving: A numerical approximation technique for elliptic
  PDEs based on Compressed Sensing}.
\newblock {\em Comput. Math. with Appl.}, 70(6), 2015.

\bibitem{Brugiapaglia2016}
S.~Brugiapaglia, F.~Nobile, S.~Micheletti, and S.~Perotto.
\newblock {A theoretical study of COmpRessed SolvING for
  advection-diffusion-reaction problems}.
\newblock {\em Math. Comp.}, to appear, 2017.

\bibitem{Cai2014}
T.~T. Cai and A.~Zhang.
\newblock {Sparse representation of a polytope and recovery of sparse signals
  and low-rank matrices}.
\newblock {\em IEEE Trans. Inform. Theory}, 60(1):122--132, 2014.

\bibitem{Candes2011}
E.~J. Cand{\`{e}}s and Y.~Plan.
\newblock {A probabilistic and RIPless theory of compressed sensing}.
\newblock {\em IEEE Trans. Inform. Theory}, 57(11):7235--7254, 2011.

\bibitem{Candes2006}
E.~J. Cand{\`{e}}s, J.~Romberg, and T.~Tao.
\newblock {Robust uncertainty principles: Exact signal reconstruction from
  highly incomplete frequency information}.
\newblock {\em IEEE Trans. Inform. Theory}, 52(2):489--509, 2006.

\bibitem{Candes2006stable}
E.~J. Cand\`{e}s, J.~K. Romberg, and T.~Tao.
\newblock Stable signal recovery from incomplete and inaccurate measurements.
\newblock {\em Comm. Pure Appl. Math.}, 59(8):1207--1223, 2006.

\bibitem{Chkifa2016}
A.~Chkifa, N.~Dexter, H.~Tran, and C.~G. Webster.
\newblock {Polynomial approximation via compressed sensing of high-dimensional
  functions on lower sets}.
\newblock {\em Math. Comp.}, to appear, 2017.

\bibitem{Cohen2008}
A.~Cohen, W.~Dahmen, and R.~DeVore.
\newblock {Compressed sensing and best $k$-term approximation}.
\newblock {\em J. Amer. Math. Soc}, 22:211--231, 2008.

\bibitem{Davenport2014}
M.~A. Davenport, Y.~Plan, E.~Berg, and M.~Wootters.
\newblock {1-Bit Matrix Completion}.
\newblock {\em Inf. Inference}, 3(3):189--223, 2014.

\bibitem{DeVore2009}
R.~DeVore, G.~Petrova, and P.~Wojtaszczyk.
\newblock {Instance-optimality in probability with an $\ell_1$-minimization
  decoder}.
\newblock {\em Appl. Comput. Harmon. Anal.}, 27(3):275--288, 2009.

\bibitem{Donoho2006}
D.~L. Donoho.
\newblock {Compressed Sensing}.
\newblock {\em IEEE Trans. Inform. Theory}, 52(4):1289--1306, 2006.

\bibitem{Donoho1992}
D.~L. Donoho and B.~F. Logan.
\newblock {Signal Recovery and the Large Sieve}.
\newblock {\em SIAM J. Appl. Math.}, 52(2):577--591, 1992.

\bibitem{Doostan2011}
A.~Doostan and H.~Owhadi.
\newblock {A non-adapted sparse approximation of {PDEs} with stochastic
  inputs}.
\newblock {\em J. Comput. Phys.}, 230(8):3015--3034, 2011.

\bibitem{Foucart2014}
S.~Foucart.
\newblock {Stability and robustness of $\ell_1$-minimizations with Weibull
  matrices and redundant dictionaries}.
\newblock {\em Linear Algebra Appl.}, 441:4--21, 2014.

\bibitem{Foucart2013}
S.~Foucart and H.~Rauhut.
\newblock {\em {A Mathematical Introduction to Compressive Sensing}}.
\newblock Appl. Numer. Harmon. Anal. Springer Science+Business Media, New York,
  2013.

\bibitem{Haviv2017}
I.~Haviv and O.~Regev.
\newblock {The Restricted Isometry Property of Subsampled Fourier Matrices}.
\newblock In B.~Klartag and E.~Milman, editors, {\em Geometric Aspects of
  Functional Analysis: Israel Seminar (GAFA) 2014--2016}, pages 163--179.
  Springer International Publishing, Cham, 2017.

\bibitem{Herman2010}
M.~A. Herman and T.~Strohmer.
\newblock {General Deviants: An Analysis of Perturbations in Compressed
  Sensing}.
\newblock {\em IEEE J. Sel. Top. Signal Process.}, 4(2):342--349, jul 2010.

\bibitem{Logan1965}
B.~F. Logan.
\newblock {\em {Properties of high-pass signals}}.
\newblock PhD thesis, Columbia University, 1965.

\bibitem{Needell2009}
D~Needell and J.~A. Tropp.
\newblock {CoSaMP: Iterative signal recovery from incomplete and inaccurate
  samples}.
\newblock {\em Appl. Comput. Harmon. Anal.}, 26(3):301--321, 2009.

\bibitem{Negahban2012}
S.~N. Negahban, R.~Pradeep, B.~Yu, and M.~J. Wainwright.
\newblock {A Unified Framework for High-Dimensional Analysis of M-Estimators
  with Decomposable Regularizers}.
\newblock {\em Stat. Sin.}, 27(4):538--557, 2012.

\bibitem{Plan2016}
Y.~Plan and R.~Vershynin.
\newblock {The generalized Lasso with non-linear observations}.
\newblock {\em IEEE Trans. Inform. Theory}, 62(3):1528--1537, 2016.

\bibitem{Rauhut2010}
H.~Rauhut.
\newblock {Compressive sensing and structured random matrices}.
\newblock In M~Fornasier, editor, {\em Theor. Found. Numer. Methods Sparse
  Recover.}, volume~9 of {\em Radon Series on Computational and Applied
  Mathematics}, pages 1--92. deGruyter, 2010.

\bibitem{Rauhut2017}
H.~Rauhut and C.~Schwab.
\newblock {Compressive sensing Petrov-Galerkin approximation of
  high-dimensional parametric operator equations}.
\newblock {\em Math. Comp.}, 86(304):661--700, 2017.

\bibitem{Rauhut2016}
H.~Rauhut and R.~Ward.
\newblock {Interpolation via weighted $\ell_1$ minimization}.
\newblock {\em Appl. Comput. Harmon. Anal.}, 40(2):321--351, 2016.

\bibitem{Unser2016}
M.~Unser, J.~Fageot, and H.~Gupta.
\newblock {Representer Theorems for Sparsity-Promoting $\ell _1$
  Regularization}.
\newblock {\em IEEE Trans. Inf. Theory}, 62(9):5167--5180, 2016.

\bibitem{spgl1}
E~{Van Den Berg} and M.~P. Friedlander.
\newblock {SPGL1: A solver for large-scale sparse reconstruction}, 2007.

\bibitem{Vershynin2012}
R.~Vershynin.
\newblock {Introduction to the non-asymptotic analysis of random matrices}.
\newblock In Y.~Eldar and G.~Kutyniok, editors, {\em Compressed Sensing: Theory
  and Applications}. Cambridge University Press, Cambridge, 2012.

\bibitem{Wojtaszczyk2010}
P.~Wojtaszczyk.
\newblock {Stability and instance optimality for Gaussian measurements in
  compressed sensing}.
\newblock {\em Found. Comput. Math.}, 10(1):1--13, 2010.

\bibitem{Yang2013}
X.~Yang and G.~E. Karniadakis.
\newblock {Reweighted $\ell^1$ minimization method for stochastic elliptic
  differential equations}.
\newblock {\em J. Comput. Phys.}, 248:87--108, 2013.

\bibitem{Zhang2011}
T.~A. Zhang.
\newblock {Sparse recovery with orthogonal matching pursuit under RIP}.
\newblock {\em IEEE Trans. Inform. Theory}, 57(9):6215--6221, 2011.

\end{thebibliography}

\end{document}